\documentclass[aps,prx,twocolumn,amsmath,longbibliography]{revtex4-2}
\usepackage{amsmath}
\usepackage{amsthm}
\usepackage{amssymb}
\usepackage{graphicx}
\usepackage{xcolor}
\usepackage{braket}
\usepackage{algorithm}
\usepackage[makeroom]{cancel}
\usepackage{algpseudocode}
\let\savecorresponds\corresponds
\let\corresponds\relax

\usepackage{changes}
\usepackage{amsthm}
\usepackage{svg}
\usepackage{changes}
\let\corresponds\savecorresponds

\DeclareFontFamily{U}{mathx}{\hyphenchar\font45}
\DeclareFontShape{U}{mathx}{m}{n}{
      <5> <6> <7> <8> <9> <10>
      <10.95> <12> <14.4> <17.28> <20.74> <24.88>
      mathx10
      }{}
\DeclareSymbolFont{mathx}{U}{mathx}{m}{n}
\DeclareMathSymbol{\bigtimes}{1}{mathx}{"91}

\newtheorem{theorem}{Theorem}[section]
\newtheorem{lemma}[theorem]{Lemma}

\definecolor{THc}{rgb}{0.5,0.6,0.2}

\definecolor{myurlcolor}{rgb}{0,0,0.7}
\definecolor{myrefcolor}{rgb}{0.1,0,0.9}

\usepackage[
	breaklinks,
	pdftex,
	colorlinks=true, 
	linkcolor=myrefcolor,
	citecolor=myurlcolor,
	urlcolor=myurlcolor
]{hyperref}

\newcommand{\argmin}{{\rm argmin}}
\newcommand{\nodes}{{\cal V}}
\newcommand{\edges}{{\cal E}}
\newcommand{\cnums}{{\mathbb C}}

\newcommand{\tr}{{\rm Tr}}
\newcommand{\Z}[1]{{\mathbb Z}_{#1}}
\newcommand{\bfi}{{\mathbf i}}
\newcommand{\bfj}{{\mathbf j}}

\begin{document}

\title{Large-scale quantum annealing simulation with tensor networks and belief propagation
}
\author{Ilia A. Luchnikov}
\author{Egor S. Tiunov}
\author{Tobias Haug}
\author{Leandro Aolita}
\affiliation{Quantum Research Center, Technology Innovation Institute, Abu Dhabi, UAE}

\begin{abstract} 
Quantum annealing and quantum approximate optimization algorithms hold a great potential to speed-up optimization problems. 
This could be game-changing for a plethora of applications. 
Yet, in order to hope to beat classical solvers, quantum circuits must scale up to sizes and performances much beyond current hardware. 
In that quest, intense experimental effort has been recently devoted to optimizations on $3$-regular graphs, which are computationally hard but experimentally relatively amenable.
However, even there, the amount and quality of quantum resources required for quantum solvers to outperform classical ones is unclear. 
Here, we show that quantum annealing for $3$-regular graphs can be classically simulated even at scales of $1000$ qubits and $4.8\times 10^6$ two-qubit gates with all-to-all connectivity.
To this end, we develop a {\it graph tensor-network quantum annealer} (GTQA) able of high-precision simulations of Trotterized circuits of near-adiabatic evolutions. Based on a recently proposed belief-propagation technique for tensor canonicalization, GTQA is equipped with re-gauging and truncation primitives that keep approximation errors small in spite of the circuits generating significant amounts of entanglement.
As a result, even with a maximal bond dimension as low as $\chi=4$, GTQA produces solutions competitive with those of state-of-the-art classical solvers. 
For non-degenerate instances, the unique solution can be read out from the final reduced single-qubit states.  In contrast, for degenerate problems, such as MaxCut, we introduce an approximate measurement simulation algorithm for graph tensor-network states. This can not only sample from the corresponding outcome distribution but also evaluate its probabilities, being thus also interesting beyond the scope of annealers. 
On one hand, our findings showcase the potential of GTQA as a powerful quantum-inspired optimizer. On the other hand, they considerably raise the bar required for experimental demonstrations of quantum speed-ups in combinatorial optimizations.
\end{abstract}

\maketitle
\section{Introduction}

Quantum computers may solve hard combinatorial optimization problems in large-scale regimes where classical methods struggle. This would have a major impact on diverse areas such as logistics, finance, energy, biotechnology, and machine learning%
~\cite{abbas2023quantum}. 
One of the most promising routes is quantum annealing (QA)~\cite{kadowaki1998quantum,farhi2000quantum,hauke2020perspectives, boixo2014evidence, lanting2014entanglement, bettina2015quantum}, where an adiabatic (i.e. slowly-varying) time-evolution from a reference state to the ground state of a target Ising Hamiltonian encoding the solution is driven. 
Another celebrated approach is quantum approximate optimization algorithms (QAOAs)~\cite{farhi2014quantum, preskill2018quantum, zhou2020quantum,PhysRevLett.125.260505}, which can be seen as coarse-grained, short-depth versions of QA. There, one variationally optimizes the Hamiltonian schedule, instead of using an adiabatic evolution.
Among the native problems solved by QA and QAOA, one finds quadratic unconstrained binary optimization (QUBO) and its closely related MaxCut. These are paradigmatic NP-hard optimizations on a graph.  
A prominent subclass is that of $d$-regular graphs, where each vertex has constant connectivity $d$. 
These lend themselves better to physical implementations than higher-connectivity graphs, yet they are known to encompass hard instances with real-world applications~\cite{commander2009maximum}.
In fact, even finding approximate solutions on $3$-regular graphs is known to be NP-hard~\cite{berman1999some}. %

This has fueled a great deal of activity on quantum optimization algorithms for $3$-regular and other sparse graphs. On the experimental side, impressive proof-of-principle demonstrations have been achieved.
With superconducting-qubit circuits, QAOAs on $3$-regular graphs have been implemented for example %
for instances of 22~\cite{harrigan2021quantum} and $120$~\cite{sachdeva2024quantum} vertices. %
For trapped ions, $32$- and $130$-vertex instances have been studied respectively with circuits of more than $300$ two-qubit gates~\cite{shaydulin2023qaoawith} and via mid-circuit measurements~\cite{decross2023qubit,moses2023race}. 
Additionally, relaxations of $3$-regular graph problems were explored on superconducting qubits~\cite{dupont2024quantum} and trapped ion~\cite{ponce2023graph}. 
Moreover, with Rydberg atoms, QAOAs for few-vertex MaxCut instances~\cite{graham2022multi} have been implemented as well as QAOAs and approximate QA for %
maximal independent set problems on sparse graphs
of up to 289 vertices~\cite{doi:10.1126/science.abo6587}, remarkably. 
In turn, on the theory side, there is extensive literature on quantum solvers on $3$-regular graphs, with analytic performance guarantees~\cite{wurtz2021maxcut}, numerical performance studies~\cite{wurtz2021fixed}, considerations of experimental noise~\cite{stilck2021limitations}, and benchmarks against classical methods~\cite{liu2015quantum}.

\begin{figure*}[t!]
		\centering
		\includegraphics[width=\linewidth]{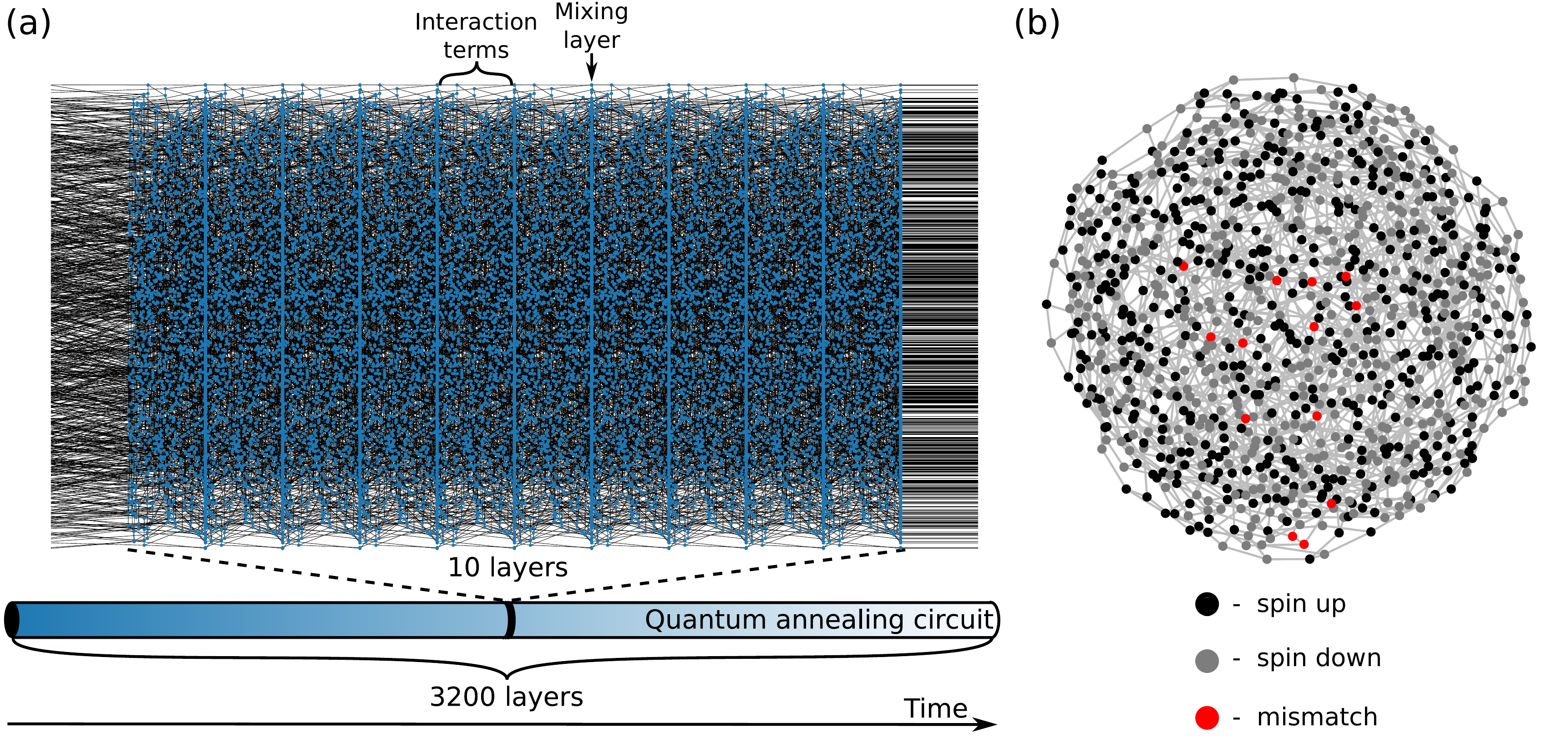}
		\caption{{\bf Schematics of the largest quantum-annealing circuit simulated and the corresponding instance solved.} 
  (a) Network representation of $10$ out of $3200$ gate layers of the entire Trotterized circuit for a random QUBO instance on the problem-input graph $G$ shown in panel (b). 
  Each edge (in black) in the network represents a %
  qubit and each node (in blue) a %
  gate. Each gate layer is in turn composed of a layer of single-qubit gates, corresponding to the mixing term of the annealing Hamiltonian, and a layer of 2-qubit entangling gates, corresponding to the interaction terms of the target Ising Hamiltonian. The  network contains in total $4.8\times 10^6$ %
  2-qubit gates, applied on $N=1000$ %
  qubits according to the connectivity in $G$.
  (b) Random 1000-vertex 3-regular graph $G$ used and the solution obtained by the
graph tensor-network quantum annealer (GTQA). Mismatches with the bit-string given by the best solution found among 40 heuristic solvers are colored in red; with the corresponding objective-function value mismatch being only of $0.04\%$ [see Fig. \ref{fig:largest_numerical_experiment}(a) for details].  
  }
    	\label{fig:circuit_tn}
\end{figure*}

However, demonstrating actual quantum speed-ups for combinatorial optimizations is challenging. First, for NP-hard problems quantum computers are expected to offer (at best) asymptotic quadratic speed-ups in the number of operations but their clock cycle per operation is orders of magnitude slower than in classical electronics. This might imply huge numbers of qubits required to actually observe concrete advantages~\cite{guerreschi2019qaoa}.
Second, classical heuristic solvers are abundant and work extremely well in practice~\cite{dunning2018what}. These include powerful physics-inspired solvers such as coherent Ising machines, simulated bifurcation machines, variational-neural annealers~\cite{mcmahon2016fully,honjo2021100,mohseni2022ising,goto2019combinatorial,goto2021high,hibat2021variational}, or tensor-networks based solvers~\cite{lami2023quantum}. 
Third, even if in principle advantageous, quantum algorithm implementations on actual quantum hardware must display better performances than their corresponding simulation on a classical computer. While classical simulation is hard in general, several impressive quantum computation experiments turned out to be easily classically simulable~\cite{tindall2023efficient,pan2022simulation}. This has triggered an arms race between improving quantum computers and developing better classical simulation methods.

Since $3$-regular graphs define highly unstructured problems, they are not amenable to standard simulation techniques. For example, matrix product states (MPSs) fail beyond one dimension~\cite{garcia2007matrix, verstraete2008matrix, schollwock2011density, vidal2003efficient, vidal2004efficient, haegeman2011time, haegeman2016unifying, kloss2018time}. For higher-dimensional tensor networks (TNs), such as projected entangled-pair states~\cite{verstraete2004renormalization, jordan2008classical, cirac2021matrix} and other sophisticated networks~\cite{gray2021hyperoptimized, gray2024hyperoptimized,pan2020contracting,hauru2018renormalization, jahromi2019universal, patra2024projected}, tensor contraction scales in general exponentially with the system size. Alternatively, approximate contraction techniques exist, but their accuracy is in general hard to control unless using intractably large tensors~\cite{nishino1996corner, orus2009simulation, corboz2010simulation, corboz2011stripes, levin2007tensor,evenbly2015tensor, evenbly2017algorithms, xie2012coarse, vanderstraeten2022variational, vieijra2021direct, zaletel2020isometric, lin2022efficient, soejima2020isometric}.
Recently, a new paradigm~\cite{alkabetz2021tensor} for TN manipulation %
based on belief propagation (BP)~\cite{mezard2009information, kschischang2001factor, pearl1988probabilistic, wainwright2003tree, yedidia2005constructing, lauritzen1988local} has been introduced, which has found various applications~\cite{tindall2023gauging, sahu2022efficient, tindall2023efficient,guo2023block,kaufmann2024blockbp}. 
These methods feature a complexity scaling linearly in the number of vertices and are exact for tree graphs. Albeit not exact for generic graphs, they can give accurate approximations for %
sufficiently structured graphs. In fact, they have proven successful~\cite{tindall2023efficient,tomislav2024fast, beguvsic2023fast, liao2023simulation, PhysRevResearch.6.013326} in simulating $127$-qubit circuits with $2,880$ two-qubit gates on IBM's heavy-hex lattice~\cite{kim2023evidence}. However, it is an open question whether these methods can accurately simulate quantum circuits at large scales on unstructured geometries; and, moreover, whether they can do it for algorithms of practical relevance.

Here, we answer these questions in the affirmative. We show that QA for random $3$-regular graphs can be classically simulated even at scales of $1000$ %
qubits and $4.8\times 10^6$ %
two-qubit gates with all-to-all connectivity. To this end, we develop a simulation toolkit for QA based on graph TNs and belief propagation, which we dub the {\it graph tensor-network quantum annealer} (GTQA). This is able of high-precision simulations of Trotterized, near-adiabatic evolutions with respect to Ising Hamiltonians on unstructured lattices with low connectivity.
To grasp the daunting scale of the TN, in Fig.~\ref{fig:circuit_tn}(a) we show $0.3\%$ of the circuit for the largest graph considered [shown in Fig. \ref{fig:circuit_tn}(b)]. 
The use of a TN geometry that matches the random graph in question, together with suitable canonicalization %
and truncation primitives, allows us to keep approximation errors low in spite of using extremely small tensor cores relative to the amount of generated entanglement.
For example, for Fig.~\ref{fig:circuit_tn}(a) maximal bond dimension $\chi=4$ is enough for GTQA to achieve solutions competitive with those of state-of-the-art solvers. In contrast, the same simulation would take an MPS a bond dimension orders of magnitude higher (lower-bounded by $580$ and upper-bounded by $4^{158}$).  

We benchmark GTQA on random QUBO and MaxCut problems of up to 1000 and 150 vertices, respectively. For QUBO, we study the transition from non-equilibirum (quench) dynamics to adiabaticity by probing the entanglement-entropy evolution over annealing schedules with different durations $T$ and fixed time step $\delta t$. That is, we use an evolution time proportional to the number of Trotter layers.
For the 1000-vertex instance, we observe adiabaticity at $T/\delta t=3600$ Trotter layers, as mentioned in Fig.~\ref{fig:circuit_tn}(a), where the final state is close to a product state encoding of the solution bit-string. We find that this solution matches the best one among 40 heuristic solvers in 987 out of the 1000 bits [see Fig.~\ref{fig:circuit_tn}(b)], giving an optimal objective function value only $0.04\%$ worse. 
Interestingly, the maximal entanglement over the evolution does not grow further with increasing $T$, indicating that even more adiabatic evolutions would be possible with the same bond dimension $\chi$.
Additionally, we estimate state errors due to truncation and non-perfect canonicalization. Estimates based on the discarded singular values show infidelities saturating at $3.2\times 10^{-4}$ and $3.2\times 10^{-2}$ for the quenched and adiabatic cases, with corresponding $\chi=32$ and $\chi=4$, respectively.
In turn, comparison with exact brute-force simulations for random instances of up to $N=26$ qubits gives average trace-distance errors between $10^{-2}$ and $10^{-3}$ for the reduced single-qubit states.
Moreover, for these small systems, the median error does not show growth with $N$ or $T$, remarkably.

MaxCut problems, in contrast, are intrinsically degenerate and typically lead to significantly higher entanglement generation than non-degenerate QUBO. For MaxCut, the final state of the QA algorithm is a superposition of computational-basis states each one encoding a valid solution.
Hence, one must measure in the computational basis 
to get a solution bit-string. 
To address this, we introduce a measurement-simulation algorithm for graph TN states, based on BP-guided sampling~\cite{mezard2009information}.
With this, we solve MaxCut on a random 3-regular graph of $N=150$ vertices. We use $T/\delta t=200$ Trotter layers and $\chi=32$.
We find that GTQA's solution matches the best one among the $40$ heuristic solvers considered up to an approximation ratio $\alpha\approx 0.99$.
Finally, apart from sampling, our measurement-simulation primitive can also estimate outcome probabilities. This makes it potentially relevant also for applications beyond the current scope, such as for instance cross-entropy benchmarking~\cite{arute2019quantum} or classical shadows~\cite{huang2020predicting,mcginley2023shadow,tran2023measuring}, as we elaborate in Sec.~\ref{sec:discussion}.

Our results provide a powerful recipe for building quantum-inspired combinatorial-optimization solvers competitive with the best available heuristics. A distinctive feature with other classical solvers though is that, since it simulates the actual quantum state of the QA process, GTQA can potentially harness quantum effects such as entanglement and quantum tunneling to escape local minima~\cite{bauer2015entanglement, layden2023quantum}. In turn, our findings also show that quantum dynamics can be classically simulated even for unstructured lattice geometries of low connectivity. In particular, this suggests that the search for quantum advantage should focus on higher-connectivity graphs, where BP-based TN methods struggle.

The paper is structured as follows: In Sec.~\ref{sec:preliminaries} we introduce the graph tensor-network Ansatz to approximate many-qubit quantum states, the belief-propagation algorithm, and the Vidal gauge (the canonical gauge used to truncate the Ansatz). In Sec.~\ref{sec:QUBO} we present the GTQA algorithm and apply it to the non-degenerate QUBO case. In Sec.~\ref{sec:MaxCut} we introduce our measurement-simulation primitive and apply GTQA to the degenerate MaxCut case. In Sec.~\ref{sec:why_gtqa_works} we discuss the limitations and expected regimes of applicability of GTQA. Finally, in Sec.~\ref{sec:discussion} we present the conclusions and discuss perspectives of our work.

\section{Preliminaries}
\label{sec:preliminaries}
\subsection{Ansatz}
We start by introducing a \emph{graph tensor network Ansatz} that is used to represent many-qubit states. 
Let $G=(\nodes, \edges)$ be a \emph{connectivity graph}, where $\nodes = \{1,\dots,N\}$ is the set of vertices and $\edges\subseteq \{\{a, b\} \in \nodes\times\nodes|a \neq b\}$ the set of edges.
We associate each vertex in $\nodes$ to a tensor and each tensor to a qubit. In turn,  $\edges$  indicates how tensors (qubits) are linked (interact).
We take edges with different orders of nodes as equivalent, i.e., $\{a, b\}\equiv\{b, a\}$. See Fig.~\ref{fig:central_fig}(a) for a connectivity graph example. By $\partial a$ we denote the set of neighboring nodes of $a$, i.e., $\partial a = \{b\in \nodes|\{a, b\}\in \edges\}$. We equip each edge $\{a, b\}$ with a \emph{bond index} $j_{ab} \in \Z{d_{ab}}$, where $\Z{d_{ab}} = \{0,\dots, d_{ab} - 1\}$ and $d_{ab}$ is the dimension of the bond index, called the \emph{bond dimension}. Note that $j_{ab}\equiv j_{ba}$ and $d_{ab} = d_{ba}$. We equip each vertex $a$ with a \emph{physical index} $i_a \in \Z{2}$ enumerating the basis state of a qubit. Finally, we equip each vertex $a$ with a \emph{tensor} $T_a$, which can be viewed as a function $T_a: \Z{2}\times\left(\bigtimes_{b\in\partial a} \Z{d_{ba}}\right) \to \cnums$ mapping a physical index and bond indices to a complex number. $T_{a}$ can also be thought of as a $|\partial a| + 1$ dimensional complex rectangular hypermatrix. To access a tensor value, we use square brackets, i.e., $T_a[i_a, \bfj_{\partial a}]$, where $\bfj_{\partial a} = \{j_{ba}|b\in \partial a\}$ is the set of bond indices of $T_a$.

Using the above notations, our Ansatz for an $N$-qubit wave function $\Psi$ reads
\begin{eqnarray}
	\label{eq:graph_tn}
	\Psi\left[\bfi_{\nodes}\right] = \sum_{\bfj_\edges}\prod_{a\in\nodes}T_a\left[i_a,\bfj_{\partial a}\right],
\end{eqnarray}
where $\bfi_\nodes$ is the set of all $N$ physical indices and $\bfj_\edges$ is the set of all bond indices. In Fig.~\ref{fig:central_fig}(b) we give an example of the Ansatz using standard graphical notations for tensor networks. Note that every tensor in this network has a physical index, in contrast to more complex Ansatzes such as the MERA~\cite{vidal2008class} which includes tensors containing exclusively bond indices. We call the Ansatz in Eq.~\eqref{eq:graph_tn} a \emph{graph tensor network} and use it throughout this paper.
\begin{figure*}
	\centering
	\includegraphics[width=\linewidth]{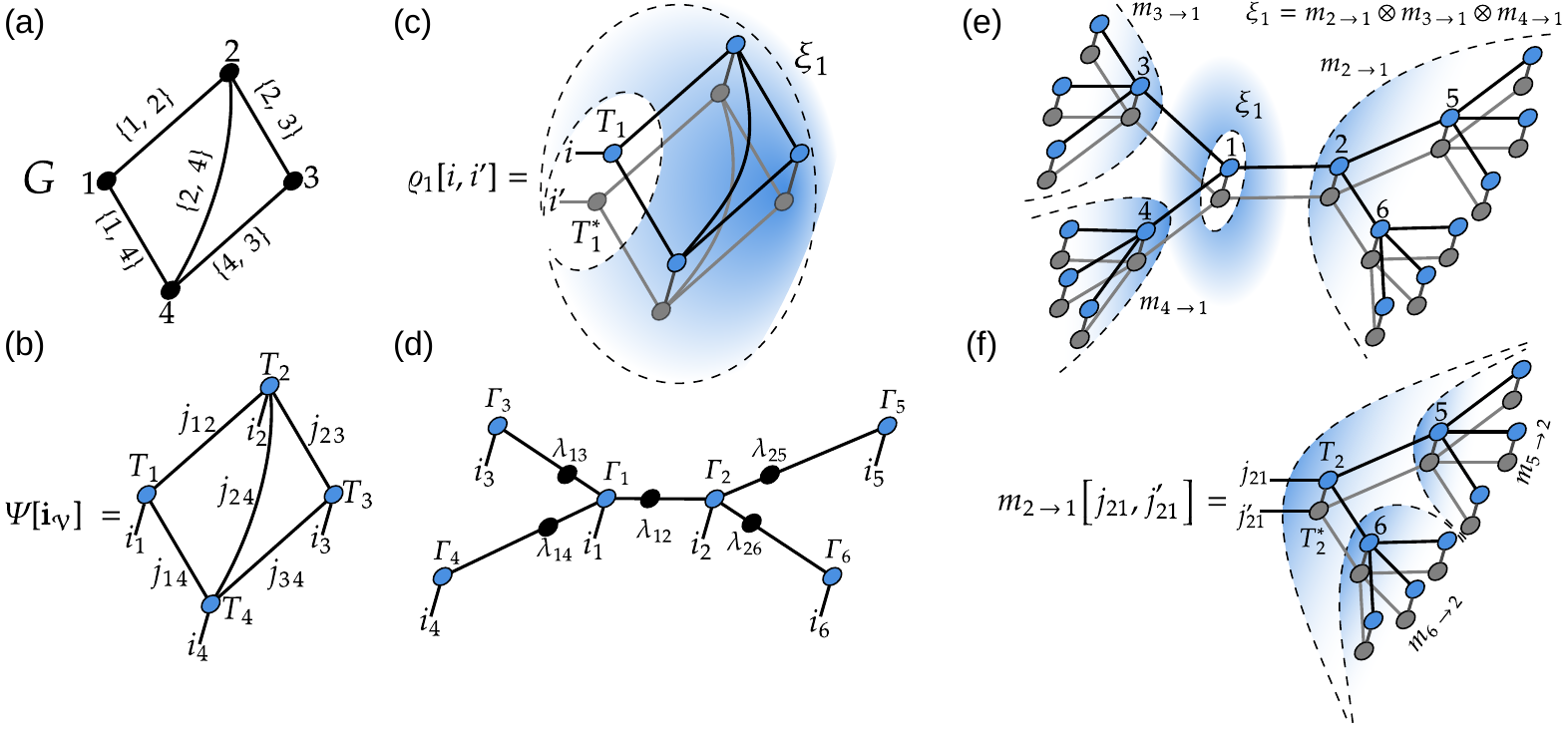}
	\caption{{\bf %
 Main graph tensor-network concepts and tools.} (a) Example of a 4-vertex connectivity graph, underlying the 4-qubit graph TN state in panel (b).
 (b) Tensor diagram example for the Ansatz in Eq.~\eqref{eq:graph_tn}. (c) Tensor diagram example for the reduced density matrix computed by Eq.~\eqref{eq:partial_dens}. The outer colored region is the environment tensor $\xi_1$. (d) Tensor diagram example for a tensor tree in the Vidal gauge. (e) Example of the factorization Eq.~\ref{eq:bp_approx} of the environment tensor for a tree graph tensor network with vertex $1$ as a root. One can note, that messages corresponding to different branches are disconnected. That is, each message matrix is equivalent to the environment tensor if all other branches are removed. Thus, the environment tensor $\xi_1$ factorizes into the tensor product of messages, one from each branch of the root. (f) A diagrammatic demonstration of the recursive relation Eq.~\eqref{eq:local_consistency}. One can note that the recursive structure of a tree implies that the message $m_{2\rightarrow 1}$ can be computed as the contraction of $m_{5\rightarrow 2}$, $m_{6\rightarrow 2}$, $T_2$ and $T_2^*$ as it is expressed in Eq.~\eqref{eq:local_consistency}.}
	\label{fig:central_fig}
\end{figure*}

\subsection{Approximate reduced density-matrix computation with BP}
\label{sec:reduced_den_mat}
An important component that we use in all our numerical experiments is the computation of single-qubit reduced density matrices from the state in Eq.~\eqref{eq:graph_tn}. The exact calculation of these reductions requires in general computational resources that scale exponentially with \(N\). It is useful to express the single-qubit reduced density matrix $\varrho_{a}$ of the $a$-th qubit in terms of $T_a$ and an \emph{environment tensor} $\xi_a$ that encapsulates the exponential complexity of the contraction in question. More precisely, we express the $(i_a,i_a')$-th element of $\varrho_{a}$ in the computational basis as
{\small
\begin{eqnarray}
	\label{eq:partial_dens}
	\varrho_{a}[i_a, i_a'] = \sum_{\bfj'_{\partial a}}\sum_{\bfj_{\partial a}} T_{a}\left[i_a,\bfj_{\partial a}\right]T^\ast_{a}\left[i_a',\bfj'_{\partial a}\right]\xi_a\left[\bfj_{\partial a}, \bfj'_{\partial a}\right],
\end{eqnarray}
}
where $\ast$ stands for complex conjugate and $\xi_a$ is the $a$-th environment tensor. The latter is the result of all the tensor contractions towards $\varrho_a$ except the very final ones involving the $a$-th physical indices $i_a$ and $i_a'$. See Fig~\ref{fig:central_fig} (c) for a graphical example. With $\xi_a$ precomputed, Eq.~\eqref{eq:partial_dens} can be computed efficiently. Formally, the environment tensor can be expressed as
\begin{eqnarray}
	\label{eq:env}
	\xi_a\left[\bfj_{\partial a}, \bfj'_{\partial a}\right]&&=\sum_{\bfj'_\edges\setminus \bfj'_{\partial a}}\sum_{\bfj_\edges\setminus \bfj_{\partial a}}\sum_{\bfi_\nodes\setminus \{i_a\}}\prod_{b\in \nodes\setminus a} T_b\left[i_b, \bfj_{\partial b}\right]\nonumber\\&&\times T^*_b\left[i_b, \bfj'_{\partial b}\right],
\end{eqnarray}
where $\bfj_\edges \setminus \bfj_{\partial a}$ represents the set of all bond indices except those of $T_a$, and $\bfi_\nodes \setminus \{i_a\}$ denotes the set of all physical indices except $i_a$. As mentioned, $\xi_a$ encapsulates the exponential complexity of the entire computation. Therefore, an approximate method for evaluating $\xi_a$ is necessary.

For this, we first note that for a tree tensor network with vertex $a$ as the root the environment tensor factorizes as
\begin{eqnarray}
	\label{eq:bp_approx}
	\xi_a\left[\bfj_{\partial a}, \bfj'_{\partial a}\right] = \prod_{b\in\partial a}m_{b\rightarrow a}[j_{ba}, j'_{ba}],
\end{eqnarray}
where $m_{b \rightarrow a}$ is the environment tensor of a sub-tree linked to $a$ by the edge $\{b, a\}$. We refer to $m_{b \rightarrow a}$ as a \emph{message} matrix, because the equations for computing $m_{b \rightarrow a}$ for all edges resemble the process of passing messages between vertices in the message passing algorithms~\cite{mezard2009information}. See Fig.~\ref{fig:central_fig}(e) for a graphical example. One can show, that messages satisfy the following relation 
\begin{eqnarray}
	\label{eq:local_consistency}
	m_{b\rightarrow a}[j_{ba},j'_{ba}] &&= \sum_{\bfj_{\partial b\setminus a}}\sum_{i_b} T_b\left[i_b, \bfj_{\partial b}\right]T_b^*\left[i_b, \bfj'_{\partial b}\right]\nonumber\\&&\times\prod_{c\in \partial b\setminus a} m_{c\rightarrow b}[j_{cb},j'_{cb}],
\end{eqnarray}
which we derive in App.~\ref{appx:self_consistency_derivation}. Eq.~\eqref{eq:local_consistency} follows from the recursive structure of the sub-tree as it is shown in Fig.~\ref{fig:central_fig}(f). Eq.~\eqref{eq:local_consistency} implies that messages are Hermitian and positive-definite. We can solve Eq.~\eqref{eq:local_consistency} starting from the leaves of the tree, recursively updating messages as we descend towards the root. Having all the messages, one can compute Eq.~\eqref{eq:partial_dens} efficiently using Eq.~\eqref{eq:bp_approx}.

If $G$ is not a tree graph, i.e., if it has loops, Eq.~\eqref{eq:bp_approx} no longer holds; but it can be taken as a mean-field-like approximation to the actual environment tensor. This approximation can still be computed through Eq.~\eqref{eq:local_consistency}. To solve Eq.~\eqref{eq:local_consistency}, one can initialize messages as random positive-definite Hermitian matrices and run a fixed-point iteration method%
. More precisely, in each iteration one updates all edges in $\edges$ using Eq.~\eqref{eq:local_consistency} and iterates until convergence (as opposed to recursively solving for messages from leaves to roots, for tree graphs). This iterative approach is known as the BP algorithm~\cite{mezard2009information, tindall2023gauging}, which we describe in detail in App.~\ref{appx:bp_alg}. Note that, in general, convergence of the BP algorithm is not guaranteed. However, in practice, it typically converges well, and cases where it struggles to converge are discussed below.

\subsection{Vidal gauge and truncation}
\label{subsec:vidal_gauge}
To control the complexity of a graph tensor network, one should be able to truncate its bond indices. For tree tensor networks, this can be done with optimality guarantees. For this, we need to introduce a modification of the Ansatz Eq.~\eqref{eq:graph_tn} which reads
\begin{eqnarray}
	\label{eq:vidal_gauge_tn}
	\Psi[\bfi_\nodes] = \sum_{\bfj_\edges}\prod_{a\in\nodes} \Gamma_{a}[i_a, \bfj_{\partial a}]\prod_{\{b, c\}\in \edges}\lambda_{bc}[j_{bc}],
\end{eqnarray}
where $\{\lambda_{bc}\}_{\{b, c\}\in \edges}$ are newly introduced vectors assigned to each edge. One can compute $\lambda_{bc}$ as a singular vector of the following matrix $\sum_{j = 0}^{d_{ab} - 1} m^{\frac{1}{2}}_{a\rightarrow b}[j, k]m^{\frac{1}{2}}_{b\rightarrow a}[j, l]$, where $(.)^\frac{1}{2}$ stands for the square root of a matrix. In the tensor tree case $\lambda_{ab}$ coincides with the Schmidt coefficients and the entire Ansatz can be viewed as the simultaneous Schmidt decomposition of a state with respect to all single edge cuts of a tensor tree. 

The modified Ansatz Eq.~\eqref{eq:vidal_gauge_tn} also satisfies a \emph{local orthogonality condition} which is discussed in App.~\ref{appx:local_orthogonality}. It guarantees orthogonality of corresponding Schmidt vectors in case of a tree tensor graph. The modified Ansatz Eq.~\eqref{eq:vidal_gauge_tn} is called a \emph{Vidal gauge}. See Fig.~\ref{fig:central_fig}(d) for a graphical example. Due to the Eckart-Young-Mirsky theorem~\cite{eckart1936approximation, mirsky1960symmetric}, truncation of minimal Schmidt coefficients of a particular edge leads to the best low-bond-dimension approximation in either the $2$-norm or the Frobenius norm. For the formal description of the truncation procedure see App.~\ref{appx:vidal_gauge}. The global Frobenius error of the edge $\{a, b\}$ truncation reads
\begin{eqnarray}
	\varepsilon_{ab}^\chi = \sqrt{\sum_{j_{ab} =\chi}^{d_{ab} - 1}\lambda^2_{ab}[j_{ab}]},
\end{eqnarray}
where $\chi$ is the new bond dimension. The Vidal gauge can be computed efficiently from Eq.~\eqref{eq:graph_tn} using messages as it is shown in App.~\ref{appx:vidal_gauge}. For general graphs with loops, the Vidal gauge and the algorithm from App.~\ref{appx:vidal_gauge} can still be applied as heuristics.

To reduce the computational complexity, one often sets a graph tensor network to the Vidal gauge only once per several truncations. Each truncation corrupts the gauge, increasing errors in subsequent truncations. To track the deviation from the Vidal gauge, one defines a residual $R$ of the local orthogonality condition discussed in App.~\ref{appx:local_orthogonality}. $R$ can be seen as the distance to the Vidal gauge.
When $R$ becomes too large, it is necessary to perform \emph{regauging} to set the tensor network back to the Vidal gauge with $R = 0$. This consists of three steps: (\emph{i}) one transforms Eq.~\eqref{eq:vidal_gauge_tn} back to the initial Ansatz Eq.~\eqref{eq:graph_tn}; (\emph{ii}) one runs BP algorithm to compute messages; (\emph{iii}) one recovers the Vidal gauge with $R=0$ from Eq.~\eqref{eq:graph_tn} using messages and the algorithm from App.~\ref{appx:bp_alg}. Step (\emph{i}) involves splitting the Schmidt coefficients between neighboring tensors as
\begin{eqnarray}
	\label{eq:symmetric_gauge}
	T_a[i, \bfj_{\partial a}] = \Gamma_{a}[i, \bfj_{\partial a}]\prod_{b\in\partial a}\lambda^{\frac{1}{2}}_{ba}[j_{ba}].
\end{eqnarray}
Step (\emph{ii}) runtime can be sufficiently reduced by a proper messages initialization before running BP. If $R$ is small, $\{\lambda_{bc}\}_{\{b,c\}\in\edges}$ contains information about converged messages and we can reconstruct these messages as follows
\begin{eqnarray}
	\label{eq:reconstructed_messages}
	m_{a\rightarrow b}[j, j'] = \delta[j, j']\frac{\lambda_{ba}^{\frac{1}{2}}[j]}{\sum_{j}\lambda_{ba}^{\frac{1}{2}}[j]}.
\end{eqnarray}
Indeed, for $R = 0$ and vertex tensors fulfilling Eq.~\eqref{eq:symmetric_gauge}, messages Eq.~\eqref{eq:reconstructed_messages} satisfy Eq.~\eqref{eq:local_consistency} up to a constant factor. If $R\neq 0$ but small, the resulting messages are close to the solution of Eq.~\eqref{eq:local_consistency} and serve as a ``warm'' start for the BP algorithm. This typically converges in few iterations, much faster than starting from scratch.

\subsection{Quantum-gate application in the Vidal gauge}
\label{sec:gate_app}
To apply a single-qubit unitary gate in the Vidal gauge, one needs to update the corresponding tensor as follows
\begin{eqnarray}
	\tilde{\Gamma}_a[i_a,\bfj_{\partial a}] = \sum_{i'_a}W[i_a, i'_a]\Gamma_a[i'_a,\bfj_{\partial a}],
\end{eqnarray}
where $W$ is a single-qubit unitary gate applied to the $a$-th qubit. The Vidal gauge is preserved due to the unitarity of $W$.

The application of a two-qubit unitary gate is more involved. In case when a gate is applied to neighboring qubits one needs to update only neighboring tensors and the Schmidt coefficients in between. The corresponding algorithm called the \emph{simple update} algorithm is given in App.~\ref{appx:simple_update}. It allows one to evolve a graph tensor network while preserving the state in the Vidal gauge and truncate it when the bond dimension reaches a threshold $\chi$.

\section{Quantum annealing simulation for non-degenerate QUBO}
\label{sec:QUBO}
In this section, we present our graph tensor-network quantum annealer (GTQA) for the case of quadratic unconstrained binary optimization (QUBO). In particular, we consider non-degenerate QUBO instances, which simplifies the extraction of problem solution from the QA final state. 
We use the techniques discussed in Sec.~\ref{sec:preliminaries} as building blocks. The section is organized as follows. First, we formulate QUBO problems, introduce QA in general, and explain GTQA. Then, we benchmark GTQA against state-of-the-art classical solvers. Next, we study the entanglement-entropy dynamics induced by GTQA, and discuss the simulation complexity. Finally, we benchmark GTQA in terms of global state infidelities as well as trace-distance errors of the single-qubit reductions.

\subsection{Simulation algorithm}
\label{prgrph:large_scale_qa}
The problem we want to solve is the maximization of the following QUBO objective function
\begin{eqnarray}
	\label{eq:objective_function}
	E(\boldsymbol{x}) = \sum_{\{a,b\} \in \edges}J_{ab} \ x_ax_b + \sum_{a=1}^N h_a x_a,
\end{eqnarray}
over $N$-long strings $\boldsymbol{x} = (x_1,\dots,x_N)$, with  spin variables $x_a\in\{-1, 1\}$ for all $a\in\mathcal{V}$, and
where $J_{ab}$ and $h_a$ are respectively coupling constants and local magnetic fields. Here we consider only QUBO problems where the objective function is maximized by only solution string $\boldsymbol{x}^* = (x^*_1,\dots,x^*_N)$. 
We now encode the objective function $E(x)$ into the Ising-model quantum Hamiltonian
\begin{equation}
    H_{\rm Ising} = \sum_{\{a,b\} \in \edges}J_{ab} \ Z_aZ_b + \sum_{a=1}^N h_a Z_a,
\end{equation}
where we have the $z$ Pauli operator $Z_a$ acting on the $a$-th qubit. We map string $\boldsymbol{x}$ into $N$-qubit computational basis states $\ket{\boldsymbol{x}} = \bigotimes^N_{a=1}\ket{(-x_a+1)/2}$, i.e. spin variable $x_a=1$ is mapped to qubit state $\ket{0}$, and $x_a=-1$ to $\ket{1}$.
One can now see that the objective-function value is given by the expectation value of the Hamiltonian, i.e. $E(\boldsymbol{x})=\bra{\boldsymbol{x}} H_{\rm Ising} \ket{\boldsymbol{x}}$ and the maximal energy state of $H_{\rm Ising}$  matches the optimal solution $\ket{\boldsymbol{x}^*}$. 

QA %
finds the maximal energy by adiabatic transformation of the Hamiltonian over time $t$~\cite{hauke2020perspectives}. First, we start in the maximal energy state $\ket{\Psi(0)} = \ket{+}^{\otimes N}$ of the simple Hamiltonian $H_{\rm mixing} = \sum_{a=1}^N X_a$, where $\ket{+} =1/\sqrt{2}(\ket{0} + \ket{1})$ and $X_a$ is the $x$ Pauli operator. 
Then, we define the time-dependent Hamiltonian $H(t)=(1-s(t))H_{\rm Ising}+s(t)H_{\rm mixing}$ with parameter $s(t)$, which is a slowly varying function between $s(0)=1$ and $s(T)=0$.
We vary the Hamiltonian $H(t)$ according to the schedule $s(t)$ to drive a  time-dependent Hamiltonian evolution of the system for a total time $T$.
When the dynamics is chosen to be much slower than the inverse of the energy gap between the two largest energy states of $H(t)$ for all $t$, then the adiabatic theorem of quantum mechanics guarantees that the final state $\ket{\Psi(T)}$ is the maximal energy state of $H_{\rm Ising}$, which is the optimal solution is $\ket{\boldsymbol{x}^*}$.

To simulate QA classically, we use GTQA algorithm containing following steps: one represents the initial state $\ket{\Psi(0)}$ as a graph tensor network in the Vidal gauge which has $d_{ab} = 1$ for $\forall \{a, b\} \in \edges$ since $\ket{\Psi(0)}$ is a product state; one trotterizes the QA passage into a quantum circuit consisting of one- and two-qubit gates~\cite{suzuki1990fractal} and apply those gates sequentially to $\ket{\Psi(0)}$ preserving the Vidal gauge. See App.~\ref{appx:annealing_trotterization} for more details on the QA Trotterization. Whenever any bond index of a graph tensor network reaches a dimension greater than $\chi$, we truncate it. We occasionally perform regauging to keep the Vidal gauge valid. We use the final graph tensor network form of $\ket{\Psi(T)}$ as an approximation of $\ket{\boldsymbol{x}^*}$.%

Usually, to read out the solution string of the QA algorithm, one measures each qubit of the output state $\ket{\Psi(T)}$ in the computational basis $\{\ket{0},\ket{1}\}$. However, as discussed in Sec.~\ref{subsec:measurements_sampling}, simulating the post-measurement states with graph tensor networks presents subtleties coming from the necessary regauging after each qubit measurement simulation. Since the $\ket{\Psi(T)}$ is close to the product state $\ket{\boldsymbol{x}^*}$, a simpler way to extract (an approximation to) $\boldsymbol{x}^*$ from our graph tensor network representation of $\ket{\Psi(T)}$ is to compute the single-qubit reduced density matrix $\varrho_a$ of qubit $a$. Indeed, if $\ket{\Psi(T)}=\ket{\boldsymbol{x}^*}$,
\begin{eqnarray}
	\label{eq:rounding_rule}
	x^*_a = \begin{cases}
		\ \ 1,\quad {\rm if} \ \varrho_a[0, 0] > \frac{1}{2}, \\
		-1,\quad {\rm otherwise},
	\end{cases}
\end{eqnarray}
to get the solution string. 
Each single-qubit reduced density matrix can in turn be computed from the graph tensor network representation of $\ket{\Psi(T)}$ using Eq.~\eqref{eq:partial_dens}, Eq.~\eqref{eq:bp_approx} and messages computed via BP algorithm. %

\subsection{Benchmarking against conventional QUBO solvers}
\label{prgrph:qubo_benchmarking}

\begin{figure*}
	\includegraphics[width=\linewidth]{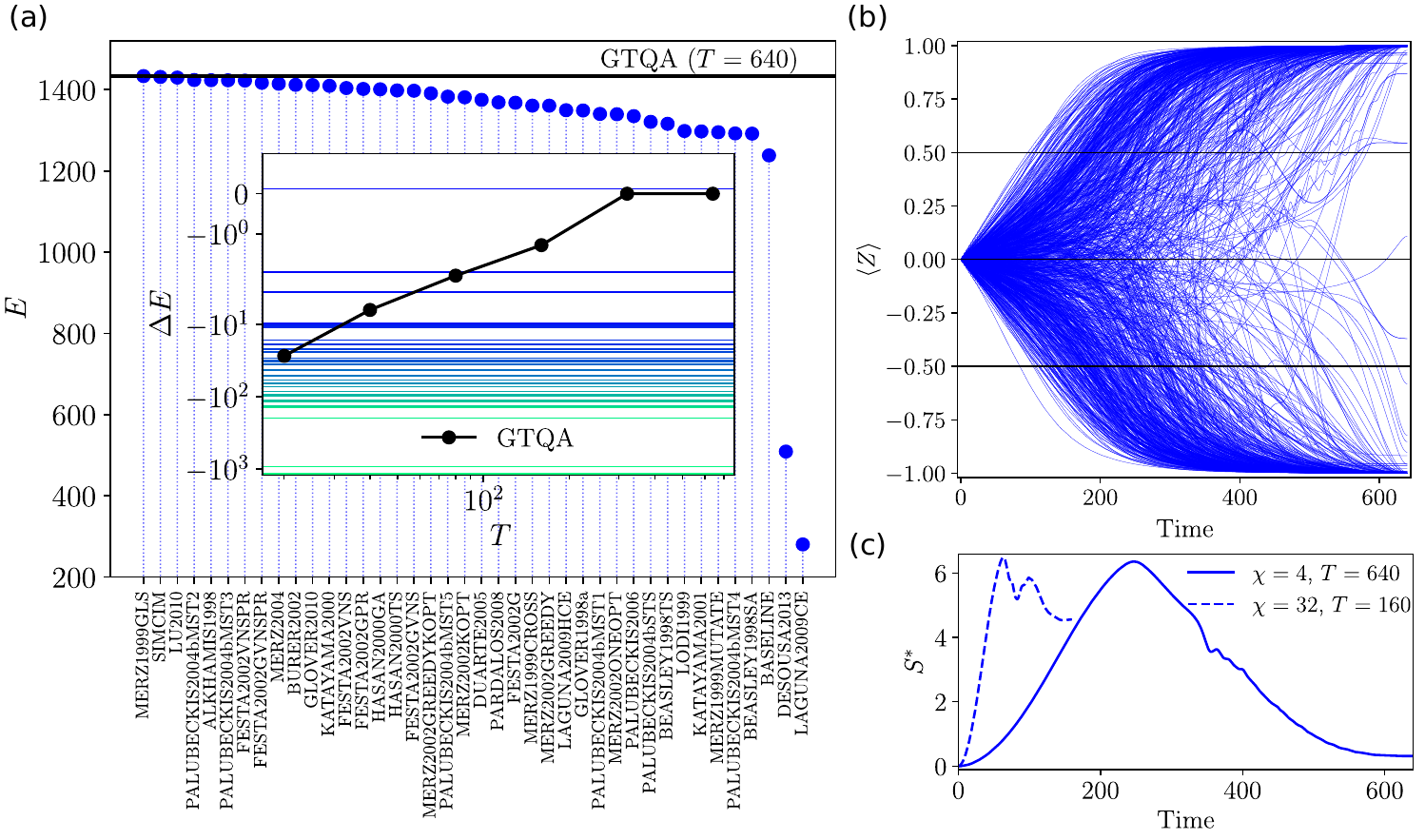}
	\caption{{\bf GTQA's performance on a QUBO.} Numerical results for a random 1000-qubit instance on the random 3-regular graph of Fig. \ref{fig:circuit_tn}(b). 
 (a) The comparison of the QUBO problem solution obtained by the GTQA algorithm and solutions obtained by heuristics from MQLib and SimCIM algorithm. The black horizontal line represents the largest $E$ value obtained by the GTQA algorithm which corresponds to $T = 640$. Blue dots corresponds to $E$ obtained by different heuristics from MQLib. In $X$-axis we provide names of the corresponding heuristics. The black curve in the inset demonstrates how $E$ obtained by the GTQA algorithm improves (larger is better) with increasing $T$. $Y$-axis in the inset represents $\Delta E$ which is the difference between the value found by GTQA and other methods. Colored horizontal lines in the inset show $\Delta E$ obtained by heuristics. (b) Dynamics of $\langle Z \rangle$ for each qubit for largest total annealing time $T=640$. $997$ qubits out of $1000$ crossed a confidence threshold $\pm0.5$ by the end of dynamics. X-axis represents time during the annealing process. (c) Approximate entanglement entropy with respect to the balanced graph bipartition cutting the minimum number of edges [see Eq. \eqref{eq:bipartition}] as a function of time $t$, for two different annealing durations $T$. For $T=160$, the system shows a quench dynamics, with the final state significantly entangled. In turn, for $T=640$, it features an almost-adiabatic behavior, with the very small final entropy. The small final entropy also agrees with the panel (b) where we observe that the final state is almost a product state, which immediately implies low entropy. Note also that the peak entanglement does not grow from $T=160$ to $T=640$. This suggests that one could tackle even higher $T$ (i.e., even more adiabatic schedules) with the same bond dimension $\chi=4$.}
	\label{fig:largest_numerical_experiment}
\end{figure*}

We compare the performance of the GTQA algorithm equipped with the simple solution string extraction Eq.~\eqref{eq:rounding_rule} with those of standard heuristic QUBO solvers from the MQLib project~\cite{dunning2018what} as well as the SimCIM algorithm~\cite{tiunov2019annealing} with fine-tuned hyperparameters. We consider a random {\it $3$-regular graph}. That is, a random choice of $\nodes$ and $\edges$ where all $N$ vertices have degree $3$ and an instance of Eq. \eqref{eq:objective_function} where all coefficients $J_{ab}$ and $h_a$ are chosen at random too. The specific choice of $\nodes$ and $\edges$ is done with an algorithm from Ref.~\cite{steger1999generating} implemented in the NetworkX library~\cite{hagberg2008exploring}; and coefficients $J_{ab}$ and $h_a$ are independently sampled from a normal distribution ${\cal N}(0, 1)$. Random choice of $h_a$ removes degeneracy of the solution string of $E(\boldsymbol{x})$ leading to the QUBO problem with the unique maximum. In turn, we consider $N = 1000$ qubits, which renders  any exact simulation based on dense state-vector representation intractable.

The numerical results are shown in Fig.~\ref{fig:largest_numerical_experiment}(a). We used a Trotterization time step $\delta t = 0.2$ and a Hamiltonian schedule given by $s(t) = 1 - \frac{t}{T}$, with different durations $T = 20, 40, 80, 160, 320$, and $640$. We used $\chi=32$ for $T \leq 160$ and $\chi=4$ for longer times to reduce computational cost. 
The runtime of the longest simulation (the one corresponding to $T = 160$ and $\chi = 32$) was close to $10$ days on a single CPU kernel.
The runtime of each heuristic from MQLib was upper-bounded by $100$ seconds, which is typically enough to get the best performance. SimCIM was run $1000$ times, and its best solution string among all 1000 runs was selected. As one can see in the figure, the GTQA algorithm performs very well. For $T=640$ it gives $E = 1433.3739$ which is very close to the highest value $1433.4906$ found by heuristics. 

To give some intuition of how each qubit converges to a particular polarized state within the GTQA algorithm, we plot the dynamics of $\langle Z\rangle$ for each qubit in Fig.~\ref{fig:largest_numerical_experiment}(b) for $T=640$. One can see that most of the qubits converged to polarized states by the end of the dynamics; in particular, $997$ qubits crossed a threshold $\langle Z\rangle = \pm 0.5$ to states close to eigenstates of $Z$ by the end of the dynamics. In particular, $997$ out of the $1000$ qubits end up in a reduced state featuring  $|\langle Z\rangle| \geq 0.5$. That is, the final $N$-qubit state is close to a pure product state in the computational basis. We also compared solution strings obtained by the GTQA algorithm for $T=640$ and by the best heuristic in Fig.~\ref{fig:circuit_tn}(b). One can note that these solution strings are very close to each other and differ only in $13$ spins. This is expected because the best solution string is unique due to random local magnetic fields.

The fact that the final state is close to a product state that encodes a %
good solution to the optimization problem %
is a strong evidence that, for $T=640$, GTQA simulates a nearly adiabatic QA process accurately. In Sec.~\ref{sec:sim_accuracy}, we also benchmark GTQA's performance in terms of state infidelities due to truncations and the mean-field approximation. 
But, before that, we study the evolution of entanglement during the simulated %
process.

\subsection{Entanglement entropy dynamics}
\label{subsec:apprx_entropy}
An interesting feature of GTQA is that, when the QA process does not generate long range correlations, one can approximately compute the entanglement entropy with respect to an arbitrary system bipartition using a mean-field approximation.
For non-degenerate problems with non-zero gap, entanglement entropy is useful for identifying the transition from non-equilibrium (quench) to adiabatic dynamics. Moreover, the amount of entanglement generated has a direct connection to the complexity of classical simulation.

Entanglement entropy is defined for a particular partitioning of a quantum system into two subsystems. Usually, to describe the global entanglement complexity of the quantum state, one chooses an equal-sized bipartition, where there is as little connection as possible between the partitions. For example, for a one-dimensional system, one partitions via the center of the system. 
However, for a random graph structure, it is a priori not clear how to choose the partition. Here, we introduce a heuristic to find a good partition to characterize the entanglement complexity of GTQA.
Let $A\subseteq \nodes$ be the qubits of the first subsystem, $\overline{A} = \nodes \setminus A$ be the second subsystem, and  $|\partial A|= \left\{\{a, b\}\in\edges|a\in A, b\in \overline{A}\right\}$ the number of interaction terms of the Hamiltonian connecting two subsystems.
The most natural requirement is to keep the partitioning balanced, i.e. $|A| \approx |\overline{A}|$. Further, we want to minimize the connectivity between the partitions to avoid contributions from  local entanglement. For example, let us imagine a system consisting of several Bell pairs. If we split the system such that for each pair its counterparts are in different subsystems, we will get the maximal entanglement entropy. If we split the system such that both counterparts of each pair belong to the same subsystem, we will get zero entanglement entropy. Clearly, from the point of view of many-body dynamics, zero entanglement is the option that reflects the genuine complexity of the system. 

Hence, we also need to minimize the entanglement entropy between subsystems while keeping the partitioning balanced. A good heuristic for that is to minimize the number of interaction terms between subsystems, i.e. $|\partial A|$. A partitioning that approximately satisfies both requirements can be found as follows
\begin{eqnarray}
	\label{eq:bipartition}
	A^* = \underset{A}{\argmin}\frac{|\partial A|}{|A||\overline{A}|},
\end{eqnarray}
where the denominator penalizes for the imbalance while the numerator is the number of the interaction terms that we want to minimize. We found the partitioning Eq.~\eqref{eq:bipartition} approximately using spectral graph partitioning~\cite{von2007tutorial}. It gave us a partitioning into two subsystems of sizes $556$ and $444$ qubits with number of interacting terms between subsystems $|\partial A^*| = 158$. We denote by $S^*$ the entanglement entropy with respect to this bipartition.

$S^*$ can be computed from the Schmidt coefficients relative to the bipartition $A^*$. These can be approximated in a natural way with the singular-value vector $\{\lambda_{ab}\}_{\{a, b\}\in\edges}$ given by the Vidal gauge. The approximate Schmidt coefficients are defined as
\begin{eqnarray}
    \label{eq:mean_field_schmidt}
    \lambda^*[j_{\partial A^*}] = \prod_{\{a, b\}\in \partial A^*}\lambda_{ab}[j_{ab}].
\end{eqnarray}
That is, this is a mean-field-like approximation, since it corresponds to approximating the Schmidt vectors of a bipartition as tensor products of  Schmidt vectors of independent tree tensor networks passing through edges cut by the bipartition.
 Using Eq.~\eqref{eq:mean_field_schmidt}, one can in turn approximate the entanglement entropy as
\begin{eqnarray}
    \label{eq:apprx_entropy}
	S^* &&\approx -\sum_{j_{\partial A^*}} \left(\lambda^*[j_{\partial A^*}]\right)^2 \log \left(\left(\lambda^*[j_{\partial A^*}]\right)^2\right)\nonumber\\ &&= -\sum_{\{a, b\}\in \partial A^*}\sum_{j_{ab}=1}^\chi \lambda_{ab}^2[j_{ab}]\log\left(\lambda_{ab}^2[j_{ab}]\right).
\end{eqnarray}
This approximation has also been used in Ref.~\cite{tindall2023efficient}.
We stress that Eq.~\eqref{eq:mean_field_schmidt} is approximate even for tree-graphs tensor networks. Hence, Eq.~\eqref{eq:apprx_entropy} can potentially accumulate errors from that approximation as well as from the intrinsic mean-field approximation of the BP procedure itself. However, in App.~\ref{appx:entropy_dynamics_cmp}, we observe a good qualitative agreement between the approximate entanglement entropy and its exact counterpart obtained via brute-force calculations for small random graphs. We also observe in our numerical experiments that the approximate entropy tends to overestimate the exact one.

\begin{figure*}[t!]
	\centering
	\includegraphics[width=1.\linewidth]{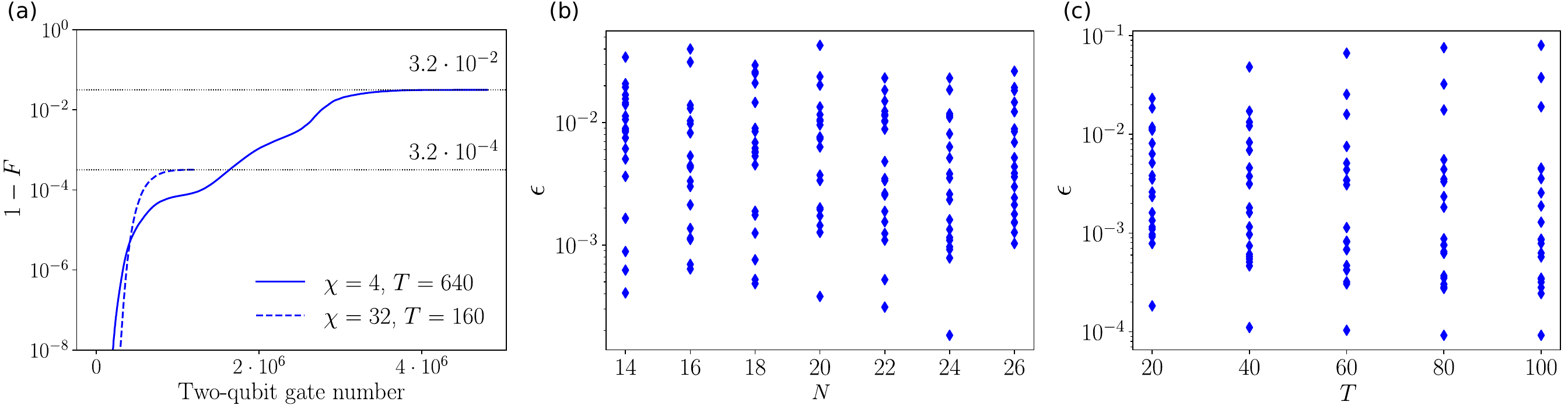}
	\caption{{\bf GTQA accuracy analysis for random 3-regular graph QUBO problems.}
	(a) Approximate dynamics of infidelity with growing number of two-qubit gates within the GTQA algorithm. Even for largest annealing time $T = 640$ and the smallest bond dimension $\chi = 4$, truncations do not have a large impact. (b) Trace-distance error of single-qubit reduced states averaged over the $N$ qubits [see Eq.~\eqref{eq:error}] and over time for $20$ random $3$-regular graphs per $N$ and for $N$ ranging from $14$ to $26$.
    The total annealing time is $T = 20$ for all instances.
    (c) Same type of errors as in panel (b) but for fixed $N=24$ and for different $T$ ranging from $20$ to $100$.
    Note that, interestingly, the median of the error does not grow with $N$ or $T$, only the dispersion of the error slightly grows with $T$.}
	\label{fig:accuracy}
\end{figure*}

We plot the dynamics of $S^*$ for $T = 160$ and $T=640$ in Fig.~\ref{fig:largest_numerical_experiment}(c). For $T=160$ the dynamics is still in the quench regime, since the final entropy is relatively high, implying 
that the final state is not a product state as one expects. But for $T=640$, the final entropy is close to zero, indicating that the dynamics is close to adiabatic, since the final state approximates a product state (that gives the solution of the optimization problem). The observed entropy behaviour is also compatible with the single-qubit $Z$ expectation values shown in Fig.~\ref{fig:largest_numerical_experiment}(b), which also indicate that the final state is close to a product state. This provides yet another confirmation of the accuracy of Eq. \eqref{eq:apprx_entropy} to approximate $S^*$ for the case in question. 

The maximal entanglement entropy achieved during the dynamics simulation for $T=640$ is $S^*_{\rm max} \approx 6.4$. 
With this, we can bound the bond dimension that an MPS would need for such simulation. 
Assuming a (highly unrealistic) flat Schmidt spectrum, one would require an MPS bond dimension at least $\chi=580$ 
to achieve the entanglement entropy $S^*_{\rm max}$.
However, the actual Schmidt spectrum is far from flat and has a long decaying tail. To accurately simulate such spectrum using an MPS, we would actually need a bond dimension many orders of magnitude higher, which is infeasible in practice. For instance, if we stack the graph TN that we use for $T=640$ into an MPS, we end up with an MPS bond dimension $4^{158}$, where $158$ comes from $|\partial A^*|$
and $4$ from the bond dimension of our graph TN. An MPS with such high bond dimension would allow us to simulate the QA process with an approximate infidelity equal to $0.032$ (see Sec.~\ref{sec:sim_accuracy}). Note however that $4^{158}$ is only an upper bound, since the corresponding MPS might be further truncated. But the argument makes the point that the necessary MPS bond dimension is actually much higher than 580. 
In contrast, for the graph tensor network Ansatz, the average maximal entanglement entropy per bond index is just $\langle S_{\rm max}^*\rangle \approx {S^*_{\rm max}}/{|\partial A^*|} = 0.04$. Interestingly, Fig.~\ref{fig:largest_numerical_experiment}(c) also shows that the peak entanglement entropy during the annealing does not grow from $T=160$ to $T=640$. This suggests that one could scale $T$ up even further without increasing 
the bond dimension ($\chi = 4$).

\subsection{Simulation accuracy analysis}
\label{sec:sim_accuracy}
To claim that the GTQA algorithm simulates the QA process accurately, we analyze two sources of error of the GTQA algorithm, i.e. the truncation error and the error caused by the mean-field-like approximation of the environment tensor Eq.~\eqref{eq:bp_approx}. We begin from the truncation error. In particular, we consider the truncation impact on the state fidelity. We rely on the assumption about multiplicativity of the fidelity~\cite{zhou2020limits}, i.e. the fidelity after applying $M$ gates can be approximated as follows
\begin{eqnarray}
	\label{eq:total_fidelity}
	F(M) \approx \prod_{i = 1}^{M} f_i,
\end{eqnarray}
where $f_i$ is the $i$-th gate fidelity. One-qubit gates always have fidelity $1$ within the GTQA algorithm, thus, we count only two-qubit gates. The fidelity of each two-qubit gate can be estimated as follows
\begin{eqnarray}
\label{eq:gate_fidelity}
	f \approx \left(\sum_{k=1}^{\chi}\lambda^2[k]\right)^2,
\end{eqnarray}
where $\lambda$ is the Schmidt vector of a particular edge that is being truncated after a two-qubit gate application. Eqs.~\eqref{eq:total_fidelity} and \eqref{eq:gate_fidelity} give us a numerically efficient way to estimate $F(M)$. For better visibility, we plot dynamics of infidelity $1 - F(M)$ in Fig.~\ref{fig:accuracy}(a). One can observe, that for $T=160$ and $\chi=32$ infidelity is negligibly small meaning that we can neglect the truncation impact on the accuracy. For $T=640$ and $\chi=4$ infidelity is notable, but still have minor impact.

Next, we analyze the impact of the mean-field-like approximation of the environment tensor Eq.~\eqref{eq:bp_approx} on the state fidelity. For this purpose, we compare GTQA algorithm with the exact state-vector-based dynamics simulation for small system sizes assuming that truncation does not introduce a sensible impact. As an accuracy metric, we consider trace distance averaged over qubits and time steps, i.e.,
\begin{eqnarray}
	\label{eq:error}
	\epsilon = \frac{\delta t}{2TN}\sum_{a=1}^N\sum_{k=1}^{T / \delta t} \left\|\varrho_a^{({\rm exact})}(k) - \varrho_a^{({\rm bp})}(k)\right\|_1,
\end{eqnarray}
where $\varrho_a(k)$ is the density matrix of the $a$-th qubit at discrete time step $k$. For $T=20$, $\chi=4$, and system sizes ranging from $14$ to $26$ qubits in steps of 2, we generated $20$ random $3$-regular graphs per system size ($140$ random graphs in total) and evaluated Eq.~\eqref{eq:error} for each graph instance. In Fig.~\ref{fig:accuracy}(b) we plotted $\epsilon$ for each graph. One can observe, that the error is at most a few percent for the worst graphs and it does not increase with increasing system size. In order to understand how error scales with increasing $T$, we plot Eq.~\eqref{eq:error} for $20$ random $3$-regular graphs consisting of $24$ qubits and various $T$ in Fig.~\ref{fig:accuracy}(c). As one can see, the standard deviation of the error slowly increases with increasing $T$, but its median does not increase.

\section{Quantum annealing simulation for MAXCUT}
\label{sec:MaxCut}
In this section we solve large MaxCut instances using our GTQA algorithm with quantum measurements simulation. In contrast with QUBO problem from Sec.~\ref{sec:QUBO} the MaxCut optimization problem has a highly degenerate maximum, and the quantum state after QA is not a product state, but an entangled superposition of many optimal solutions. 
To extract the solutions from the quantum state, we have to simulate quantum measurements in the computational basis. The section consists of two parts. First, we discuss the MaxCut problem and the GTQA algorithm equipped with quantum measurements simulation. Second, we benchmark the MaxCut solution found by the GTQA algorithm against solutions obtained with conventional solvers.
\subsection{Simulation algorithm}
\label{subsec:measurements_sampling}

MaxCut can be formulated as a maximization problem of the following objective function
\begin{eqnarray}
	\label{eq:maxcut_energy}
	E(\boldsymbol{x}) = \sum_{\{a, b\}\in \edges} \frac{1 - x_a x_b}{2},
\end{eqnarray}
where a spin variable $x_a$ is viewed as a label of a class that vertex $a$ belongs to and the objective-function value is the size of a cut, i.e., the number of edges connecting classes. Eq.~\eqref{eq:maxcut_energy} can be transformed to the equivalent one by a simple global shift and rescaling
\begin{eqnarray}
	\label{eq:maxcut_energy_shifted}
	E'(\boldsymbol{x}) = -\sum_{\{a, b\}\in \edges} x_a x_b,
\end{eqnarray}
which corresponds to choosing $J_{ab} = -1$ for $\forall \{a, b\} \in \edges$ and $h_a = 0$ for $\forall a \in \nodes$ in Eq.~\eqref{eq:objective_function}. Thus, we use the same GTQA algorithm as in Sec.~\ref{prgrph:large_scale_qa} with the final Hamiltonian
\begin{eqnarray}
	\label{eq:maxcut_qa_final_hamiltonian}
	H_{\rm Ising} = -\sum_{\{a, b\}\in\edges} Z_a Z_b,
\end{eqnarray}
encoding the MaxCut objective function Eq.~\eqref{eq:maxcut_energy_shifted}.

The final state $\ket{\Psi(T)}$ of the QA process is close to the ground state of Eq.~\eqref{eq:maxcut_qa_final_hamiltonian}. But the maximum of the MaxCut objective function is typically highly degenerate, meaning that $\ket{\Psi(T)}$ is the superposition of all the solution strings of the MaxCut objective function. Therefore, the simple rounding rule Eq.~\eqref{eq:rounding_rule} does not work anymore, since reduced density matrices could be proportional to the identity. Thus, we need to simulate quantum measurements in order to extract the solution string. It can be done in tree steps for each qubit $a$: (\emph{i}) compute the partial density matrix $\varrho_a$ using messages and sample a measurement outcome $x_a\in\{-1, 1\}$ from the probability mass function
\begin{figure}
	\centering
	\includegraphics[width=\linewidth]{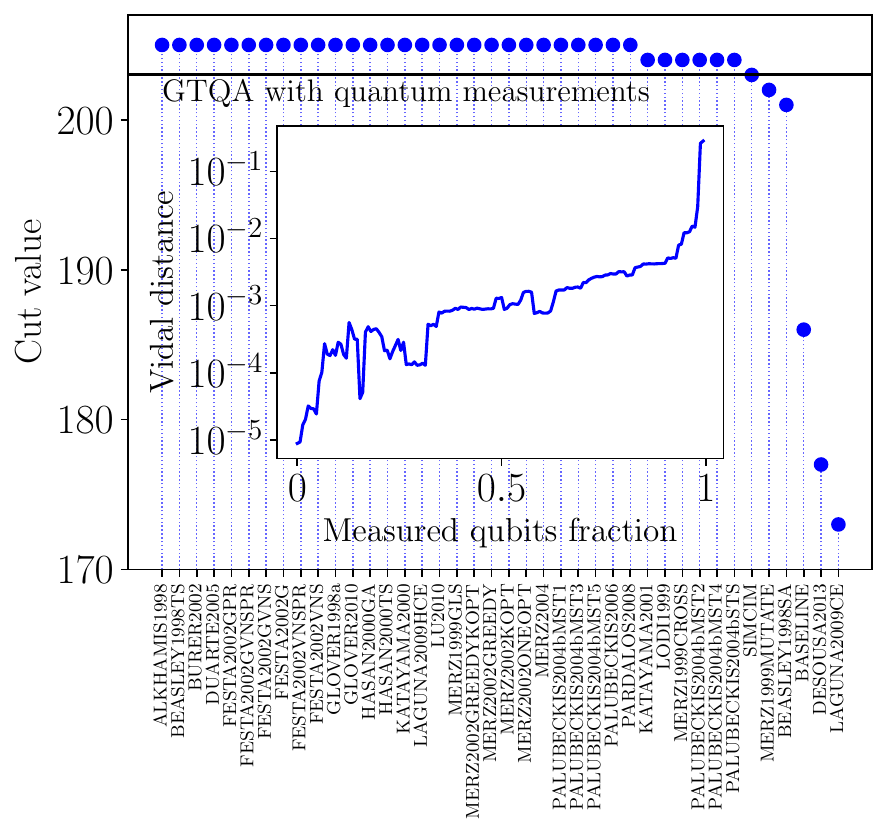}
	\caption{{\bf GTQA's performance on MaxCut.}
 The main figure compares the cut values found by GTQA with those found by $40$ MQLib heuristics and SimCIM, for a random 3-regular graph of $N=150$ vertices. Higher cut values correspond to better solutions. All the benchmark heuristics are ordered on the horizontal axis by decreasing cut value. For GTQA, we use a total annealing time $T = 40$ and bond dimension $\chi = 32$. The cut value found by GTQA is $203$ while the highest cut value among heuristics is $205$. This gives GTQA's solution an estimated approximation ratio $\alpha = 203/205 > 0.99$, remarkably. The inset shows the Vidal distance after regauging (during the measurement simulation) as a function of the fraction of qubits measured.
 The growth in Vidal distance is attributed to the fact that our sampling primitive's accuracy deteriorates close to the last sampled bit of the string. However, brute-force calculations indicate that the obtained MaxCut solution is unaffected by the observed final jump in distance (see main text).}
	\label{fig:maxcut_benchmark}
\end{figure}
\begin{equation}
	p_a(x_a) = \begin{cases}
		\varrho_a[0, 0], \ \text{if} \ x_a = 1,\\
		\varrho_a[1, 1], \ \text{if} \ x_a = -1;
	\end{cases}
\end{equation}
(\emph{ii}) to get a post measurement state, update the vertex $a$ tensor as follows
\begin{eqnarray}
	\label{eq:projection}
	\tilde{\Gamma}_{a}[i_a, \bfj_{\partial a}] = \sum_{i'_a}\pi_{x_a}[i_a, i'_a]\Gamma_{a}[i'_a, \bfj_{\partial a}],
\end{eqnarray}
where $\pi_{x_a}=\ket{x_a}\bra{x_a}$ is the orthogonal projection operator which corresponds to the measurement outcome $x_a$; (\emph{iii}) perform regauging of the tensor network and recompute messages. The step (\emph{iii}) is necessary since the update Eq.~\eqref{eq:projection} heavily breaks the gauge of the tensor network. Measurements sampling dramatically slows down the overall simulation since one needs to perform complete regauging after each qubit measurement: $N$ regaugings in total. It also affects the accuracy of the final result since BP does not always fully converge after a measurement sampling.

\subsection{Benchmarking against conventional MaxCut solvers}
\label{sec:sampling}
As in Sec.~\ref{prgrph:qubo_benchmarking}, we consider a random $3$-regular graph, but with a smaller number of qubits $N=150$ due to the increased computational complexity required to sample measurement outcomes. Here, we take $\chi=32$, $T=40$, $\delta t = 0.2$, and $s(t) = 1 - \frac{t}{T}$.
The maximal number of BP iterations is set to $K = 100$. Potential convergence failure requires an upper bound on the number of BP iterations in order not to go into an infinite loop.

We benchmark GTQA's performance on MaxCut against those of all 40 solvers from MQLib and SimCIM with fine-tuned hyperparameters. The comparison of cut values obtained by different methods is given in Fig.~\ref{fig:maxcut_benchmark}.
The highest cut value obtained by the best heuristic is $205$, while GTQA's cut value is $203$. This corresponds to an estimated approximation ratio for GTQA's solution of $\alpha=\frac{203}{205} > 0.99$, remarkably. This small discrepancy can be attributed to the mean-field error, because longer annealing time $T$ does not improve the accuracy. This is why we consider only a single value of $T$ here. 

As explained in the end of Sec.~\ref{subsec:measurements_sampling}, we rerun BP after each qubit measurement to re-gauge the TN back into the canonical form. Hence, the accuracy of the measurement-sampling primitive deteriorates due to non-perfect BP convergence as one approaches the last bit to sample in the bit-string. In fact, this is actually a well-known issue in BP-guided sampling~\cite{mezard2009information}.
In the inset of Fig.~\ref{fig:maxcut_benchmark}, we plot how this affects the Vidal distance of our TN during the sampling process. One can see that the Vidal distance increases towards the last measurement, with a particularly high jump in the very end. However, we verified that this jump does not affect the approximation ratio the produced MaxCut solution, interestingly. 
To see this, we also sampled 136 bits with our measurement-sampling primitive but chose the last 14 ones (where Vidal distance goes above $5\times 10^{-3}$) via brute-force maximization instead of sampling them. We did not observe any improvement in cut value over the case where all 150 bits are produced via BP-guided sampling. Moreover, we note also that BP convergence may be greatly improved by generalizing more advanced BP-like approaches, such as for instance tree-reweighed BP~\cite{wainwright2003tree, wainwright2005new}, to tensor networks.

\section{When do we expect GTQA to perform well?}
\label{sec:why_gtqa_works}

\begin{figure}[b!]
	\centering
	\includegraphics[width=1.\linewidth]{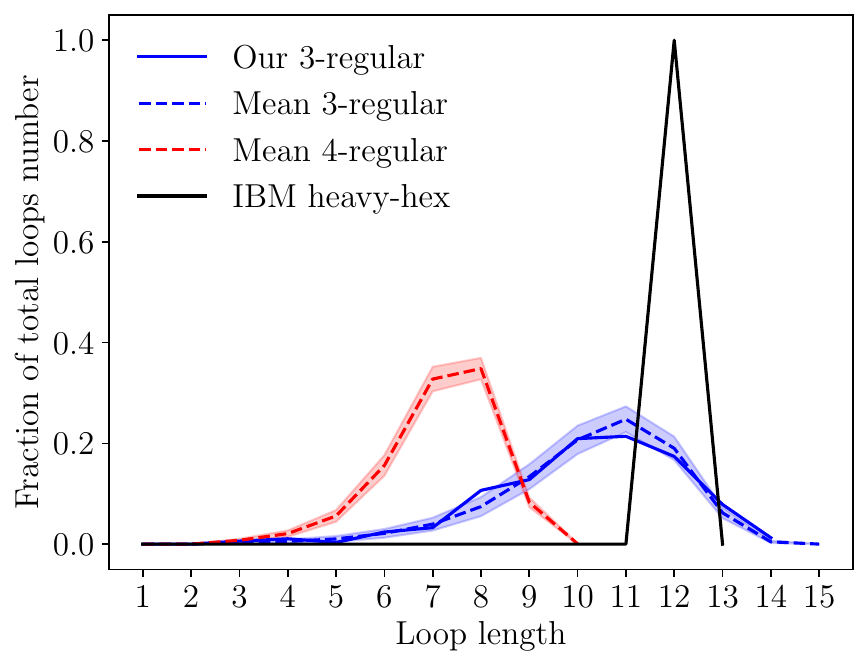}
	\caption{%
 {\bf Distribution of shortest loop lengths for different graphs.} The solid blue line corresponds to the 1000-vertex graph of Fig. \ref{fig:circuit_tn} (b). The dashed blue line is the average over $100$ random $3$-regular graphs with $1000$ vertices, and blue shaded area around it depicts the standard deviation. The dashed red line is the average over $100$ random $4$-regular graphs with $1000$ vertices, and the red shaded area around it depicts the standard deviation. The black curve corresponds to IBM's heavy-hex lattice with $127$ vertices~\cite{kim2023evidence}, where we do not count edges that have no loops passing through them.}
	\label{fig:loops_distribution}
\end{figure}
The BP algorithm works exactly on tree graphs and is known to show a high performance on problems where the underlying geometry is close to that of a tree graph, in the sense of there being loops but the loops being large. Thus, we expect optimization problems on any graph with only large loops to be efficiently simulable by GTQA. For a more quantitative intuition, we study the distribution of loop lengths for three exemplary types of graphs $G$: random $3$- and 4-regular graphs of 1000 vertices, and  IBM's quantum processor's heavy-hex lattice of $127$ vertices. The results are shown in Fig.~\ref{fig:loops_distribution}. For each edge in $G$, we compute the length of the shortest loop passing through the edge using the breadth-first search algorithm. As we can see in the plot, for the heavy-hex circuit, the shortest loop has length $12$~\cite{kim2023evidence}. This is relatively large compared to the other graphs studied, which is compatible with previous studies indicating that this geometry is simulable by BP based methods~\cite{tindall2023efficient, beguvsic2023fast}. In turn, random $3$-regular graphs are also dominated by relatively large loops. In contrast, random 4-regular graphs exhibit significantly shorter loops and are hence expected to be typically challenging for GTQA.

With the same reasoning, for random $d$-regular graphs of higher order, we expect the typical complexity of the classical simulation to grow with $d$, until a certain point where it is expected to drop again.
In fact, for very large $d$ (of the order of $N$, e.g.) we actually expect GTQA to exhibit a good performance again, even though the corresponding loops are very short. This is because BP-based tensor-network methods can be seen as a higher-order mean-field approximation~\cite{alkabetz2021tensor}, which is expected to work well for highly connected graphs.
This suggests that the search for quantum advantage should avoid problems with either too low or too high connectivity.

\section{Discussion}\label{sec:discussion}

In short, we introduced the {\it graph tensor-network quantum annealer} (GTQA), a powerful classical toolkit for high-precision simulations of Trotterized circuits of near-adiabatic evolutions on low-connectivity graphs. 
We showcased GTQA on random 3-regular graph QUBO problems of up to 1000 qubits, which required the simulation of circuits with up to $4.8$ million two-qubit gates in all-to-all connectivity. In addition, we took advantage of these simulations to study the dynamics of entanglement during a large-scale quantum annealing process.
Moreover, we introduced a measurement-simulation primitive for graph tensor-network states, based on BP-guided sampling. We applied this to solve highly-degenerate MaxCut instances on random 3-regular graphs of up to 150 vertices.
For both the QUBO and MaxCut cases, we demonstrated that GTQA's solutions are competitive with those of state-of-the-art classical solvers. This is remarkable because these instances correspond to highly-unstructured combinatorial optimization problems, known to be NP-hard in the worst case.

Apart from sampling measurement outcomes, the measurement-simulation primitive we introduced can also estimate the corresponding probabilities, which makes it interesting in its own right.
For example, our methods can be relevant to classical shadows~\cite{huang2020predicting} through projected tomographic ensembles~\cite{mcginley2023shadow,tran2023measuring,hu2023classical}. There, a scrambling (quench) unitary dynamics followed by a computational-basis measurement simulates a random measurement on a sub-system, which is used to learn properties of it.
The latter requires a suitable classical post-processing step that relies on projection probabilities similar to those we compute in Sec.~\ref{sec:MaxCut}. Hence, such task may be tackled with our primitive; and this has the potential to significantly broaden the scope of platforms amenable to 
projected-ensemble methods can be applied, including IBM's heavy-hex circuits. Another example is in linear cross-entropy benchmarking~\cite{boixo2018characterizing,arute2019quantum}, which quantifies the performance of noisy quantum computations.
This involves sampling measurement outcomes from the experimental device and numerically computing the probabilities of the observed outcomes.
With our method, one may extend cross-entropy benchmarking to QA or any other circuits simulatable with graph TNs.

On a different note, computing entanglement is important in QA~\cite{bauer2015entanglement,hauke2015probing}, as it can indicate whether the protocol converges or not. For non-degenerate problems, the final state is separable in the adiabatic limit while entangled for non-adiabatic protocols~\cite{hauke2015probing}. 
In fact, studying entanglement dynamics may improve annealing schedules beyond simple linear ones~\cite{roland2002quantum,yan2022analytical,mc2024towards}. Our solver can help explore the link between entanglement and quantum speedups~\cite{susa2017relation,albash2018adiabatic}, and clarify the benefits of QA over classical methods, especially for large problem sizes, where polynomial speedups might appear~\cite{rajak2023quantum,liu2015quantum}. Furthermore, our methods could be applied to investigate different quantum simulation schemes.
For example, here we use a first-order Trotter product formula~\cite{lloyd1996universal}. Higher-order Suzuki-Trotter terms require higher circuit cost per time step $\delta t$, but also allow for larger $\delta t$~\cite{heyl2019quantum,childs2021theory}.
Numerical studies of this trade-off has been limited to small systems or highly-structured problems, due to the difficulty of classical simulation.
With GTQA, one could explore higher-order expansions~\cite{heyl2019quantum,childs2021theory} as well as randomized product formulas~\cite{childs2019faster,PhysRevLett.123.070503} for large system sizes and deep circuits.

Finally, we stress that GTQA provides solutions for hard optimization problems which are competitive with those of state-of-the-art solvers. However, in contrast to fully-classical heuristics like simulated annealing, GTQA simulates the quantum state of the QA algorithm. As such, GTQA can potentially harness quantum effects such as entanglement and quantum tunneling to escape local minima~\cite{bauer2015entanglement, layden2023quantum}. Hence, it provides a powerful framework for quantum-inspired solvers that brings in novel approaches to tackle optimization problems. In this regard, we emphasize that we have not attempted to optimize the computational runtime of GTQA, but we expect it to be amenable to significant improvements. For example, message updates may be executed in parallel on a GPU~\cite{tomislav2024fast}. Additionally, one can trade %
state-simulation accuracy for runtime reduction (while still aiming at high-quality solutions of the optimization problem), for instance by combining GTQA  with imaginary time evolution~\cite{patra2024projected}.
To end up with, as for quantum optimization solvers, GTQA %
considerably raises the bar required
for experimental demonstrations of quantum speedups; highlighting the need of high-connectivity unstructured problems for a hopeful route towards quantum advantages in combinatorial optimizations.

\bibliography{bibliography.bib}

\begin{thebibliography}{113}%
\makeatletter
\providecommand \@ifxundefined [1]{%
 \@ifx{#1\undefined}
}%
\providecommand \@ifnum [1]{%
 \ifnum #1\expandafter \@firstoftwo
 \else \expandafter \@secondoftwo
 \fi
}%
\providecommand \@ifx [1]{%
 \ifx #1\expandafter \@firstoftwo
 \else \expandafter \@secondoftwo
 \fi
}%
\providecommand \natexlab [1]{#1}%
\providecommand \enquote  [1]{``#1''}%
\providecommand \bibnamefont  [1]{#1}%
\providecommand \bibfnamefont [1]{#1}%
\providecommand \citenamefont [1]{#1}%
\providecommand \href@noop [0]{\@secondoftwo}%
\providecommand \href [0]{\begingroup \@sanitize@url \@href}%
\providecommand \@href[1]{\@@startlink{#1}\@@href}%
\providecommand \@@href[1]{\endgroup#1\@@endlink}%
\providecommand \@sanitize@url [0]{\catcode `\\12\catcode `\$12\catcode `\&12\catcode `\#12\catcode `\^12\catcode `\_12\catcode `\%12\relax}%
\providecommand \@@startlink[1]{}%
\providecommand \@@endlink[0]{}%
\providecommand \url  [0]{\begingroup\@sanitize@url \@url }%
\providecommand \@url [1]{\endgroup\@href {#1}{\urlprefix }}%
\providecommand \urlprefix  [0]{URL }%
\providecommand \Eprint [0]{\href }%
\providecommand \doibase [0]{https://doi.org/}%
\providecommand \selectlanguage [0]{\@gobble}%
\providecommand \bibinfo  [0]{\@secondoftwo}%
\providecommand \bibfield  [0]{\@secondoftwo}%
\providecommand \translation [1]{[#1]}%
\providecommand \BibitemOpen [0]{}%
\providecommand \bibitemStop [0]{}%
\providecommand \bibitemNoStop [0]{.\EOS\space}%
\providecommand \EOS [0]{\spacefactor3000\relax}%
\providecommand \BibitemShut  [1]{\csname bibitem#1\endcsname}%
\let\auto@bib@innerbib\@empty
\bibitem [{\citenamefont {Abbas}\ \emph {et~al.}(2023)\citenamefont {Abbas}, \citenamefont {Ambainis}, \citenamefont {Augustino}, \citenamefont {B{\"a}rtschi}, \citenamefont {Buhrman}, \citenamefont {Coffrin}, \citenamefont {Cortiana}, \citenamefont {Dunjko}, \citenamefont {Egger}, \citenamefont {Elmegreen} \emph {et~al.}}]{abbas2023quantum}%
  \BibitemOpen
  \bibfield  {author} {\bibinfo {author} {\bibfnamefont {A.}~\bibnamefont {Abbas}}, \bibinfo {author} {\bibfnamefont {A.}~\bibnamefont {Ambainis}}, \bibinfo {author} {\bibfnamefont {B.}~\bibnamefont {Augustino}}, \bibinfo {author} {\bibfnamefont {A.}~\bibnamefont {B{\"a}rtschi}}, \bibinfo {author} {\bibfnamefont {H.}~\bibnamefont {Buhrman}}, \bibinfo {author} {\bibfnamefont {C.}~\bibnamefont {Coffrin}}, \bibinfo {author} {\bibfnamefont {G.}~\bibnamefont {Cortiana}}, \bibinfo {author} {\bibfnamefont {V.}~\bibnamefont {Dunjko}}, \bibinfo {author} {\bibfnamefont {D.~J.}\ \bibnamefont {Egger}}, \bibinfo {author} {\bibfnamefont {B.~G.}\ \bibnamefont {Elmegreen}}, \emph {et~al.},\ }\bibfield  {title} {\bibinfo {title} {Quantum optimization: Potential, challenges, and the path forward},\ }\href@noop {} {\bibfield  {journal} {\bibinfo  {journal} {arXiv preprint arXiv:2312.02279}\ } (\bibinfo {year} {2023})}\BibitemShut {NoStop}%
\bibitem [{\citenamefont {Kadowaki}\ and\ \citenamefont {Nishimori}(1998)}]{kadowaki1998quantum}%
  \BibitemOpen
  \bibfield  {author} {\bibinfo {author} {\bibfnamefont {T.}~\bibnamefont {Kadowaki}}\ and\ \bibinfo {author} {\bibfnamefont {H.}~\bibnamefont {Nishimori}},\ }\bibfield  {title} {\bibinfo {title} {Quantum annealing in the transverse ising model},\ }\href@noop {} {\bibfield  {journal} {\bibinfo  {journal} {Physical Review E}\ }\textbf {\bibinfo {volume} {58}},\ \bibinfo {pages} {5355} (\bibinfo {year} {1998})}\BibitemShut {NoStop}%
\bibitem [{\citenamefont {Farhi}\ \emph {et~al.}(2000)\citenamefont {Farhi}, \citenamefont {Goldstone}, \citenamefont {Gutmann},\ and\ \citenamefont {Sipser}}]{farhi2000quantum}%
  \BibitemOpen
  \bibfield  {author} {\bibinfo {author} {\bibfnamefont {E.}~\bibnamefont {Farhi}}, \bibinfo {author} {\bibfnamefont {J.}~\bibnamefont {Goldstone}}, \bibinfo {author} {\bibfnamefont {S.}~\bibnamefont {Gutmann}},\ and\ \bibinfo {author} {\bibfnamefont {M.}~\bibnamefont {Sipser}},\ }\bibfield  {title} {\bibinfo {title} {Quantum computation by adiabatic evolution},\ }\href@noop {} {\bibfield  {journal} {\bibinfo  {journal} {arXiv preprint quant-ph/0001106}\ } (\bibinfo {year} {2000})}\BibitemShut {NoStop}%
\bibitem [{\citenamefont {Hauke}\ \emph {et~al.}(2020)\citenamefont {Hauke}, \citenamefont {Katzgraber}, \citenamefont {Lechner}, \citenamefont {Nishimori},\ and\ \citenamefont {Oliver}}]{hauke2020perspectives}%
  \BibitemOpen
  \bibfield  {author} {\bibinfo {author} {\bibfnamefont {P.}~\bibnamefont {Hauke}}, \bibinfo {author} {\bibfnamefont {H.~G.}\ \bibnamefont {Katzgraber}}, \bibinfo {author} {\bibfnamefont {W.}~\bibnamefont {Lechner}}, \bibinfo {author} {\bibfnamefont {H.}~\bibnamefont {Nishimori}},\ and\ \bibinfo {author} {\bibfnamefont {W.~D.}\ \bibnamefont {Oliver}},\ }\bibfield  {title} {\bibinfo {title} {Perspectives of quantum annealing: Methods and implementations},\ }\href@noop {} {\bibfield  {journal} {\bibinfo  {journal} {Reports on Progress in Physics}\ }\textbf {\bibinfo {volume} {83}},\ \bibinfo {pages} {054401} (\bibinfo {year} {2020})}\BibitemShut {NoStop}%
\bibitem [{\citenamefont {Boixo}\ \emph {et~al.}(2014)\citenamefont {Boixo}, \citenamefont {R{\o}nnow}, \citenamefont {Isakov}, \citenamefont {Wang}, \citenamefont {Wecker}, \citenamefont {Lidar}, \citenamefont {Martinis},\ and\ \citenamefont {Troyer}}]{boixo2014evidence}%
  \BibitemOpen
  \bibfield  {author} {\bibinfo {author} {\bibfnamefont {S.}~\bibnamefont {Boixo}}, \bibinfo {author} {\bibfnamefont {T.~F.}\ \bibnamefont {R{\o}nnow}}, \bibinfo {author} {\bibfnamefont {S.~V.}\ \bibnamefont {Isakov}}, \bibinfo {author} {\bibfnamefont {Z.}~\bibnamefont {Wang}}, \bibinfo {author} {\bibfnamefont {D.}~\bibnamefont {Wecker}}, \bibinfo {author} {\bibfnamefont {D.~A.}\ \bibnamefont {Lidar}}, \bibinfo {author} {\bibfnamefont {J.~M.}\ \bibnamefont {Martinis}},\ and\ \bibinfo {author} {\bibfnamefont {M.}~\bibnamefont {Troyer}},\ }\bibfield  {title} {\bibinfo {title} {Evidence for quantum annealing with more than one hundred qubits},\ }\href@noop {} {\bibfield  {journal} {\bibinfo  {journal} {Nature physics}\ }\textbf {\bibinfo {volume} {10}},\ \bibinfo {pages} {218} (\bibinfo {year} {2014})}\BibitemShut {NoStop}%
\bibitem [{\citenamefont {Lanting}\ \emph {et~al.}(2014)\citenamefont {Lanting}, \citenamefont {Przybysz}, \citenamefont {Smirnov}, \citenamefont {Spedalieri}, \citenamefont {Amin}, \citenamefont {Berkley}, \citenamefont {Harris}, \citenamefont {Altomare}, \citenamefont {Boixo}, \citenamefont {Bunyk} \emph {et~al.}}]{lanting2014entanglement}%
  \BibitemOpen
  \bibfield  {author} {\bibinfo {author} {\bibfnamefont {T.}~\bibnamefont {Lanting}}, \bibinfo {author} {\bibfnamefont {A.~J.}\ \bibnamefont {Przybysz}}, \bibinfo {author} {\bibfnamefont {A.~Y.}\ \bibnamefont {Smirnov}}, \bibinfo {author} {\bibfnamefont {F.~M.}\ \bibnamefont {Spedalieri}}, \bibinfo {author} {\bibfnamefont {M.~H.}\ \bibnamefont {Amin}}, \bibinfo {author} {\bibfnamefont {A.~J.}\ \bibnamefont {Berkley}}, \bibinfo {author} {\bibfnamefont {R.}~\bibnamefont {Harris}}, \bibinfo {author} {\bibfnamefont {F.}~\bibnamefont {Altomare}}, \bibinfo {author} {\bibfnamefont {S.}~\bibnamefont {Boixo}}, \bibinfo {author} {\bibfnamefont {P.}~\bibnamefont {Bunyk}}, \emph {et~al.},\ }\bibfield  {title} {\bibinfo {title} {Entanglement in a quantum annealing processor},\ }\href@noop {} {\bibfield  {journal} {\bibinfo  {journal} {Physical Review X}\ }\textbf {\bibinfo {volume} {4}},\ \bibinfo {pages} {021041} (\bibinfo {year} {2014})}\BibitemShut {NoStop}%
\bibitem [{\citenamefont {Heim}\ \emph {et~al.}(2015)\citenamefont {Heim}, \citenamefont {Rønnow}, \citenamefont {Isakov},\ and\ \citenamefont {Troyer}}]{bettina2015quantum}%
  \BibitemOpen
  \bibfield  {author} {\bibinfo {author} {\bibfnamefont {B.}~\bibnamefont {Heim}}, \bibinfo {author} {\bibfnamefont {T.~F.}\ \bibnamefont {Rønnow}}, \bibinfo {author} {\bibfnamefont {S.~V.}\ \bibnamefont {Isakov}},\ and\ \bibinfo {author} {\bibfnamefont {M.}~\bibnamefont {Troyer}},\ }\bibfield  {title} {\bibinfo {title} {Quantum versus classical annealing of ising spin glasses},\ }\href {https://doi.org/10.1126/science.aaa4170} {\bibfield  {journal} {\bibinfo  {journal} {Science}\ }\textbf {\bibinfo {volume} {348}},\ \bibinfo {pages} {215} (\bibinfo {year} {2015})}\BibitemShut {NoStop}%
\bibitem [{\citenamefont {Farhi}\ \emph {et~al.}(2014)\citenamefont {Farhi}, \citenamefont {Goldstone},\ and\ \citenamefont {Gutmann}}]{farhi2014quantum}%
  \BibitemOpen
  \bibfield  {author} {\bibinfo {author} {\bibfnamefont {E.}~\bibnamefont {Farhi}}, \bibinfo {author} {\bibfnamefont {J.}~\bibnamefont {Goldstone}},\ and\ \bibinfo {author} {\bibfnamefont {S.}~\bibnamefont {Gutmann}},\ }\bibfield  {title} {\bibinfo {title} {A quantum approximate optimization algorithm},\ }\href@noop {} {\bibfield  {journal} {\bibinfo  {journal} {arXiv preprint arXiv:1411.4028}\ } (\bibinfo {year} {2014})}\BibitemShut {NoStop}%
\bibitem [{\citenamefont {Preskill}(2018)}]{preskill2018quantum}%
  \BibitemOpen
  \bibfield  {author} {\bibinfo {author} {\bibfnamefont {J.}~\bibnamefont {Preskill}},\ }\bibfield  {title} {\bibinfo {title} {Quantum computing in the nisq era and beyond},\ }\href@noop {} {\bibfield  {journal} {\bibinfo  {journal} {Quantum}\ }\textbf {\bibinfo {volume} {2}},\ \bibinfo {pages} {79} (\bibinfo {year} {2018})}\BibitemShut {NoStop}%
\bibitem [{\citenamefont {Zhou}\ \emph {et~al.}(2020{\natexlab{a}})\citenamefont {Zhou}, \citenamefont {Wang}, \citenamefont {Choi}, \citenamefont {Pichler},\ and\ \citenamefont {Lukin}}]{zhou2020quantum}%
  \BibitemOpen
  \bibfield  {author} {\bibinfo {author} {\bibfnamefont {L.}~\bibnamefont {Zhou}}, \bibinfo {author} {\bibfnamefont {S.-T.}\ \bibnamefont {Wang}}, \bibinfo {author} {\bibfnamefont {S.}~\bibnamefont {Choi}}, \bibinfo {author} {\bibfnamefont {H.}~\bibnamefont {Pichler}},\ and\ \bibinfo {author} {\bibfnamefont {M.~D.}\ \bibnamefont {Lukin}},\ }\bibfield  {title} {\bibinfo {title} {Quantum approximate optimization algorithm: Performance, mechanism, and implementation on near-term devices},\ }\href@noop {} {\bibfield  {journal} {\bibinfo  {journal} {Physical Review X}\ }\textbf {\bibinfo {volume} {10}},\ \bibinfo {pages} {021067} (\bibinfo {year} {2020}{\natexlab{a}})}\BibitemShut {NoStop}%
\bibitem [{\citenamefont {Bravyi}\ \emph {et~al.}(2020)\citenamefont {Bravyi}, \citenamefont {Kliesch}, \citenamefont {Koenig},\ and\ \citenamefont {Tang}}]{PhysRevLett.125.260505}%
  \BibitemOpen
  \bibfield  {author} {\bibinfo {author} {\bibfnamefont {S.}~\bibnamefont {Bravyi}}, \bibinfo {author} {\bibfnamefont {A.}~\bibnamefont {Kliesch}}, \bibinfo {author} {\bibfnamefont {R.}~\bibnamefont {Koenig}},\ and\ \bibinfo {author} {\bibfnamefont {E.}~\bibnamefont {Tang}},\ }\bibfield  {title} {\bibinfo {title} {Obstacles to variational quantum optimization from symmetry protection},\ }\href {https://doi.org/10.1103/PhysRevLett.125.260505} {\bibfield  {journal} {\bibinfo  {journal} {Phys. Rev. Lett.}\ }\textbf {\bibinfo {volume} {125}},\ \bibinfo {pages} {260505} (\bibinfo {year} {2020})}\BibitemShut {NoStop}%
\bibitem [{\citenamefont {Commander}(2009)}]{commander2009maximum}%
  \BibitemOpen
  \bibfield  {author} {\bibinfo {author} {\bibfnamefont {C.~W.}\ \bibnamefont {Commander}},\ }\bibfield  {title} {\bibinfo {title} {Maximum cut problem, max-cut.},\ }\href@noop {} {\bibfield  {journal} {\bibinfo  {journal} {Encyclopedia of Optimization}\ }\textbf {\bibinfo {volume} {2}} (\bibinfo {year} {2009})}\BibitemShut {NoStop}%
\bibitem [{\citenamefont {Berman}\ and\ \citenamefont {Karpinski}(1999)}]{berman1999some}%
  \BibitemOpen
  \bibfield  {author} {\bibinfo {author} {\bibfnamefont {P.}~\bibnamefont {Berman}}\ and\ \bibinfo {author} {\bibfnamefont {M.}~\bibnamefont {Karpinski}},\ }\bibfield  {title} {\bibinfo {title} {On some tighter inapproximability results},\ }in\ \href@noop {} {\emph {\bibinfo {booktitle} {Automata, Languages and Programming: 26th International Colloquium, ICALP’99 Prague, Czech Republic, July 11--15, 1999 Proceedings 26}}}\ (\bibinfo {organization} {Springer},\ \bibinfo {year} {1999})\ pp.\ \bibinfo {pages} {200--209}\BibitemShut {NoStop}%
\bibitem [{\citenamefont {Harrigan}\ \emph {et~al.}(2021)\citenamefont {Harrigan}, \citenamefont {Sung}, \citenamefont {Neeley}, \citenamefont {Satzinger}, \citenamefont {Arute}, \citenamefont {Arya}, \citenamefont {Atalaya}, \citenamefont {Bardin}, \citenamefont {Barends}, \citenamefont {Boixo} \emph {et~al.}}]{harrigan2021quantum}%
  \BibitemOpen
  \bibfield  {author} {\bibinfo {author} {\bibfnamefont {M.~P.}\ \bibnamefont {Harrigan}}, \bibinfo {author} {\bibfnamefont {K.~J.}\ \bibnamefont {Sung}}, \bibinfo {author} {\bibfnamefont {M.}~\bibnamefont {Neeley}}, \bibinfo {author} {\bibfnamefont {K.~J.}\ \bibnamefont {Satzinger}}, \bibinfo {author} {\bibfnamefont {F.}~\bibnamefont {Arute}}, \bibinfo {author} {\bibfnamefont {K.}~\bibnamefont {Arya}}, \bibinfo {author} {\bibfnamefont {J.}~\bibnamefont {Atalaya}}, \bibinfo {author} {\bibfnamefont {J.~C.}\ \bibnamefont {Bardin}}, \bibinfo {author} {\bibfnamefont {R.}~\bibnamefont {Barends}}, \bibinfo {author} {\bibfnamefont {S.}~\bibnamefont {Boixo}}, \emph {et~al.},\ }\bibfield  {title} {\bibinfo {title} {Quantum approximate optimization of non-planar graph problems on a planar superconducting processor},\ }\href@noop {} {\bibfield  {journal} {\bibinfo  {journal} {Nature Physics}\ }\textbf {\bibinfo {volume} {17}},\ \bibinfo {pages} {332} (\bibinfo {year} {2021})}\BibitemShut {NoStop}%
\bibitem [{\citenamefont {Sachdeva}\ \emph {et~al.}(2024)\citenamefont {Sachdeva}, \citenamefont {Harnett}, \citenamefont {Maity}, \citenamefont {Marsh}, \citenamefont {Wang}, \citenamefont {Winick}, \citenamefont {Dougherty}, \citenamefont {Canuto}, \citenamefont {Chong}, \citenamefont {Hush} \emph {et~al.}}]{sachdeva2024quantum}%
  \BibitemOpen
  \bibfield  {author} {\bibinfo {author} {\bibfnamefont {N.}~\bibnamefont {Sachdeva}}, \bibinfo {author} {\bibfnamefont {G.~S.}\ \bibnamefont {Harnett}}, \bibinfo {author} {\bibfnamefont {S.}~\bibnamefont {Maity}}, \bibinfo {author} {\bibfnamefont {S.}~\bibnamefont {Marsh}}, \bibinfo {author} {\bibfnamefont {Y.}~\bibnamefont {Wang}}, \bibinfo {author} {\bibfnamefont {A.}~\bibnamefont {Winick}}, \bibinfo {author} {\bibfnamefont {R.}~\bibnamefont {Dougherty}}, \bibinfo {author} {\bibfnamefont {D.}~\bibnamefont {Canuto}}, \bibinfo {author} {\bibfnamefont {Y.~Q.}\ \bibnamefont {Chong}}, \bibinfo {author} {\bibfnamefont {M.}~\bibnamefont {Hush}}, \emph {et~al.},\ }\bibfield  {title} {\bibinfo {title} {Quantum optimization using a 127-qubit gate-model ibm quantum computer can outperform quantum annealers for nontrivial binary optimization problems},\ }\href@noop {} {\bibfield  {journal} {\bibinfo  {journal} {arXiv preprint arXiv:2406.01743}\ } (\bibinfo {year} {2024})}\BibitemShut {NoStop}%
\bibitem [{\citenamefont {Shaydulin}\ and\ \citenamefont {Pistoia}(2023)}]{shaydulin2023qaoawith}%
  \BibitemOpen
  \bibfield  {author} {\bibinfo {author} {\bibfnamefont {R.}~\bibnamefont {Shaydulin}}\ and\ \bibinfo {author} {\bibfnamefont {M.}~\bibnamefont {Pistoia}},\ }\bibfield  {title} {\bibinfo {title} {Qaoawith $n\cdot p\geq 200$},\ }in\ \href@noop {} {\emph {\bibinfo {booktitle} {2023 IEEE International Conference on Quantum Computing and Engineering (QCE)}}},\ Vol.~\bibinfo {volume} {1}\ (\bibinfo {organization} {IEEE},\ \bibinfo {year} {2023})\ pp.\ \bibinfo {pages} {1074--1077}\BibitemShut {NoStop}%
\bibitem [{\citenamefont {DeCross}\ \emph {et~al.}(2023)\citenamefont {DeCross}, \citenamefont {Chertkov}, \citenamefont {Kohagen},\ and\ \citenamefont {Foss-Feig}}]{decross2023qubit}%
  \BibitemOpen
  \bibfield  {author} {\bibinfo {author} {\bibfnamefont {M.}~\bibnamefont {DeCross}}, \bibinfo {author} {\bibfnamefont {E.}~\bibnamefont {Chertkov}}, \bibinfo {author} {\bibfnamefont {M.}~\bibnamefont {Kohagen}},\ and\ \bibinfo {author} {\bibfnamefont {M.}~\bibnamefont {Foss-Feig}},\ }\bibfield  {title} {\bibinfo {title} {Qubit-reuse compilation with mid-circuit measurement and reset},\ }\href@noop {} {\bibfield  {journal} {\bibinfo  {journal} {Physical Review X}\ }\textbf {\bibinfo {volume} {13}},\ \bibinfo {pages} {041057} (\bibinfo {year} {2023})}\BibitemShut {NoStop}%
\bibitem [{\citenamefont {Moses}\ \emph {et~al.}(2023)\citenamefont {Moses}, \citenamefont {Baldwin}, \citenamefont {Allman}, \citenamefont {Ancona}, \citenamefont {Ascarrunz}, \citenamefont {Barnes}, \citenamefont {Bartolotta}, \citenamefont {Bjork}, \citenamefont {Blanchard}, \citenamefont {Bohn} \emph {et~al.}}]{moses2023race}%
  \BibitemOpen
  \bibfield  {author} {\bibinfo {author} {\bibfnamefont {S.~A.}\ \bibnamefont {Moses}}, \bibinfo {author} {\bibfnamefont {C.~H.}\ \bibnamefont {Baldwin}}, \bibinfo {author} {\bibfnamefont {M.~S.}\ \bibnamefont {Allman}}, \bibinfo {author} {\bibfnamefont {R.}~\bibnamefont {Ancona}}, \bibinfo {author} {\bibfnamefont {L.}~\bibnamefont {Ascarrunz}}, \bibinfo {author} {\bibfnamefont {C.}~\bibnamefont {Barnes}}, \bibinfo {author} {\bibfnamefont {J.}~\bibnamefont {Bartolotta}}, \bibinfo {author} {\bibfnamefont {B.}~\bibnamefont {Bjork}}, \bibinfo {author} {\bibfnamefont {P.}~\bibnamefont {Blanchard}}, \bibinfo {author} {\bibfnamefont {M.}~\bibnamefont {Bohn}}, \emph {et~al.},\ }\bibfield  {title} {\bibinfo {title} {A race-track trapped-ion quantum processor},\ }\href@noop {} {\bibfield  {journal} {\bibinfo  {journal} {Physical Review X}\ }\textbf {\bibinfo {volume} {13}},\ \bibinfo {pages} {041052} (\bibinfo {year} {2023})}\BibitemShut {NoStop}%
\bibitem [{\citenamefont {Dupont}\ \emph {et~al.}(2024)\citenamefont {Dupont}, \citenamefont {Sundar}, \citenamefont {Evert}, \citenamefont {Neira}, \citenamefont {Peng}, \citenamefont {Jeffrey},\ and\ \citenamefont {Hodson}}]{dupont2024quantum}%
  \BibitemOpen
  \bibfield  {author} {\bibinfo {author} {\bibfnamefont {M.}~\bibnamefont {Dupont}}, \bibinfo {author} {\bibfnamefont {B.}~\bibnamefont {Sundar}}, \bibinfo {author} {\bibfnamefont {B.}~\bibnamefont {Evert}}, \bibinfo {author} {\bibfnamefont {D.~E.~B.}\ \bibnamefont {Neira}}, \bibinfo {author} {\bibfnamefont {Z.}~\bibnamefont {Peng}}, \bibinfo {author} {\bibfnamefont {S.}~\bibnamefont {Jeffrey}},\ and\ \bibinfo {author} {\bibfnamefont {M.~J.}\ \bibnamefont {Hodson}},\ }\bibfield  {title} {\bibinfo {title} {Quantum optimization for the maximum cut problem on a superconducting quantum computer},\ }\href@noop {} {\bibfield  {journal} {\bibinfo  {journal} {arXiv preprint arXiv:2404.17579}\ } (\bibinfo {year} {2024})}\BibitemShut {NoStop}%
\bibitem [{\citenamefont {Ponce}\ \emph {et~al.}(2023)\citenamefont {Ponce}, \citenamefont {Herrman}, \citenamefont {Lotshaw}, \citenamefont {Powers}, \citenamefont {Siopsis}, \citenamefont {Humble},\ and\ \citenamefont {Ostrowski}}]{ponce2023graph}%
  \BibitemOpen
  \bibfield  {author} {\bibinfo {author} {\bibfnamefont {M.}~\bibnamefont {Ponce}}, \bibinfo {author} {\bibfnamefont {R.}~\bibnamefont {Herrman}}, \bibinfo {author} {\bibfnamefont {P.~C.}\ \bibnamefont {Lotshaw}}, \bibinfo {author} {\bibfnamefont {S.}~\bibnamefont {Powers}}, \bibinfo {author} {\bibfnamefont {G.}~\bibnamefont {Siopsis}}, \bibinfo {author} {\bibfnamefont {T.}~\bibnamefont {Humble}},\ and\ \bibinfo {author} {\bibfnamefont {J.}~\bibnamefont {Ostrowski}},\ }\bibfield  {title} {\bibinfo {title} {Graph decomposition techniques for solving combinatorial optimization problems with variational quantum algorithms},\ }\href@noop {} {\bibfield  {journal} {\bibinfo  {journal} {arXiv preprint arXiv:2306.00494}\ } (\bibinfo {year} {2023})}\BibitemShut {NoStop}%
\bibitem [{\citenamefont {Graham}\ \emph {et~al.}(2022)\citenamefont {Graham}, \citenamefont {Song}, \citenamefont {Scott}, \citenamefont {Poole}, \citenamefont {Phuttitarn}, \citenamefont {Jooya}, \citenamefont {Eichler}, \citenamefont {Jiang}, \citenamefont {Marra}, \citenamefont {Grinkemeyer} \emph {et~al.}}]{graham2022multi}%
  \BibitemOpen
  \bibfield  {author} {\bibinfo {author} {\bibfnamefont {T.}~\bibnamefont {Graham}}, \bibinfo {author} {\bibfnamefont {Y.}~\bibnamefont {Song}}, \bibinfo {author} {\bibfnamefont {J.}~\bibnamefont {Scott}}, \bibinfo {author} {\bibfnamefont {C.}~\bibnamefont {Poole}}, \bibinfo {author} {\bibfnamefont {L.}~\bibnamefont {Phuttitarn}}, \bibinfo {author} {\bibfnamefont {K.}~\bibnamefont {Jooya}}, \bibinfo {author} {\bibfnamefont {P.}~\bibnamefont {Eichler}}, \bibinfo {author} {\bibfnamefont {X.}~\bibnamefont {Jiang}}, \bibinfo {author} {\bibfnamefont {A.}~\bibnamefont {Marra}}, \bibinfo {author} {\bibfnamefont {B.}~\bibnamefont {Grinkemeyer}}, \emph {et~al.},\ }\bibfield  {title} {\bibinfo {title} {Multi-qubit entanglement and algorithms on a neutral-atom quantum computer},\ }\href@noop {} {\bibfield  {journal} {\bibinfo  {journal} {Nature}\ }\textbf {\bibinfo {volume} {604}},\ \bibinfo {pages} {457} (\bibinfo {year} {2022})}\BibitemShut {NoStop}%
\bibitem [{\citenamefont {Ebadi}\ \emph {et~al.}(2022)\citenamefont {Ebadi}, \citenamefont {Keesling}, \citenamefont {Cain}, \citenamefont {Wang}, \citenamefont {Levine}, \citenamefont {Bluvstein}, \citenamefont {Semeghini}, \citenamefont {Omran}, \citenamefont {Liu}, \citenamefont {Samajdar}, \citenamefont {Luo}, \citenamefont {Nash}, \citenamefont {Gao}, \citenamefont {Barak}, \citenamefont {Farhi}, \citenamefont {Sachdev}, \citenamefont {Gemelke}, \citenamefont {Zhou}, \citenamefont {Choi}, \citenamefont {Pichler}, \citenamefont {Wang}, \citenamefont {Greiner}, \citenamefont {Vuletić},\ and\ \citenamefont {Lukin}}]{doi:10.1126/science.abo6587}%
  \BibitemOpen
  \bibfield  {author} {\bibinfo {author} {\bibfnamefont {S.}~\bibnamefont {Ebadi}}, \bibinfo {author} {\bibfnamefont {A.}~\bibnamefont {Keesling}}, \bibinfo {author} {\bibfnamefont {M.}~\bibnamefont {Cain}}, \bibinfo {author} {\bibfnamefont {T.~T.}\ \bibnamefont {Wang}}, \bibinfo {author} {\bibfnamefont {H.}~\bibnamefont {Levine}}, \bibinfo {author} {\bibfnamefont {D.}~\bibnamefont {Bluvstein}}, \bibinfo {author} {\bibfnamefont {G.}~\bibnamefont {Semeghini}}, \bibinfo {author} {\bibfnamefont {A.}~\bibnamefont {Omran}}, \bibinfo {author} {\bibfnamefont {J.-G.}\ \bibnamefont {Liu}}, \bibinfo {author} {\bibfnamefont {R.}~\bibnamefont {Samajdar}}, \bibinfo {author} {\bibfnamefont {X.-Z.}\ \bibnamefont {Luo}}, \bibinfo {author} {\bibfnamefont {B.}~\bibnamefont {Nash}}, \bibinfo {author} {\bibfnamefont {X.}~\bibnamefont {Gao}}, \bibinfo {author} {\bibfnamefont {B.}~\bibnamefont {Barak}}, \bibinfo {author} {\bibfnamefont {E.}~\bibnamefont {Farhi}}, \bibinfo {author} {\bibfnamefont {S.}~\bibnamefont {Sachdev}},
  \bibinfo {author} {\bibfnamefont {N.}~\bibnamefont {Gemelke}}, \bibinfo {author} {\bibfnamefont {L.}~\bibnamefont {Zhou}}, \bibinfo {author} {\bibfnamefont {S.}~\bibnamefont {Choi}}, \bibinfo {author} {\bibfnamefont {H.}~\bibnamefont {Pichler}}, \bibinfo {author} {\bibfnamefont {S.-T.}\ \bibnamefont {Wang}}, \bibinfo {author} {\bibfnamefont {M.}~\bibnamefont {Greiner}}, \bibinfo {author} {\bibfnamefont {V.}~\bibnamefont {Vuletić}},\ and\ \bibinfo {author} {\bibfnamefont {M.~D.}\ \bibnamefont {Lukin}},\ }\bibfield  {title} {\bibinfo {title} {Quantum optimization of maximum independent set using rydberg atom arrays},\ }\href {https://doi.org/10.1126/science.abo6587} {\bibfield  {journal} {\bibinfo  {journal} {Science}\ }\textbf {\bibinfo {volume} {376}},\ \bibinfo {pages} {1209} (\bibinfo {year} {2022})}\BibitemShut {NoStop}%
\bibitem [{\citenamefont {Wurtz}\ and\ \citenamefont {Love}(2021)}]{wurtz2021maxcut}%
  \BibitemOpen
  \bibfield  {author} {\bibinfo {author} {\bibfnamefont {J.}~\bibnamefont {Wurtz}}\ and\ \bibinfo {author} {\bibfnamefont {P.}~\bibnamefont {Love}},\ }\bibfield  {title} {\bibinfo {title} {Maxcut quantum approximate optimization algorithm performance guarantees for p> 1},\ }\href@noop {} {\bibfield  {journal} {\bibinfo  {journal} {Physical Review A}\ }\textbf {\bibinfo {volume} {103}},\ \bibinfo {pages} {042612} (\bibinfo {year} {2021})}\BibitemShut {NoStop}%
\bibitem [{\citenamefont {Wurtz}\ and\ \citenamefont {Lykov}(2021)}]{wurtz2021fixed}%
  \BibitemOpen
  \bibfield  {author} {\bibinfo {author} {\bibfnamefont {J.}~\bibnamefont {Wurtz}}\ and\ \bibinfo {author} {\bibfnamefont {D.}~\bibnamefont {Lykov}},\ }\bibfield  {title} {\bibinfo {title} {Fixed-angle conjectures for the quantum approximate optimization algorithm on regular maxcut graphs},\ }\href@noop {} {\bibfield  {journal} {\bibinfo  {journal} {Physical Review A}\ }\textbf {\bibinfo {volume} {104}},\ \bibinfo {pages} {052419} (\bibinfo {year} {2021})}\BibitemShut {NoStop}%
\bibitem [{\citenamefont {Stilck~Fran{\c{c}}a}\ and\ \citenamefont {Garcia-Patron}(2021)}]{stilck2021limitations}%
  \BibitemOpen
  \bibfield  {author} {\bibinfo {author} {\bibfnamefont {D.}~\bibnamefont {Stilck~Fran{\c{c}}a}}\ and\ \bibinfo {author} {\bibfnamefont {R.}~\bibnamefont {Garcia-Patron}},\ }\bibfield  {title} {\bibinfo {title} {Limitations of optimization algorithms on noisy quantum devices},\ }\href@noop {} {\bibfield  {journal} {\bibinfo  {journal} {Nature Physics}\ }\textbf {\bibinfo {volume} {17}},\ \bibinfo {pages} {1221} (\bibinfo {year} {2021})}\BibitemShut {NoStop}%
\bibitem [{\citenamefont {Liu}\ \emph {et~al.}(2015)\citenamefont {Liu}, \citenamefont {Polkovnikov},\ and\ \citenamefont {Sandvik}}]{liu2015quantum}%
  \BibitemOpen
  \bibfield  {author} {\bibinfo {author} {\bibfnamefont {C.-W.}\ \bibnamefont {Liu}}, \bibinfo {author} {\bibfnamefont {A.}~\bibnamefont {Polkovnikov}},\ and\ \bibinfo {author} {\bibfnamefont {A.~W.}\ \bibnamefont {Sandvik}},\ }\bibfield  {title} {\bibinfo {title} {Quantum versus classical annealing: insights from scaling theory and results for spin glasses on 3-regular graphs},\ }\href@noop {} {\bibfield  {journal} {\bibinfo  {journal} {Physical Review Letters}\ }\textbf {\bibinfo {volume} {114}},\ \bibinfo {pages} {147203} (\bibinfo {year} {2015})}\BibitemShut {NoStop}%
\bibitem [{\citenamefont {Guerreschi}\ and\ \citenamefont {Matsuura}(2019)}]{guerreschi2019qaoa}%
  \BibitemOpen
  \bibfield  {author} {\bibinfo {author} {\bibfnamefont {G.~G.}\ \bibnamefont {Guerreschi}}\ and\ \bibinfo {author} {\bibfnamefont {A.~Y.}\ \bibnamefont {Matsuura}},\ }\bibfield  {title} {\bibinfo {title} {Qaoa for max-cut requires hundreds of qubits for quantum speed-up},\ }\href@noop {} {\bibfield  {journal} {\bibinfo  {journal} {Scientific reports}\ }\textbf {\bibinfo {volume} {9}},\ \bibinfo {pages} {6903} (\bibinfo {year} {2019})}\BibitemShut {NoStop}%
\bibitem [{\citenamefont {Dunning}\ \emph {et~al.}(2018)\citenamefont {Dunning}, \citenamefont {Gupta},\ and\ \citenamefont {Silberholz}}]{dunning2018what}%
  \BibitemOpen
  \bibfield  {author} {\bibinfo {author} {\bibfnamefont {I.}~\bibnamefont {Dunning}}, \bibinfo {author} {\bibfnamefont {S.}~\bibnamefont {Gupta}},\ and\ \bibinfo {author} {\bibfnamefont {J.}~\bibnamefont {Silberholz}},\ }\bibfield  {title} {\bibinfo {title} {What works best when? a systematic evaluation of heuristics for max-cut and {QUBO}},\ }\href@noop {} {\bibfield  {journal} {\bibinfo  {journal} {{INFORMS} Journal on Computing}\ }\textbf {\bibinfo {volume} {30}} (\bibinfo {year} {2018})}\BibitemShut {NoStop}%
\bibitem [{\citenamefont {McMahon}\ \emph {et~al.}(2016)\citenamefont {McMahon}, \citenamefont {Marandi}, \citenamefont {Haribara}, \citenamefont {Hamerly}, \citenamefont {Langrock}, \citenamefont {Tamate}, \citenamefont {Inagaki}, \citenamefont {Takesue}, \citenamefont {Utsunomiya}, \citenamefont {Aihara} \emph {et~al.}}]{mcmahon2016fully}%
  \BibitemOpen
  \bibfield  {author} {\bibinfo {author} {\bibfnamefont {P.~L.}\ \bibnamefont {McMahon}}, \bibinfo {author} {\bibfnamefont {A.}~\bibnamefont {Marandi}}, \bibinfo {author} {\bibfnamefont {Y.}~\bibnamefont {Haribara}}, \bibinfo {author} {\bibfnamefont {R.}~\bibnamefont {Hamerly}}, \bibinfo {author} {\bibfnamefont {C.}~\bibnamefont {Langrock}}, \bibinfo {author} {\bibfnamefont {S.}~\bibnamefont {Tamate}}, \bibinfo {author} {\bibfnamefont {T.}~\bibnamefont {Inagaki}}, \bibinfo {author} {\bibfnamefont {H.}~\bibnamefont {Takesue}}, \bibinfo {author} {\bibfnamefont {S.}~\bibnamefont {Utsunomiya}}, \bibinfo {author} {\bibfnamefont {K.}~\bibnamefont {Aihara}}, \emph {et~al.},\ }\bibfield  {title} {\bibinfo {title} {A fully programmable 100-spin coherent ising machine with all-to-all connections},\ }\href@noop {} {\bibfield  {journal} {\bibinfo  {journal} {Science}\ }\textbf {\bibinfo {volume} {354}},\ \bibinfo {pages} {614} (\bibinfo {year} {2016})}\BibitemShut {NoStop}%
\bibitem [{\citenamefont {Honjo}\ \emph {et~al.}(2021)\citenamefont {Honjo}, \citenamefont {Sonobe}, \citenamefont {Inaba}, \citenamefont {Inagaki}, \citenamefont {Ikuta}, \citenamefont {Yamada}, \citenamefont {Kazama}, \citenamefont {Enbutsu}, \citenamefont {Umeki}, \citenamefont {Kasahara} \emph {et~al.}}]{honjo2021100}%
  \BibitemOpen
  \bibfield  {author} {\bibinfo {author} {\bibfnamefont {T.}~\bibnamefont {Honjo}}, \bibinfo {author} {\bibfnamefont {T.}~\bibnamefont {Sonobe}}, \bibinfo {author} {\bibfnamefont {K.}~\bibnamefont {Inaba}}, \bibinfo {author} {\bibfnamefont {T.}~\bibnamefont {Inagaki}}, \bibinfo {author} {\bibfnamefont {T.}~\bibnamefont {Ikuta}}, \bibinfo {author} {\bibfnamefont {Y.}~\bibnamefont {Yamada}}, \bibinfo {author} {\bibfnamefont {T.}~\bibnamefont {Kazama}}, \bibinfo {author} {\bibfnamefont {K.}~\bibnamefont {Enbutsu}}, \bibinfo {author} {\bibfnamefont {T.}~\bibnamefont {Umeki}}, \bibinfo {author} {\bibfnamefont {R.}~\bibnamefont {Kasahara}}, \emph {et~al.},\ }\bibfield  {title} {\bibinfo {title} {100,000-spin coherent ising machine},\ }\href@noop {} {\bibfield  {journal} {\bibinfo  {journal} {Science advances}\ }\textbf {\bibinfo {volume} {7}},\ \bibinfo {pages} {eabh0952} (\bibinfo {year} {2021})}\BibitemShut {NoStop}%
\bibitem [{\citenamefont {Mohseni}\ \emph {et~al.}(2022)\citenamefont {Mohseni}, \citenamefont {McMahon},\ and\ \citenamefont {Byrnes}}]{mohseni2022ising}%
  \BibitemOpen
  \bibfield  {author} {\bibinfo {author} {\bibfnamefont {N.}~\bibnamefont {Mohseni}}, \bibinfo {author} {\bibfnamefont {P.~L.}\ \bibnamefont {McMahon}},\ and\ \bibinfo {author} {\bibfnamefont {T.}~\bibnamefont {Byrnes}},\ }\bibfield  {title} {\bibinfo {title} {Ising machines as hardware solvers of combinatorial optimization problems},\ }\href@noop {} {\bibfield  {journal} {\bibinfo  {journal} {Nature Reviews Physics}\ }\textbf {\bibinfo {volume} {4}},\ \bibinfo {pages} {363} (\bibinfo {year} {2022})}\BibitemShut {NoStop}%
\bibitem [{\citenamefont {Goto}\ \emph {et~al.}(2019)\citenamefont {Goto}, \citenamefont {Tatsumura},\ and\ \citenamefont {Dixon}}]{goto2019combinatorial}%
  \BibitemOpen
  \bibfield  {author} {\bibinfo {author} {\bibfnamefont {H.}~\bibnamefont {Goto}}, \bibinfo {author} {\bibfnamefont {K.}~\bibnamefont {Tatsumura}},\ and\ \bibinfo {author} {\bibfnamefont {A.~R.}\ \bibnamefont {Dixon}},\ }\bibfield  {title} {\bibinfo {title} {Combinatorial optimization by simulating adiabatic bifurcations in nonlinear hamiltonian systems},\ }\href@noop {} {\bibfield  {journal} {\bibinfo  {journal} {Science advances}\ }\textbf {\bibinfo {volume} {5}},\ \bibinfo {pages} {eaav2372} (\bibinfo {year} {2019})}\BibitemShut {NoStop}%
\bibitem [{\citenamefont {Goto}\ \emph {et~al.}(2021)\citenamefont {Goto}, \citenamefont {Endo}, \citenamefont {Suzuki}, \citenamefont {Sakai}, \citenamefont {Kanao}, \citenamefont {Hamakawa}, \citenamefont {Hidaka}, \citenamefont {Yamasaki},\ and\ \citenamefont {Tatsumura}}]{goto2021high}%
  \BibitemOpen
  \bibfield  {author} {\bibinfo {author} {\bibfnamefont {H.}~\bibnamefont {Goto}}, \bibinfo {author} {\bibfnamefont {K.}~\bibnamefont {Endo}}, \bibinfo {author} {\bibfnamefont {M.}~\bibnamefont {Suzuki}}, \bibinfo {author} {\bibfnamefont {Y.}~\bibnamefont {Sakai}}, \bibinfo {author} {\bibfnamefont {T.}~\bibnamefont {Kanao}}, \bibinfo {author} {\bibfnamefont {Y.}~\bibnamefont {Hamakawa}}, \bibinfo {author} {\bibfnamefont {R.}~\bibnamefont {Hidaka}}, \bibinfo {author} {\bibfnamefont {M.}~\bibnamefont {Yamasaki}},\ and\ \bibinfo {author} {\bibfnamefont {K.}~\bibnamefont {Tatsumura}},\ }\bibfield  {title} {\bibinfo {title} {High-performance combinatorial optimization based on classical mechanics},\ }\href@noop {} {\bibfield  {journal} {\bibinfo  {journal} {Science Advances}\ }\textbf {\bibinfo {volume} {7}},\ \bibinfo {pages} {eabe7953} (\bibinfo {year} {2021})}\BibitemShut {NoStop}%
\bibitem [{\citenamefont {Hibat-Allah}\ \emph {et~al.}(2021)\citenamefont {Hibat-Allah}, \citenamefont {Inack}, \citenamefont {Wiersema}, \citenamefont {Melko},\ and\ \citenamefont {Carrasquilla}}]{hibat2021variational}%
  \BibitemOpen
  \bibfield  {author} {\bibinfo {author} {\bibfnamefont {M.}~\bibnamefont {Hibat-Allah}}, \bibinfo {author} {\bibfnamefont {E.~M.}\ \bibnamefont {Inack}}, \bibinfo {author} {\bibfnamefont {R.}~\bibnamefont {Wiersema}}, \bibinfo {author} {\bibfnamefont {R.~G.}\ \bibnamefont {Melko}},\ and\ \bibinfo {author} {\bibfnamefont {J.}~\bibnamefont {Carrasquilla}},\ }\bibfield  {title} {\bibinfo {title} {Variational neural annealing},\ }\href@noop {} {\bibfield  {journal} {\bibinfo  {journal} {Nature Machine Intelligence}\ }\textbf {\bibinfo {volume} {3}},\ \bibinfo {pages} {952} (\bibinfo {year} {2021})}\BibitemShut {NoStop}%
\bibitem [{\citenamefont {Lami}\ \emph {et~al.}(2023)\citenamefont {Lami}, \citenamefont {Torta}, \citenamefont {Santoro},\ and\ \citenamefont {Collura}}]{lami2023quantum}%
  \BibitemOpen
  \bibfield  {author} {\bibinfo {author} {\bibfnamefont {G.}~\bibnamefont {Lami}}, \bibinfo {author} {\bibfnamefont {P.}~\bibnamefont {Torta}}, \bibinfo {author} {\bibfnamefont {G.~E.}\ \bibnamefont {Santoro}},\ and\ \bibinfo {author} {\bibfnamefont {M.}~\bibnamefont {Collura}},\ }\bibfield  {title} {\bibinfo {title} {Quantum annealing for neural network optimization problems: A new approach via tensor network simulations},\ }\href@noop {} {\bibfield  {journal} {\bibinfo  {journal} {SciPost Physics}\ }\textbf {\bibinfo {volume} {14}},\ \bibinfo {pages} {117} (\bibinfo {year} {2023})}\BibitemShut {NoStop}%
\bibitem [{\citenamefont {Tindall}\ \emph {et~al.}(2023)\citenamefont {Tindall}, \citenamefont {Fishman}, \citenamefont {Stoudenmire},\ and\ \citenamefont {Sels}}]{tindall2023efficient}%
  \BibitemOpen
  \bibfield  {author} {\bibinfo {author} {\bibfnamefont {J.}~\bibnamefont {Tindall}}, \bibinfo {author} {\bibfnamefont {M.}~\bibnamefont {Fishman}}, \bibinfo {author} {\bibfnamefont {M.}~\bibnamefont {Stoudenmire}},\ and\ \bibinfo {author} {\bibfnamefont {D.}~\bibnamefont {Sels}},\ }\bibfield  {title} {\bibinfo {title} {Efficient tensor network simulation of ibm's kicked ising experiment},\ }\href@noop {} {\bibfield  {journal} {\bibinfo  {journal} {arXiv preprint arXiv:2306.14887}\ } (\bibinfo {year} {2023})}\BibitemShut {NoStop}%
\bibitem [{\citenamefont {Pan}\ and\ \citenamefont {Zhang}(2022)}]{pan2022simulation}%
  \BibitemOpen
  \bibfield  {author} {\bibinfo {author} {\bibfnamefont {F.}~\bibnamefont {Pan}}\ and\ \bibinfo {author} {\bibfnamefont {P.}~\bibnamefont {Zhang}},\ }\bibfield  {title} {\bibinfo {title} {Simulation of quantum circuits using the big-batch tensor network method},\ }\href@noop {} {\bibfield  {journal} {\bibinfo  {journal} {Physical Review Letters}\ }\textbf {\bibinfo {volume} {128}},\ \bibinfo {pages} {030501} (\bibinfo {year} {2022})}\BibitemShut {NoStop}%
\bibitem [{\citenamefont {Perez-Garcia}\ \emph {et~al.}(2007)\citenamefont {Perez-Garcia}, \citenamefont {Verstraete}, \citenamefont {Wolf},\ and\ \citenamefont {Cirac}}]{garcia2007matrix}%
  \BibitemOpen
  \bibfield  {author} {\bibinfo {author} {\bibfnamefont {D.}~\bibnamefont {Perez-Garcia}}, \bibinfo {author} {\bibfnamefont {F.}~\bibnamefont {Verstraete}}, \bibinfo {author} {\bibfnamefont {M.~M.}\ \bibnamefont {Wolf}},\ and\ \bibinfo {author} {\bibfnamefont {J.~I.}\ \bibnamefont {Cirac}},\ }\bibfield  {title} {\bibinfo {title} {Matrix product state representations},\ }\href@noop {} {\bibfield  {journal} {\bibinfo  {journal} {Quantum Info. Comput.}\ }\textbf {\bibinfo {volume} {7}},\ \bibinfo {pages} {401–430} (\bibinfo {year} {2007})}\BibitemShut {NoStop}%
\bibitem [{\citenamefont {Verstraete}\ \emph {et~al.}(2008)\citenamefont {Verstraete}, \citenamefont {Murg},\ and\ \citenamefont {Cirac}}]{verstraete2008matrix}%
  \BibitemOpen
  \bibfield  {author} {\bibinfo {author} {\bibfnamefont {F.}~\bibnamefont {Verstraete}}, \bibinfo {author} {\bibfnamefont {V.}~\bibnamefont {Murg}},\ and\ \bibinfo {author} {\bibfnamefont {J.~I.}\ \bibnamefont {Cirac}},\ }\bibfield  {title} {\bibinfo {title} {Matrix product states, projected entangled pair states, and variational renormalization group methods for quantum spin systems},\ }\href@noop {} {\bibfield  {journal} {\bibinfo  {journal} {Advances in physics}\ }\textbf {\bibinfo {volume} {57}},\ \bibinfo {pages} {143} (\bibinfo {year} {2008})}\BibitemShut {NoStop}%
\bibitem [{\citenamefont {Schollw{\"o}ck}(2011)}]{schollwock2011density}%
  \BibitemOpen
  \bibfield  {author} {\bibinfo {author} {\bibfnamefont {U.}~\bibnamefont {Schollw{\"o}ck}},\ }\bibfield  {title} {\bibinfo {title} {The density-matrix renormalization group in the age of matrix product states},\ }\href@noop {} {\bibfield  {journal} {\bibinfo  {journal} {Annals of physics}\ }\textbf {\bibinfo {volume} {326}},\ \bibinfo {pages} {96} (\bibinfo {year} {2011})}\BibitemShut {NoStop}%
\bibitem [{\citenamefont {Vidal}(2003)}]{vidal2003efficient}%
  \BibitemOpen
  \bibfield  {author} {\bibinfo {author} {\bibfnamefont {G.}~\bibnamefont {Vidal}},\ }\bibfield  {title} {\bibinfo {title} {Efficient classical simulation of slightly entangled quantum computations},\ }\href@noop {} {\bibfield  {journal} {\bibinfo  {journal} {Physical review letters}\ }\textbf {\bibinfo {volume} {91}},\ \bibinfo {pages} {147902} (\bibinfo {year} {2003})}\BibitemShut {NoStop}%
\bibitem [{\citenamefont {Vidal}(2004)}]{vidal2004efficient}%
  \BibitemOpen
  \bibfield  {author} {\bibinfo {author} {\bibfnamefont {G.}~\bibnamefont {Vidal}},\ }\bibfield  {title} {\bibinfo {title} {Efficient simulation of one-dimensional quantum many-body systems},\ }\href@noop {} {\bibfield  {journal} {\bibinfo  {journal} {Physical review letters}\ }\textbf {\bibinfo {volume} {93}},\ \bibinfo {pages} {040502} (\bibinfo {year} {2004})}\BibitemShut {NoStop}%
\bibitem [{\citenamefont {Haegeman}\ \emph {et~al.}(2011)\citenamefont {Haegeman}, \citenamefont {Cirac}, \citenamefont {Osborne}, \citenamefont {Pi{\v{z}}orn}, \citenamefont {Verschelde},\ and\ \citenamefont {Verstraete}}]{haegeman2011time}%
  \BibitemOpen
  \bibfield  {author} {\bibinfo {author} {\bibfnamefont {J.}~\bibnamefont {Haegeman}}, \bibinfo {author} {\bibfnamefont {J.~I.}\ \bibnamefont {Cirac}}, \bibinfo {author} {\bibfnamefont {T.~J.}\ \bibnamefont {Osborne}}, \bibinfo {author} {\bibfnamefont {I.}~\bibnamefont {Pi{\v{z}}orn}}, \bibinfo {author} {\bibfnamefont {H.}~\bibnamefont {Verschelde}},\ and\ \bibinfo {author} {\bibfnamefont {F.}~\bibnamefont {Verstraete}},\ }\bibfield  {title} {\bibinfo {title} {Time-dependent variational principle for quantum lattices},\ }\href@noop {} {\bibfield  {journal} {\bibinfo  {journal} {Physical review letters}\ }\textbf {\bibinfo {volume} {107}},\ \bibinfo {pages} {070601} (\bibinfo {year} {2011})}\BibitemShut {NoStop}%
\bibitem [{\citenamefont {Haegeman}\ \emph {et~al.}(2016)\citenamefont {Haegeman}, \citenamefont {Lubich}, \citenamefont {Oseledets}, \citenamefont {Vandereycken},\ and\ \citenamefont {Verstraete}}]{haegeman2016unifying}%
  \BibitemOpen
  \bibfield  {author} {\bibinfo {author} {\bibfnamefont {J.}~\bibnamefont {Haegeman}}, \bibinfo {author} {\bibfnamefont {C.}~\bibnamefont {Lubich}}, \bibinfo {author} {\bibfnamefont {I.}~\bibnamefont {Oseledets}}, \bibinfo {author} {\bibfnamefont {B.}~\bibnamefont {Vandereycken}},\ and\ \bibinfo {author} {\bibfnamefont {F.}~\bibnamefont {Verstraete}},\ }\bibfield  {title} {\bibinfo {title} {Unifying time evolution and optimization with matrix product states},\ }\href {https://doi.org/10.1103/PhysRevB.94.165116} {\bibfield  {journal} {\bibinfo  {journal} {Phys. Rev. B}\ }\textbf {\bibinfo {volume} {94}},\ \bibinfo {pages} {165116} (\bibinfo {year} {2016})}\BibitemShut {NoStop}%
\bibitem [{\citenamefont {Kloss}\ \emph {et~al.}(2018)\citenamefont {Kloss}, \citenamefont {Lev},\ and\ \citenamefont {Reichman}}]{kloss2018time}%
  \BibitemOpen
  \bibfield  {author} {\bibinfo {author} {\bibfnamefont {B.}~\bibnamefont {Kloss}}, \bibinfo {author} {\bibfnamefont {Y.~B.}\ \bibnamefont {Lev}},\ and\ \bibinfo {author} {\bibfnamefont {D.}~\bibnamefont {Reichman}},\ }\bibfield  {title} {\bibinfo {title} {Time-dependent variational principle in matrix-product state manifolds: Pitfalls and potential},\ }\href {https://doi.org/10.1103/PhysRevB.97.024307} {\bibfield  {journal} {\bibinfo  {journal} {Phys. Rev. B}\ }\textbf {\bibinfo {volume} {97}},\ \bibinfo {pages} {024307} (\bibinfo {year} {2018})}\BibitemShut {NoStop}%
\bibitem [{\citenamefont {Verstraete}\ and\ \citenamefont {Cirac}(2004)}]{verstraete2004renormalization}%
  \BibitemOpen
  \bibfield  {author} {\bibinfo {author} {\bibfnamefont {F.}~\bibnamefont {Verstraete}}\ and\ \bibinfo {author} {\bibfnamefont {J.~I.}\ \bibnamefont {Cirac}},\ }\bibfield  {title} {\bibinfo {title} {Renormalization algorithms for quantum-many body systems in two and higher dimensions},\ }\href@noop {} {\bibfield  {journal} {\bibinfo  {journal} {arXiv preprint cond-mat/0407066}\ } (\bibinfo {year} {2004})}\BibitemShut {NoStop}%
\bibitem [{\citenamefont {Jordan}\ \emph {et~al.}(2008)\citenamefont {Jordan}, \citenamefont {Or\'us}, \citenamefont {Vidal}, \citenamefont {Verstraete},\ and\ \citenamefont {Cirac}}]{jordan2008classical}%
  \BibitemOpen
  \bibfield  {author} {\bibinfo {author} {\bibfnamefont {J.}~\bibnamefont {Jordan}}, \bibinfo {author} {\bibfnamefont {R.}~\bibnamefont {Or\'us}}, \bibinfo {author} {\bibfnamefont {G.}~\bibnamefont {Vidal}}, \bibinfo {author} {\bibfnamefont {F.}~\bibnamefont {Verstraete}},\ and\ \bibinfo {author} {\bibfnamefont {J.~I.}\ \bibnamefont {Cirac}},\ }\bibfield  {title} {\bibinfo {title} {Classical simulation of infinite-size quantum lattice systems in two spatial dimensions},\ }\href {https://doi.org/10.1103/PhysRevLett.101.250602} {\bibfield  {journal} {\bibinfo  {journal} {Phys. Rev. Lett.}\ }\textbf {\bibinfo {volume} {101}},\ \bibinfo {pages} {250602} (\bibinfo {year} {2008})}\BibitemShut {NoStop}%
\bibitem [{\citenamefont {Cirac}\ \emph {et~al.}(2021)\citenamefont {Cirac}, \citenamefont {Perez-Garcia}, \citenamefont {Schuch},\ and\ \citenamefont {Verstraete}}]{cirac2021matrix}%
  \BibitemOpen
  \bibfield  {author} {\bibinfo {author} {\bibfnamefont {J.~I.}\ \bibnamefont {Cirac}}, \bibinfo {author} {\bibfnamefont {D.}~\bibnamefont {Perez-Garcia}}, \bibinfo {author} {\bibfnamefont {N.}~\bibnamefont {Schuch}},\ and\ \bibinfo {author} {\bibfnamefont {F.}~\bibnamefont {Verstraete}},\ }\bibfield  {title} {\bibinfo {title} {Matrix product states and projected entangled pair states: Concepts, symmetries, theorems},\ }\href@noop {} {\bibfield  {journal} {\bibinfo  {journal} {Reviews of Modern Physics}\ }\textbf {\bibinfo {volume} {93}},\ \bibinfo {pages} {045003} (\bibinfo {year} {2021})}\BibitemShut {NoStop}%
\bibitem [{\citenamefont {Gray}\ and\ \citenamefont {Kourtis}(2021)}]{gray2021hyperoptimized}%
  \BibitemOpen
  \bibfield  {author} {\bibinfo {author} {\bibfnamefont {J.}~\bibnamefont {Gray}}\ and\ \bibinfo {author} {\bibfnamefont {S.}~\bibnamefont {Kourtis}},\ }\bibfield  {title} {\bibinfo {title} {Hyper-optimized tensor network contraction},\ }\href {https://doi.org/10.22331/q-2021-03-15-410} {\bibfield  {journal} {\bibinfo  {journal} {{Quantum}}\ }\textbf {\bibinfo {volume} {5}},\ \bibinfo {pages} {410} (\bibinfo {year} {2021})}\BibitemShut {NoStop}%
\bibitem [{\citenamefont {Gray}\ and\ \citenamefont {Chan}(2024)}]{gray2024hyperoptimized}%
  \BibitemOpen
  \bibfield  {author} {\bibinfo {author} {\bibfnamefont {J.}~\bibnamefont {Gray}}\ and\ \bibinfo {author} {\bibfnamefont {G.~K.-L.}\ \bibnamefont {Chan}},\ }\bibfield  {title} {\bibinfo {title} {Hyperoptimized approximate contraction of tensor networks with arbitrary geometry},\ }\href {https://doi.org/10.1103/PhysRevX.14.011009} {\bibfield  {journal} {\bibinfo  {journal} {Phys. Rev. X}\ }\textbf {\bibinfo {volume} {14}},\ \bibinfo {pages} {011009} (\bibinfo {year} {2024})}\BibitemShut {NoStop}%
\bibitem [{\citenamefont {Pan}\ \emph {et~al.}(2020)\citenamefont {Pan}, \citenamefont {Zhou}, \citenamefont {Li},\ and\ \citenamefont {Zhang}}]{pan2020contracting}%
  \BibitemOpen
  \bibfield  {author} {\bibinfo {author} {\bibfnamefont {F.}~\bibnamefont {Pan}}, \bibinfo {author} {\bibfnamefont {P.}~\bibnamefont {Zhou}}, \bibinfo {author} {\bibfnamefont {S.}~\bibnamefont {Li}},\ and\ \bibinfo {author} {\bibfnamefont {P.}~\bibnamefont {Zhang}},\ }\bibfield  {title} {\bibinfo {title} {Contracting arbitrary tensor networks: general approximate algorithm and applications in graphical models and quantum circuit simulations},\ }\href@noop {} {\bibfield  {journal} {\bibinfo  {journal} {Physical Review Letters}\ }\textbf {\bibinfo {volume} {125}},\ \bibinfo {pages} {060503} (\bibinfo {year} {2020})}\BibitemShut {NoStop}%
\bibitem [{\citenamefont {Hauru}\ \emph {et~al.}(2018)\citenamefont {Hauru}, \citenamefont {Delcamp},\ and\ \citenamefont {Mizera}}]{hauru2018renormalization}%
  \BibitemOpen
  \bibfield  {author} {\bibinfo {author} {\bibfnamefont {M.}~\bibnamefont {Hauru}}, \bibinfo {author} {\bibfnamefont {C.}~\bibnamefont {Delcamp}},\ and\ \bibinfo {author} {\bibfnamefont {S.}~\bibnamefont {Mizera}},\ }\bibfield  {title} {\bibinfo {title} {Renormalization of tensor networks using graph-independent local truncations},\ }\href@noop {} {\bibfield  {journal} {\bibinfo  {journal} {Physical Review B}\ }\textbf {\bibinfo {volume} {97}},\ \bibinfo {pages} {045111} (\bibinfo {year} {2018})}\BibitemShut {NoStop}%
\bibitem [{\citenamefont {Jahromi}\ and\ \citenamefont {Or{\'u}s}(2019)}]{jahromi2019universal}%
  \BibitemOpen
  \bibfield  {author} {\bibinfo {author} {\bibfnamefont {S.~S.}\ \bibnamefont {Jahromi}}\ and\ \bibinfo {author} {\bibfnamefont {R.}~\bibnamefont {Or{\'u}s}},\ }\bibfield  {title} {\bibinfo {title} {Universal tensor-network algorithm for any infinite lattice},\ }\href@noop {} {\bibfield  {journal} {\bibinfo  {journal} {Physical Review B}\ }\textbf {\bibinfo {volume} {99}},\ \bibinfo {pages} {195105} (\bibinfo {year} {2019})}\BibitemShut {NoStop}%
\bibitem [{\citenamefont {Patra}\ \emph {et~al.}(2024{\natexlab{a}})\citenamefont {Patra}, \citenamefont {Singh},\ and\ \citenamefont {Or{\'u}s}}]{patra2024projected}%
  \BibitemOpen
  \bibfield  {author} {\bibinfo {author} {\bibfnamefont {S.}~\bibnamefont {Patra}}, \bibinfo {author} {\bibfnamefont {S.}~\bibnamefont {Singh}},\ and\ \bibinfo {author} {\bibfnamefont {R.}~\bibnamefont {Or{\'u}s}},\ }\bibfield  {title} {\bibinfo {title} {Projected entangled pair states with flexible geometry},\ }\href@noop {} {\bibfield  {journal} {\bibinfo  {journal} {arXiv preprint arXiv:2407.21140}\ } (\bibinfo {year} {2024}{\natexlab{a}})}\BibitemShut {NoStop}%
\bibitem [{\citenamefont {Nishino}\ and\ \citenamefont {Okunishi}(1996)}]{nishino1996corner}%
  \BibitemOpen
  \bibfield  {author} {\bibinfo {author} {\bibfnamefont {T.}~\bibnamefont {Nishino}}\ and\ \bibinfo {author} {\bibfnamefont {K.}~\bibnamefont {Okunishi}},\ }\bibfield  {title} {\bibinfo {title} {Corner transfer matrix renormalization group method},\ }\href@noop {} {\bibfield  {journal} {\bibinfo  {journal} {Journal of the Physical Society of Japan}\ }\textbf {\bibinfo {volume} {65}},\ \bibinfo {pages} {891} (\bibinfo {year} {1996})}\BibitemShut {NoStop}%
\bibitem [{\citenamefont {Or\'us}\ and\ \citenamefont {Vidal}(2009)}]{orus2009simulation}%
  \BibitemOpen
  \bibfield  {author} {\bibinfo {author} {\bibfnamefont {R.}~\bibnamefont {Or\'us}}\ and\ \bibinfo {author} {\bibfnamefont {G.}~\bibnamefont {Vidal}},\ }\bibfield  {title} {\bibinfo {title} {Simulation of two-dimensional quantum systems on an infinite lattice revisited: Corner transfer matrix for tensor contraction},\ }\href {https://doi.org/10.1103/PhysRevB.80.094403} {\bibfield  {journal} {\bibinfo  {journal} {Phys. Rev. B}\ }\textbf {\bibinfo {volume} {80}},\ \bibinfo {pages} {094403} (\bibinfo {year} {2009})}\BibitemShut {NoStop}%
\bibitem [{\citenamefont {Corboz}\ \emph {et~al.}(2010)\citenamefont {Corboz}, \citenamefont {Or\'us}, \citenamefont {Bauer},\ and\ \citenamefont {Vidal}}]{corboz2010simulation}%
  \BibitemOpen
  \bibfield  {author} {\bibinfo {author} {\bibfnamefont {P.}~\bibnamefont {Corboz}}, \bibinfo {author} {\bibfnamefont {R.}~\bibnamefont {Or\'us}}, \bibinfo {author} {\bibfnamefont {B.}~\bibnamefont {Bauer}},\ and\ \bibinfo {author} {\bibfnamefont {G.}~\bibnamefont {Vidal}},\ }\bibfield  {title} {\bibinfo {title} {Simulation of strongly correlated fermions in two spatial dimensions with fermionic projected entangled-pair states},\ }\href {https://doi.org/10.1103/PhysRevB.81.165104} {\bibfield  {journal} {\bibinfo  {journal} {Phys. Rev. B}\ }\textbf {\bibinfo {volume} {81}},\ \bibinfo {pages} {165104} (\bibinfo {year} {2010})}\BibitemShut {NoStop}%
\bibitem [{\citenamefont {Corboz}\ \emph {et~al.}(2011)\citenamefont {Corboz}, \citenamefont {White}, \citenamefont {Vidal},\ and\ \citenamefont {Troyer}}]{corboz2011stripes}%
  \BibitemOpen
  \bibfield  {author} {\bibinfo {author} {\bibfnamefont {P.}~\bibnamefont {Corboz}}, \bibinfo {author} {\bibfnamefont {S.~R.}\ \bibnamefont {White}}, \bibinfo {author} {\bibfnamefont {G.}~\bibnamefont {Vidal}},\ and\ \bibinfo {author} {\bibfnamefont {M.}~\bibnamefont {Troyer}},\ }\bibfield  {title} {\bibinfo {title} {Stripes in the two-dimensional $t$-$j$ model with infinite projected entangled-pair states},\ }\href {https://doi.org/10.1103/PhysRevB.84.041108} {\bibfield  {journal} {\bibinfo  {journal} {Phys. Rev. B}\ }\textbf {\bibinfo {volume} {84}},\ \bibinfo {pages} {041108} (\bibinfo {year} {2011})}\BibitemShut {NoStop}%
\bibitem [{\citenamefont {Levin}\ and\ \citenamefont {Nave}(2007)}]{levin2007tensor}%
  \BibitemOpen
  \bibfield  {author} {\bibinfo {author} {\bibfnamefont {M.}~\bibnamefont {Levin}}\ and\ \bibinfo {author} {\bibfnamefont {C.~P.}\ \bibnamefont {Nave}},\ }\bibfield  {title} {\bibinfo {title} {Tensor renormalization group approach to two-dimensional classical lattice models},\ }\href {https://doi.org/10.1103/PhysRevLett.99.120601} {\bibfield  {journal} {\bibinfo  {journal} {Phys. Rev. Lett.}\ }\textbf {\bibinfo {volume} {99}},\ \bibinfo {pages} {120601} (\bibinfo {year} {2007})}\BibitemShut {NoStop}%
\bibitem [{\citenamefont {Evenbly}\ and\ \citenamefont {Vidal}(2015)}]{evenbly2015tensor}%
  \BibitemOpen
  \bibfield  {author} {\bibinfo {author} {\bibfnamefont {G.}~\bibnamefont {Evenbly}}\ and\ \bibinfo {author} {\bibfnamefont {G.}~\bibnamefont {Vidal}},\ }\bibfield  {title} {\bibinfo {title} {Tensor network renormalization},\ }\href {https://doi.org/10.1103/PhysRevLett.115.180405} {\bibfield  {journal} {\bibinfo  {journal} {Phys. Rev. Lett.}\ }\textbf {\bibinfo {volume} {115}},\ \bibinfo {pages} {180405} (\bibinfo {year} {2015})}\BibitemShut {NoStop}%
\bibitem [{\citenamefont {Evenbly}(2017)}]{evenbly2017algorithms}%
  \BibitemOpen
  \bibfield  {author} {\bibinfo {author} {\bibfnamefont {G.}~\bibnamefont {Evenbly}},\ }\bibfield  {title} {\bibinfo {title} {Algorithms for tensor network renormalization},\ }\href {https://doi.org/10.1103/PhysRevB.95.045117} {\bibfield  {journal} {\bibinfo  {journal} {Phys. Rev. B}\ }\textbf {\bibinfo {volume} {95}},\ \bibinfo {pages} {045117} (\bibinfo {year} {2017})}\BibitemShut {NoStop}%
\bibitem [{\citenamefont {Xie}\ \emph {et~al.}(2012)\citenamefont {Xie}, \citenamefont {Chen}, \citenamefont {Qin}, \citenamefont {Zhu}, \citenamefont {Yang},\ and\ \citenamefont {Xiang}}]{xie2012coarse}%
  \BibitemOpen
  \bibfield  {author} {\bibinfo {author} {\bibfnamefont {Z.-Y.}\ \bibnamefont {Xie}}, \bibinfo {author} {\bibfnamefont {J.}~\bibnamefont {Chen}}, \bibinfo {author} {\bibfnamefont {M.-P.}\ \bibnamefont {Qin}}, \bibinfo {author} {\bibfnamefont {J.~W.}\ \bibnamefont {Zhu}}, \bibinfo {author} {\bibfnamefont {L.-P.}\ \bibnamefont {Yang}},\ and\ \bibinfo {author} {\bibfnamefont {T.}~\bibnamefont {Xiang}},\ }\bibfield  {title} {\bibinfo {title} {Coarse-graining renormalization by higher-order singular value decomposition},\ }\href@noop {} {\bibfield  {journal} {\bibinfo  {journal} {Physical Review B}\ }\textbf {\bibinfo {volume} {86}},\ \bibinfo {pages} {045139} (\bibinfo {year} {2012})}\BibitemShut {NoStop}%
\bibitem [{\citenamefont {Vanderstraeten}\ \emph {et~al.}(2022)\citenamefont {Vanderstraeten}, \citenamefont {Burgelman}, \citenamefont {Ponsioen}, \citenamefont {Van~Damme}, \citenamefont {Vanhecke}, \citenamefont {Corboz}, \citenamefont {Haegeman},\ and\ \citenamefont {Verstraete}}]{vanderstraeten2022variational}%
  \BibitemOpen
  \bibfield  {author} {\bibinfo {author} {\bibfnamefont {L.}~\bibnamefont {Vanderstraeten}}, \bibinfo {author} {\bibfnamefont {L.}~\bibnamefont {Burgelman}}, \bibinfo {author} {\bibfnamefont {B.}~\bibnamefont {Ponsioen}}, \bibinfo {author} {\bibfnamefont {M.}~\bibnamefont {Van~Damme}}, \bibinfo {author} {\bibfnamefont {B.}~\bibnamefont {Vanhecke}}, \bibinfo {author} {\bibfnamefont {P.}~\bibnamefont {Corboz}}, \bibinfo {author} {\bibfnamefont {J.}~\bibnamefont {Haegeman}},\ and\ \bibinfo {author} {\bibfnamefont {F.}~\bibnamefont {Verstraete}},\ }\bibfield  {title} {\bibinfo {title} {Variational methods for contracting projected entangled-pair states},\ }\href@noop {} {\bibfield  {journal} {\bibinfo  {journal} {Physical Review B}\ }\textbf {\bibinfo {volume} {105}},\ \bibinfo {pages} {195140} (\bibinfo {year} {2022})}\BibitemShut {NoStop}%
\bibitem [{\citenamefont {Vieijra}\ \emph {et~al.}(2021)\citenamefont {Vieijra}, \citenamefont {Haegeman}, \citenamefont {Verstraete},\ and\ \citenamefont {Vanderstraeten}}]{vieijra2021direct}%
  \BibitemOpen
  \bibfield  {author} {\bibinfo {author} {\bibfnamefont {T.}~\bibnamefont {Vieijra}}, \bibinfo {author} {\bibfnamefont {J.}~\bibnamefont {Haegeman}}, \bibinfo {author} {\bibfnamefont {F.}~\bibnamefont {Verstraete}},\ and\ \bibinfo {author} {\bibfnamefont {L.}~\bibnamefont {Vanderstraeten}},\ }\bibfield  {title} {\bibinfo {title} {Direct sampling of projected entangled-pair states},\ }\href@noop {} {\bibfield  {journal} {\bibinfo  {journal} {Physical Review B}\ }\textbf {\bibinfo {volume} {104}},\ \bibinfo {pages} {235141} (\bibinfo {year} {2021})}\BibitemShut {NoStop}%
\bibitem [{\citenamefont {Zaletel}\ and\ \citenamefont {Pollmann}(2020)}]{zaletel2020isometric}%
  \BibitemOpen
  \bibfield  {author} {\bibinfo {author} {\bibfnamefont {M.~P.}\ \bibnamefont {Zaletel}}\ and\ \bibinfo {author} {\bibfnamefont {F.}~\bibnamefont {Pollmann}},\ }\bibfield  {title} {\bibinfo {title} {Isometric tensor network states in two dimensions},\ }\href@noop {} {\bibfield  {journal} {\bibinfo  {journal} {Physical review letters}\ }\textbf {\bibinfo {volume} {124}},\ \bibinfo {pages} {037201} (\bibinfo {year} {2020})}\BibitemShut {NoStop}%
\bibitem [{\citenamefont {Lin}\ \emph {et~al.}(2022)\citenamefont {Lin}, \citenamefont {Zaletel},\ and\ \citenamefont {Pollmann}}]{lin2022efficient}%
  \BibitemOpen
  \bibfield  {author} {\bibinfo {author} {\bibfnamefont {S.-H.}\ \bibnamefont {Lin}}, \bibinfo {author} {\bibfnamefont {M.~P.}\ \bibnamefont {Zaletel}},\ and\ \bibinfo {author} {\bibfnamefont {F.}~\bibnamefont {Pollmann}},\ }\bibfield  {title} {\bibinfo {title} {Efficient simulation of dynamics in two-dimensional quantum spin systems with isometric tensor networks},\ }\href@noop {} {\bibfield  {journal} {\bibinfo  {journal} {Physical Review B}\ }\textbf {\bibinfo {volume} {106}},\ \bibinfo {pages} {245102} (\bibinfo {year} {2022})}\BibitemShut {NoStop}%
\bibitem [{\citenamefont {Soejima}\ \emph {et~al.}(2020)\citenamefont {Soejima}, \citenamefont {Siva}, \citenamefont {Bultinck}, \citenamefont {Chatterjee}, \citenamefont {Pollmann}, \citenamefont {Zaletel} \emph {et~al.}}]{soejima2020isometric}%
  \BibitemOpen
  \bibfield  {author} {\bibinfo {author} {\bibfnamefont {T.}~\bibnamefont {Soejima}}, \bibinfo {author} {\bibfnamefont {K.}~\bibnamefont {Siva}}, \bibinfo {author} {\bibfnamefont {N.}~\bibnamefont {Bultinck}}, \bibinfo {author} {\bibfnamefont {S.}~\bibnamefont {Chatterjee}}, \bibinfo {author} {\bibfnamefont {F.}~\bibnamefont {Pollmann}}, \bibinfo {author} {\bibfnamefont {M.~P.}\ \bibnamefont {Zaletel}}, \emph {et~al.},\ }\bibfield  {title} {\bibinfo {title} {Isometric tensor network representation of string-net liquids},\ }\href@noop {} {\bibfield  {journal} {\bibinfo  {journal} {Physical Review B}\ }\textbf {\bibinfo {volume} {101}},\ \bibinfo {pages} {085117} (\bibinfo {year} {2020})}\BibitemShut {NoStop}%
\bibitem [{\citenamefont {Alkabetz}\ and\ \citenamefont {Arad}(2021)}]{alkabetz2021tensor}%
  \BibitemOpen
  \bibfield  {author} {\bibinfo {author} {\bibfnamefont {R.}~\bibnamefont {Alkabetz}}\ and\ \bibinfo {author} {\bibfnamefont {I.}~\bibnamefont {Arad}},\ }\bibfield  {title} {\bibinfo {title} {Tensor networks contraction and the belief propagation algorithm},\ }\href@noop {} {\bibfield  {journal} {\bibinfo  {journal} {Physical Review Research}\ }\textbf {\bibinfo {volume} {3}},\ \bibinfo {pages} {023073} (\bibinfo {year} {2021})}\BibitemShut {NoStop}%
\bibitem [{\citenamefont {Mezard}\ and\ \citenamefont {Montanari}(2009)}]{mezard2009information}%
  \BibitemOpen
  \bibfield  {author} {\bibinfo {author} {\bibfnamefont {M.}~\bibnamefont {Mezard}}\ and\ \bibinfo {author} {\bibfnamefont {A.}~\bibnamefont {Montanari}},\ }\href@noop {} {\emph {\bibinfo {title} {Information, Physics, and Computation}}}\ (\bibinfo  {publisher} {Oxford University Press, Inc.},\ \bibinfo {address} {USA},\ \bibinfo {year} {2009})\BibitemShut {NoStop}%
\bibitem [{\citenamefont {Kschischang}\ \emph {et~al.}(2001)\citenamefont {Kschischang}, \citenamefont {Frey},\ and\ \citenamefont {Loeliger}}]{kschischang2001factor}%
  \BibitemOpen
  \bibfield  {author} {\bibinfo {author} {\bibfnamefont {F.}~\bibnamefont {Kschischang}}, \bibinfo {author} {\bibfnamefont {B.}~\bibnamefont {Frey}},\ and\ \bibinfo {author} {\bibfnamefont {H.-A.}\ \bibnamefont {Loeliger}},\ }\bibfield  {title} {\bibinfo {title} {Factor graphs and the sum-product algorithm},\ }\href {https://doi.org/10.1109/18.910572} {\bibfield  {journal} {\bibinfo  {journal} {IEEE Transactions on Information Theory}\ }\textbf {\bibinfo {volume} {47}},\ \bibinfo {pages} {498} (\bibinfo {year} {2001})}\BibitemShut {NoStop}%
\bibitem [{\citenamefont {Pearl}(1988)}]{pearl1988probabilistic}%
  \BibitemOpen
  \bibfield  {author} {\bibinfo {author} {\bibfnamefont {J.}~\bibnamefont {Pearl}},\ }\href@noop {} {\emph {\bibinfo {title} {Probabilistic Reasoning in Intelligent Systems: Networks of Plausible Inference}}}\ (\bibinfo  {publisher} {Morgan Kaufmann Publishers Inc.},\ \bibinfo {address} {San Francisco, CA, USA},\ \bibinfo {year} {1988})\BibitemShut {NoStop}%
\bibitem [{\citenamefont {Wainwright}\ \emph {et~al.}(2003)\citenamefont {Wainwright}, \citenamefont {Jaakkola},\ and\ \citenamefont {Willsky}}]{wainwright2003tree}%
  \BibitemOpen
  \bibfield  {author} {\bibinfo {author} {\bibfnamefont {M.~J.}\ \bibnamefont {Wainwright}}, \bibinfo {author} {\bibfnamefont {T.~S.}\ \bibnamefont {Jaakkola}},\ and\ \bibinfo {author} {\bibfnamefont {A.~S.}\ \bibnamefont {Willsky}},\ }\bibfield  {title} {\bibinfo {title} {Tree-reweighted belief propagation algorithms and approximate ml estimation by pseudo-moment matching},\ }in\ \href@noop {} {\emph {\bibinfo {booktitle} {International Workshop on Artificial Intelligence and Statistics}}}\ (\bibinfo {organization} {PMLR},\ \bibinfo {year} {2003})\ pp.\ \bibinfo {pages} {308--315}\BibitemShut {NoStop}%
\bibitem [{\citenamefont {Yedidia}\ \emph {et~al.}(2005)\citenamefont {Yedidia}, \citenamefont {Freeman},\ and\ \citenamefont {Weiss}}]{yedidia2005constructing}%
  \BibitemOpen
  \bibfield  {author} {\bibinfo {author} {\bibfnamefont {J.}~\bibnamefont {Yedidia}}, \bibinfo {author} {\bibfnamefont {W.}~\bibnamefont {Freeman}},\ and\ \bibinfo {author} {\bibfnamefont {Y.}~\bibnamefont {Weiss}},\ }\bibfield  {title} {\bibinfo {title} {Constructing free-energy approximations and generalized belief propagation algorithms},\ }\href {https://doi.org/10.1109/TIT.2005.850085} {\bibfield  {journal} {\bibinfo  {journal} {IEEE Transactions on Information Theory}\ }\textbf {\bibinfo {volume} {51}},\ \bibinfo {pages} {2282} (\bibinfo {year} {2005})}\BibitemShut {NoStop}%
\bibitem [{\citenamefont {Lauritzen}\ and\ \citenamefont {Spiegelhalter}(1988)}]{lauritzen1988local}%
  \BibitemOpen
  \bibfield  {author} {\bibinfo {author} {\bibfnamefont {S.~L.}\ \bibnamefont {Lauritzen}}\ and\ \bibinfo {author} {\bibfnamefont {D.~J.}\ \bibnamefont {Spiegelhalter}},\ }\bibfield  {title} {\bibinfo {title} {Local computations with probabilities on graphical structures and their application to expert systems},\ }\href {https://doi.org/https://doi.org/10.1111/j.2517-6161.1988.tb01721.x} {\bibfield  {journal} {\bibinfo  {journal} {Journal of the Royal Statistical Society: Series B (Methodological)}\ }\textbf {\bibinfo {volume} {50}},\ \bibinfo {pages} {157} (\bibinfo {year} {1988})},\ \Eprint {https://arxiv.org/abs/https://rss.onlinelibrary.wiley.com/doi/pdf/10.1111/j.2517-6161.1988.tb01721.x} {https://rss.onlinelibrary.wiley.com/doi/pdf/10.1111/j.2517-6161.1988.tb01721.x} \BibitemShut {NoStop}%
\bibitem [{\citenamefont {Tindall}\ and\ \citenamefont {Fishman}(2023)}]{tindall2023gauging}%
  \BibitemOpen
  \bibfield  {author} {\bibinfo {author} {\bibfnamefont {J.}~\bibnamefont {Tindall}}\ and\ \bibinfo {author} {\bibfnamefont {M.}~\bibnamefont {Fishman}},\ }\bibfield  {title} {\bibinfo {title} {Gauging tensor networks with belief propagation},\ }\href@noop {} {\bibfield  {journal} {\bibinfo  {journal} {SciPost Physics}\ }\textbf {\bibinfo {volume} {15}},\ \bibinfo {pages} {222} (\bibinfo {year} {2023})}\BibitemShut {NoStop}%
\bibitem [{\citenamefont {Sahu}\ and\ \citenamefont {Swingle}(2022)}]{sahu2022efficient}%
  \BibitemOpen
  \bibfield  {author} {\bibinfo {author} {\bibfnamefont {S.}~\bibnamefont {Sahu}}\ and\ \bibinfo {author} {\bibfnamefont {B.}~\bibnamefont {Swingle}},\ }\bibfield  {title} {\bibinfo {title} {Efficient tensor network simulation of quantum many-body physics on sparse graphs},\ }\href@noop {} {\bibfield  {journal} {\bibinfo  {journal} {arXiv preprint arXiv:2206.04701}\ } (\bibinfo {year} {2022})}\BibitemShut {NoStop}%
\bibitem [{\citenamefont {Guo}\ \emph {et~al.}(2023)\citenamefont {Guo}, \citenamefont {Poletti},\ and\ \citenamefont {Arad}}]{guo2023block}%
  \BibitemOpen
  \bibfield  {author} {\bibinfo {author} {\bibfnamefont {C.}~\bibnamefont {Guo}}, \bibinfo {author} {\bibfnamefont {D.}~\bibnamefont {Poletti}},\ and\ \bibinfo {author} {\bibfnamefont {I.}~\bibnamefont {Arad}},\ }\bibfield  {title} {\bibinfo {title} {Block belief propagation algorithm for two-dimensional tensor networks},\ }\href@noop {} {\bibfield  {journal} {\bibinfo  {journal} {Physical Review B}\ }\textbf {\bibinfo {volume} {108}},\ \bibinfo {pages} {125111} (\bibinfo {year} {2023})}\BibitemShut {NoStop}%
\bibitem [{\citenamefont {Kaufmann}\ and\ \citenamefont {Arad}(2024)}]{kaufmann2024blockbp}%
  \BibitemOpen
  \bibfield  {author} {\bibinfo {author} {\bibfnamefont {A.}~\bibnamefont {Kaufmann}}\ and\ \bibinfo {author} {\bibfnamefont {I.}~\bibnamefont {Arad}},\ }\bibfield  {title} {\bibinfo {title} {A blockbp decoder for the surface code},\ }\href@noop {} {\bibfield  {journal} {\bibinfo  {journal} {arXiv preprint arXiv:2402.04834}\ } (\bibinfo {year} {2024})}\BibitemShut {NoStop}%
\bibitem [{\citenamefont {Begušić}\ \emph {et~al.}(2024)\citenamefont {Begušić}, \citenamefont {Gray},\ and\ \citenamefont {Chan}}]{tomislav2024fast}%
  \BibitemOpen
  \bibfield  {author} {\bibinfo {author} {\bibfnamefont {T.}~\bibnamefont {Begušić}}, \bibinfo {author} {\bibfnamefont {J.}~\bibnamefont {Gray}},\ and\ \bibinfo {author} {\bibfnamefont {G.~K.-L.}\ \bibnamefont {Chan}},\ }\bibfield  {title} {\bibinfo {title} {Fast and converged classical simulations of evidence for the utility of quantum computing before fault tolerance},\ }\href {https://doi.org/10.1126/sciadv.adk4321} {\bibfield  {journal} {\bibinfo  {journal} {Science Advances}\ }\textbf {\bibinfo {volume} {10}},\ \bibinfo {pages} {eadk4321} (\bibinfo {year} {2024})},\ \Eprint {https://arxiv.org/abs/https://www.science.org/doi/pdf/10.1126/sciadv.adk4321} {https://www.science.org/doi/pdf/10.1126/sciadv.adk4321} \BibitemShut {NoStop}%
\bibitem [{\citenamefont {Begu{\v{s}}i{\'c}}\ and\ \citenamefont {Chan}(2023)}]{beguvsic2023fast}%
  \BibitemOpen
  \bibfield  {author} {\bibinfo {author} {\bibfnamefont {T.}~\bibnamefont {Begu{\v{s}}i{\'c}}}\ and\ \bibinfo {author} {\bibfnamefont {G.~K.}\ \bibnamefont {Chan}},\ }\bibfield  {title} {\bibinfo {title} {Fast classical simulation of evidence for the utility of quantum computing before fault tolerance},\ }\href@noop {} {\bibfield  {journal} {\bibinfo  {journal} {arXiv preprint arXiv:2306.16372}\ } (\bibinfo {year} {2023})}\BibitemShut {NoStop}%
\bibitem [{\citenamefont {Liao}\ \emph {et~al.}(2023)\citenamefont {Liao}, \citenamefont {Wang}, \citenamefont {Zhou}, \citenamefont {Zhang},\ and\ \citenamefont {Xiang}}]{liao2023simulation}%
  \BibitemOpen
  \bibfield  {author} {\bibinfo {author} {\bibfnamefont {H.-J.}\ \bibnamefont {Liao}}, \bibinfo {author} {\bibfnamefont {K.}~\bibnamefont {Wang}}, \bibinfo {author} {\bibfnamefont {Z.-S.}\ \bibnamefont {Zhou}}, \bibinfo {author} {\bibfnamefont {P.}~\bibnamefont {Zhang}},\ and\ \bibinfo {author} {\bibfnamefont {T.}~\bibnamefont {Xiang}},\ }\bibfield  {title} {\bibinfo {title} {Simulation of ibm's kicked ising experiment with projected entangled pair operator},\ }\href@noop {} {\bibfield  {journal} {\bibinfo  {journal} {arXiv preprint arXiv:2308.03082}\ } (\bibinfo {year} {2023})}\BibitemShut {NoStop}%
\bibitem [{\citenamefont {Patra}\ \emph {et~al.}(2024{\natexlab{b}})\citenamefont {Patra}, \citenamefont {Jahromi}, \citenamefont {Singh},\ and\ \citenamefont {Or\'us}}]{PhysRevResearch.6.013326}%
  \BibitemOpen
  \bibfield  {author} {\bibinfo {author} {\bibfnamefont {S.}~\bibnamefont {Patra}}, \bibinfo {author} {\bibfnamefont {S.~S.}\ \bibnamefont {Jahromi}}, \bibinfo {author} {\bibfnamefont {S.}~\bibnamefont {Singh}},\ and\ \bibinfo {author} {\bibfnamefont {R.}~\bibnamefont {Or\'us}},\ }\bibfield  {title} {\bibinfo {title} {Efficient tensor network simulation of ibm's largest quantum processors},\ }\href {https://doi.org/10.1103/PhysRevResearch.6.013326} {\bibfield  {journal} {\bibinfo  {journal} {Phys. Rev. Res.}\ }\textbf {\bibinfo {volume} {6}},\ \bibinfo {pages} {013326} (\bibinfo {year} {2024}{\natexlab{b}})}\BibitemShut {NoStop}%
\bibitem [{\citenamefont {Kim}\ \emph {et~al.}(2023)\citenamefont {Kim}, \citenamefont {Eddins}, \citenamefont {Anand}, \citenamefont {Wei}, \citenamefont {Van Den~Berg}, \citenamefont {Rosenblatt}, \citenamefont {Nayfeh}, \citenamefont {Wu}, \citenamefont {Zaletel}, \citenamefont {Temme} \emph {et~al.}}]{kim2023evidence}%
  \BibitemOpen
  \bibfield  {author} {\bibinfo {author} {\bibfnamefont {Y.}~\bibnamefont {Kim}}, \bibinfo {author} {\bibfnamefont {A.}~\bibnamefont {Eddins}}, \bibinfo {author} {\bibfnamefont {S.}~\bibnamefont {Anand}}, \bibinfo {author} {\bibfnamefont {K.~X.}\ \bibnamefont {Wei}}, \bibinfo {author} {\bibfnamefont {E.}~\bibnamefont {Van Den~Berg}}, \bibinfo {author} {\bibfnamefont {S.}~\bibnamefont {Rosenblatt}}, \bibinfo {author} {\bibfnamefont {H.}~\bibnamefont {Nayfeh}}, \bibinfo {author} {\bibfnamefont {Y.}~\bibnamefont {Wu}}, \bibinfo {author} {\bibfnamefont {M.}~\bibnamefont {Zaletel}}, \bibinfo {author} {\bibfnamefont {K.}~\bibnamefont {Temme}}, \emph {et~al.},\ }\bibfield  {title} {\bibinfo {title} {Evidence for the utility of quantum computing before fault tolerance},\ }\href@noop {} {\bibfield  {journal} {\bibinfo  {journal} {Nature}\ }\textbf {\bibinfo {volume} {618}},\ \bibinfo {pages} {500} (\bibinfo {year} {2023})}\BibitemShut {NoStop}%
\bibitem [{\citenamefont {Arute}\ \emph {et~al.}(2019)\citenamefont {Arute}, \citenamefont {Arya}, \citenamefont {Babbush}, \citenamefont {Bacon}, \citenamefont {Bardin}, \citenamefont {Barends}, \citenamefont {Biswas}, \citenamefont {Boixo}, \citenamefont {Brandao}, \citenamefont {Buell} \emph {et~al.}}]{arute2019quantum}%
  \BibitemOpen
  \bibfield  {author} {\bibinfo {author} {\bibfnamefont {F.}~\bibnamefont {Arute}}, \bibinfo {author} {\bibfnamefont {K.}~\bibnamefont {Arya}}, \bibinfo {author} {\bibfnamefont {R.}~\bibnamefont {Babbush}}, \bibinfo {author} {\bibfnamefont {D.}~\bibnamefont {Bacon}}, \bibinfo {author} {\bibfnamefont {J.~C.}\ \bibnamefont {Bardin}}, \bibinfo {author} {\bibfnamefont {R.}~\bibnamefont {Barends}}, \bibinfo {author} {\bibfnamefont {R.}~\bibnamefont {Biswas}}, \bibinfo {author} {\bibfnamefont {S.}~\bibnamefont {Boixo}}, \bibinfo {author} {\bibfnamefont {F.~G.}\ \bibnamefont {Brandao}}, \bibinfo {author} {\bibfnamefont {D.~A.}\ \bibnamefont {Buell}}, \emph {et~al.},\ }\bibfield  {title} {\bibinfo {title} {Quantum supremacy using a programmable superconducting processor},\ }\href@noop {} {\bibfield  {journal} {\bibinfo  {journal} {Nature}\ }\textbf {\bibinfo {volume} {574}},\ \bibinfo {pages} {505} (\bibinfo {year} {2019})}\BibitemShut {NoStop}%
\bibitem [{\citenamefont {Huang}\ \emph {et~al.}(2020)\citenamefont {Huang}, \citenamefont {Kueng},\ and\ \citenamefont {Preskill}}]{huang2020predicting}%
  \BibitemOpen
  \bibfield  {author} {\bibinfo {author} {\bibfnamefont {H.-Y.}\ \bibnamefont {Huang}}, \bibinfo {author} {\bibfnamefont {R.}~\bibnamefont {Kueng}},\ and\ \bibinfo {author} {\bibfnamefont {J.}~\bibnamefont {Preskill}},\ }\bibfield  {title} {\bibinfo {title} {Predicting many properties of a quantum system from very few measurements},\ }\href@noop {} {\bibfield  {journal} {\bibinfo  {journal} {Nature Physics}\ }\textbf {\bibinfo {volume} {16}},\ \bibinfo {pages} {1050} (\bibinfo {year} {2020})}\BibitemShut {NoStop}%
\bibitem [{\citenamefont {McGinley}\ and\ \citenamefont {Fava}(2023)}]{mcginley2023shadow}%
  \BibitemOpen
  \bibfield  {author} {\bibinfo {author} {\bibfnamefont {M.}~\bibnamefont {McGinley}}\ and\ \bibinfo {author} {\bibfnamefont {M.}~\bibnamefont {Fava}},\ }\bibfield  {title} {\bibinfo {title} {Shadow tomography from emergent state designs in analog quantum simulators},\ }\href@noop {} {\bibfield  {journal} {\bibinfo  {journal} {Physical Review Letters}\ }\textbf {\bibinfo {volume} {131}},\ \bibinfo {pages} {160601} (\bibinfo {year} {2023})}\BibitemShut {NoStop}%
\bibitem [{\citenamefont {Tran}\ \emph {et~al.}(2023)\citenamefont {Tran}, \citenamefont {Mark}, \citenamefont {Ho},\ and\ \citenamefont {Choi}}]{tran2023measuring}%
  \BibitemOpen
  \bibfield  {author} {\bibinfo {author} {\bibfnamefont {M.~C.}\ \bibnamefont {Tran}}, \bibinfo {author} {\bibfnamefont {D.~K.}\ \bibnamefont {Mark}}, \bibinfo {author} {\bibfnamefont {W.~W.}\ \bibnamefont {Ho}},\ and\ \bibinfo {author} {\bibfnamefont {S.}~\bibnamefont {Choi}},\ }\bibfield  {title} {\bibinfo {title} {Measuring arbitrary physical properties in analog quantum simulation},\ }\href {https://doi.org/10.1103/PhysRevX.13.011049} {\bibfield  {journal} {\bibinfo  {journal} {Phys. Rev. X}\ }\textbf {\bibinfo {volume} {13}},\ \bibinfo {pages} {011049} (\bibinfo {year} {2023})}\BibitemShut {NoStop}%
\bibitem [{\citenamefont {Bauer}\ \emph {et~al.}(2015)\citenamefont {Bauer}, \citenamefont {Wang}, \citenamefont {Pi{\v{z}}orn},\ and\ \citenamefont {Troyer}}]{bauer2015entanglement}%
  \BibitemOpen
  \bibfield  {author} {\bibinfo {author} {\bibfnamefont {B.}~\bibnamefont {Bauer}}, \bibinfo {author} {\bibfnamefont {L.}~\bibnamefont {Wang}}, \bibinfo {author} {\bibfnamefont {I.}~\bibnamefont {Pi{\v{z}}orn}},\ and\ \bibinfo {author} {\bibfnamefont {M.}~\bibnamefont {Troyer}},\ }\bibfield  {title} {\bibinfo {title} {Entanglement as a resource in adiabatic quantum optimization},\ }\href@noop {} {\bibfield  {journal} {\bibinfo  {journal} {arXiv preprint arXiv:1501.06914}\ } (\bibinfo {year} {2015})}\BibitemShut {NoStop}%
\bibitem [{\citenamefont {Layden}\ \emph {et~al.}(2023)\citenamefont {Layden}, \citenamefont {Mazzola}, \citenamefont {Mishmash}, \citenamefont {Motta}, \citenamefont {Wocjan}, \citenamefont {Kim},\ and\ \citenamefont {Sheldon}}]{layden2023quantum}%
  \BibitemOpen
  \bibfield  {author} {\bibinfo {author} {\bibfnamefont {D.}~\bibnamefont {Layden}}, \bibinfo {author} {\bibfnamefont {G.}~\bibnamefont {Mazzola}}, \bibinfo {author} {\bibfnamefont {R.~V.}\ \bibnamefont {Mishmash}}, \bibinfo {author} {\bibfnamefont {M.}~\bibnamefont {Motta}}, \bibinfo {author} {\bibfnamefont {P.}~\bibnamefont {Wocjan}}, \bibinfo {author} {\bibfnamefont {J.-S.}\ \bibnamefont {Kim}},\ and\ \bibinfo {author} {\bibfnamefont {S.}~\bibnamefont {Sheldon}},\ }\bibfield  {title} {\bibinfo {title} {Quantum-enhanced markov chain monte carlo},\ }\href@noop {} {\bibfield  {journal} {\bibinfo  {journal} {Nature}\ }\textbf {\bibinfo {volume} {619}},\ \bibinfo {pages} {282} (\bibinfo {year} {2023})}\BibitemShut {NoStop}%
\bibitem [{\citenamefont {Vidal}(2008)}]{vidal2008class}%
  \BibitemOpen
  \bibfield  {author} {\bibinfo {author} {\bibfnamefont {G.}~\bibnamefont {Vidal}},\ }\bibfield  {title} {\bibinfo {title} {Class of quantum many-body states that can be efficiently simulated},\ }\href@noop {} {\bibfield  {journal} {\bibinfo  {journal} {Physical review letters}\ }\textbf {\bibinfo {volume} {101}},\ \bibinfo {pages} {110501} (\bibinfo {year} {2008})}\BibitemShut {NoStop}%
\bibitem [{\citenamefont {Eckart}\ and\ \citenamefont {Young}(1936)}]{eckart1936approximation}%
  \BibitemOpen
  \bibfield  {author} {\bibinfo {author} {\bibfnamefont {C.}~\bibnamefont {Eckart}}\ and\ \bibinfo {author} {\bibfnamefont {G.}~\bibnamefont {Young}},\ }\bibfield  {title} {\bibinfo {title} {The approximation of one matrix by another of lower rank},\ }\href@noop {} {\bibfield  {journal} {\bibinfo  {journal} {Psychometrika}\ }\textbf {\bibinfo {volume} {1}},\ \bibinfo {pages} {211} (\bibinfo {year} {1936})}\BibitemShut {NoStop}%
\bibitem [{\citenamefont {Mirsky}(1960)}]{mirsky1960symmetric}%
  \BibitemOpen
  \bibfield  {author} {\bibinfo {author} {\bibfnamefont {L.}~\bibnamefont {Mirsky}},\ }\bibfield  {title} {\bibinfo {title} {Symmetric gauge functions and unitarily invariant norms},\ }\href@noop {} {\bibfield  {journal} {\bibinfo  {journal} {The quarterly journal of mathematics}\ }\textbf {\bibinfo {volume} {11}},\ \bibinfo {pages} {50} (\bibinfo {year} {1960})}\BibitemShut {NoStop}%
\bibitem [{\citenamefont {Suzuki}(1990)}]{suzuki1990fractal}%
  \BibitemOpen
  \bibfield  {author} {\bibinfo {author} {\bibfnamefont {M.}~\bibnamefont {Suzuki}},\ }\bibfield  {title} {\bibinfo {title} {Fractal decomposition of exponential operators with applications to many-body theories and monte carlo simulations},\ }\href@noop {} {\bibfield  {journal} {\bibinfo  {journal} {Physics Letters A}\ }\textbf {\bibinfo {volume} {146}},\ \bibinfo {pages} {319} (\bibinfo {year} {1990})}\BibitemShut {NoStop}%
\bibitem [{\citenamefont {Tiunov}\ \emph {et~al.}(2019)\citenamefont {Tiunov}, \citenamefont {Ulanov},\ and\ \citenamefont {Lvovsky}}]{tiunov2019annealing}%
  \BibitemOpen
  \bibfield  {author} {\bibinfo {author} {\bibfnamefont {E.~S.}\ \bibnamefont {Tiunov}}, \bibinfo {author} {\bibfnamefont {A.~E.}\ \bibnamefont {Ulanov}},\ and\ \bibinfo {author} {\bibfnamefont {A.}~\bibnamefont {Lvovsky}},\ }\bibfield  {title} {\bibinfo {title} {Annealing by simulating the coherent ising machine},\ }\href@noop {} {\bibfield  {journal} {\bibinfo  {journal} {Optics express}\ }\textbf {\bibinfo {volume} {27}},\ \bibinfo {pages} {10288} (\bibinfo {year} {2019})}\BibitemShut {NoStop}%
\bibitem [{\citenamefont {Steger}\ and\ \citenamefont {Wormald}(1999)}]{steger1999generating}%
  \BibitemOpen
  \bibfield  {author} {\bibinfo {author} {\bibfnamefont {A.}~\bibnamefont {Steger}}\ and\ \bibinfo {author} {\bibfnamefont {N.~C.}\ \bibnamefont {Wormald}},\ }\bibfield  {title} {\bibinfo {title} {Generating random regular graphs quickly},\ }\href@noop {} {\bibfield  {journal} {\bibinfo  {journal} {Combinatorics, Probability and Computing}\ }\textbf {\bibinfo {volume} {8}},\ \bibinfo {pages} {377} (\bibinfo {year} {1999})}\BibitemShut {NoStop}%
\bibitem [{\citenamefont {Hagberg}\ \emph {et~al.}(2008)\citenamefont {Hagberg}, \citenamefont {Swart},\ and\ \citenamefont {Schult}}]{hagberg2008exploring}%
  \BibitemOpen
  \bibfield  {author} {\bibinfo {author} {\bibfnamefont {A.}~\bibnamefont {Hagberg}}, \bibinfo {author} {\bibfnamefont {P.~J.}\ \bibnamefont {Swart}},\ and\ \bibinfo {author} {\bibfnamefont {D.~A.}\ \bibnamefont {Schult}},\ }\href@noop {} {\emph {\bibinfo {title} {Exploring network structure, dynamics, and function using NetworkX}}},\ \bibinfo {type} {Tech. Rep.}\ (\bibinfo  {institution} {Los Alamos National Laboratory (LANL), Los Alamos, NM (United States)},\ \bibinfo {year} {2008})\BibitemShut {NoStop}%
\bibitem [{\citenamefont {Von~Luxburg}(2007)}]{von2007tutorial}%
  \BibitemOpen
  \bibfield  {author} {\bibinfo {author} {\bibfnamefont {U.}~\bibnamefont {Von~Luxburg}},\ }\bibfield  {title} {\bibinfo {title} {A tutorial on spectral clustering},\ }\href@noop {} {\bibfield  {journal} {\bibinfo  {journal} {Statistics and computing}\ }\textbf {\bibinfo {volume} {17}},\ \bibinfo {pages} {395} (\bibinfo {year} {2007})}\BibitemShut {NoStop}%
\bibitem [{\citenamefont {Zhou}\ \emph {et~al.}(2020{\natexlab{b}})\citenamefont {Zhou}, \citenamefont {Stoudenmire},\ and\ \citenamefont {Waintal}}]{zhou2020limits}%
  \BibitemOpen
  \bibfield  {author} {\bibinfo {author} {\bibfnamefont {Y.}~\bibnamefont {Zhou}}, \bibinfo {author} {\bibfnamefont {E.~M.}\ \bibnamefont {Stoudenmire}},\ and\ \bibinfo {author} {\bibfnamefont {X.}~\bibnamefont {Waintal}},\ }\bibfield  {title} {\bibinfo {title} {What limits the simulation of quantum computers?},\ }\href@noop {} {\bibfield  {journal} {\bibinfo  {journal} {Physical Review X}\ }\textbf {\bibinfo {volume} {10}},\ \bibinfo {pages} {041038} (\bibinfo {year} {2020}{\natexlab{b}})}\BibitemShut {NoStop}%
\bibitem [{\citenamefont {Wainwright}\ \emph {et~al.}(2005)\citenamefont {Wainwright}, \citenamefont {Jaakkola},\ and\ \citenamefont {Willsky}}]{wainwright2005new}%
  \BibitemOpen
  \bibfield  {author} {\bibinfo {author} {\bibfnamefont {M.~J.}\ \bibnamefont {Wainwright}}, \bibinfo {author} {\bibfnamefont {T.~S.}\ \bibnamefont {Jaakkola}},\ and\ \bibinfo {author} {\bibfnamefont {A.~S.}\ \bibnamefont {Willsky}},\ }\bibfield  {title} {\bibinfo {title} {A new class of upper bounds on the log partition function},\ }\href@noop {} {\bibfield  {journal} {\bibinfo  {journal} {IEEE Transactions on Information Theory}\ }\textbf {\bibinfo {volume} {51}},\ \bibinfo {pages} {2313} (\bibinfo {year} {2005})}\BibitemShut {NoStop}%
\bibitem [{\citenamefont {Hu}\ \emph {et~al.}(2023)\citenamefont {Hu}, \citenamefont {Choi},\ and\ \citenamefont {You}}]{hu2023classical}%
  \BibitemOpen
  \bibfield  {author} {\bibinfo {author} {\bibfnamefont {H.-Y.}\ \bibnamefont {Hu}}, \bibinfo {author} {\bibfnamefont {S.}~\bibnamefont {Choi}},\ and\ \bibinfo {author} {\bibfnamefont {Y.-Z.}\ \bibnamefont {You}},\ }\bibfield  {title} {\bibinfo {title} {Classical shadow tomography with locally scrambled quantum dynamics},\ }\href@noop {} {\bibfield  {journal} {\bibinfo  {journal} {Physical Review Research}\ }\textbf {\bibinfo {volume} {5}},\ \bibinfo {pages} {023027} (\bibinfo {year} {2023})}\BibitemShut {NoStop}%
\bibitem [{\citenamefont {Boixo}\ \emph {et~al.}(2018)\citenamefont {Boixo}, \citenamefont {Isakov}, \citenamefont {Smelyanskiy}, \citenamefont {Babbush}, \citenamefont {Ding}, \citenamefont {Jiang}, \citenamefont {Bremner}, \citenamefont {Martinis},\ and\ \citenamefont {Neven}}]{boixo2018characterizing}%
  \BibitemOpen
  \bibfield  {author} {\bibinfo {author} {\bibfnamefont {S.}~\bibnamefont {Boixo}}, \bibinfo {author} {\bibfnamefont {S.~V.}\ \bibnamefont {Isakov}}, \bibinfo {author} {\bibfnamefont {V.~N.}\ \bibnamefont {Smelyanskiy}}, \bibinfo {author} {\bibfnamefont {R.}~\bibnamefont {Babbush}}, \bibinfo {author} {\bibfnamefont {N.}~\bibnamefont {Ding}}, \bibinfo {author} {\bibfnamefont {Z.}~\bibnamefont {Jiang}}, \bibinfo {author} {\bibfnamefont {M.~J.}\ \bibnamefont {Bremner}}, \bibinfo {author} {\bibfnamefont {J.~M.}\ \bibnamefont {Martinis}},\ and\ \bibinfo {author} {\bibfnamefont {H.}~\bibnamefont {Neven}},\ }\bibfield  {title} {\bibinfo {title} {Characterizing quantum supremacy in near-term devices},\ }\href@noop {} {\bibfield  {journal} {\bibinfo  {journal} {Nature Physics}\ }\textbf {\bibinfo {volume} {14}},\ \bibinfo {pages} {595} (\bibinfo {year} {2018})}\BibitemShut {NoStop}%
\bibitem [{\citenamefont {Hauke}\ \emph {et~al.}(2015)\citenamefont {Hauke}, \citenamefont {Bonnes}, \citenamefont {Heyl},\ and\ \citenamefont {Lechner}}]{hauke2015probing}%
  \BibitemOpen
  \bibfield  {author} {\bibinfo {author} {\bibfnamefont {P.}~\bibnamefont {Hauke}}, \bibinfo {author} {\bibfnamefont {L.}~\bibnamefont {Bonnes}}, \bibinfo {author} {\bibfnamefont {M.}~\bibnamefont {Heyl}},\ and\ \bibinfo {author} {\bibfnamefont {W.}~\bibnamefont {Lechner}},\ }\bibfield  {title} {\bibinfo {title} {Probing entanglement in adiabatic quantum optimization with trapped ions},\ }\href@noop {} {\bibfield  {journal} {\bibinfo  {journal} {Frontiers in Physics}\ }\textbf {\bibinfo {volume} {3}},\ \bibinfo {pages} {21} (\bibinfo {year} {2015})}\BibitemShut {NoStop}%
\bibitem [{\citenamefont {Roland}\ and\ \citenamefont {Cerf}(2002)}]{roland2002quantum}%
  \BibitemOpen
  \bibfield  {author} {\bibinfo {author} {\bibfnamefont {J.}~\bibnamefont {Roland}}\ and\ \bibinfo {author} {\bibfnamefont {N.~J.}\ \bibnamefont {Cerf}},\ }\bibfield  {title} {\bibinfo {title} {Quantum search by local adiabatic evolution},\ }\href@noop {} {\bibfield  {journal} {\bibinfo  {journal} {Physical Review A}\ }\textbf {\bibinfo {volume} {65}},\ \bibinfo {pages} {042308} (\bibinfo {year} {2002})}\BibitemShut {NoStop}%
\bibitem [{\citenamefont {Yan}\ and\ \citenamefont {Sinitsyn}(2022)}]{yan2022analytical}%
  \BibitemOpen
  \bibfield  {author} {\bibinfo {author} {\bibfnamefont {B.}~\bibnamefont {Yan}}\ and\ \bibinfo {author} {\bibfnamefont {N.~A.}\ \bibnamefont {Sinitsyn}},\ }\bibfield  {title} {\bibinfo {title} {Analytical solution for nonadiabatic quantum annealing to arbitrary ising spin hamiltonian},\ }\href@noop {} {\bibfield  {journal} {\bibinfo  {journal} {Nature Communications}\ }\textbf {\bibinfo {volume} {13}},\ \bibinfo {pages} {2212} (\bibinfo {year} {2022})}\BibitemShut {NoStop}%
\bibitem [{\citenamefont {Mc~Keever}\ and\ \citenamefont {Lubasch}(2024)}]{mc2024towards}%
  \BibitemOpen
  \bibfield  {author} {\bibinfo {author} {\bibfnamefont {C.}~\bibnamefont {Mc~Keever}}\ and\ \bibinfo {author} {\bibfnamefont {M.}~\bibnamefont {Lubasch}},\ }\bibfield  {title} {\bibinfo {title} {Towards adiabatic quantum computing using compressed quantum circuits},\ }\href@noop {} {\bibfield  {journal} {\bibinfo  {journal} {PRX Quantum}\ }\textbf {\bibinfo {volume} {5}},\ \bibinfo {pages} {020362} (\bibinfo {year} {2024})}\BibitemShut {NoStop}%
\bibitem [{\citenamefont {Susa}\ \emph {et~al.}(2017)\citenamefont {Susa}, \citenamefont {Jadebeck},\ and\ \citenamefont {Nishimori}}]{susa2017relation}%
  \BibitemOpen
  \bibfield  {author} {\bibinfo {author} {\bibfnamefont {Y.}~\bibnamefont {Susa}}, \bibinfo {author} {\bibfnamefont {J.~F.}\ \bibnamefont {Jadebeck}},\ and\ \bibinfo {author} {\bibfnamefont {H.}~\bibnamefont {Nishimori}},\ }\bibfield  {title} {\bibinfo {title} {Relation between quantum fluctuations and the performance enhancement of quantum annealing in a nonstoquastic hamiltonian},\ }\href@noop {} {\bibfield  {journal} {\bibinfo  {journal} {Physical Review A}\ }\textbf {\bibinfo {volume} {95}},\ \bibinfo {pages} {042321} (\bibinfo {year} {2017})}\BibitemShut {NoStop}%
\bibitem [{\citenamefont {Albash}\ and\ \citenamefont {Lidar}(2018)}]{albash2018adiabatic}%
  \BibitemOpen
  \bibfield  {author} {\bibinfo {author} {\bibfnamefont {T.}~\bibnamefont {Albash}}\ and\ \bibinfo {author} {\bibfnamefont {D.~A.}\ \bibnamefont {Lidar}},\ }\bibfield  {title} {\bibinfo {title} {Adiabatic quantum computation},\ }\href@noop {} {\bibfield  {journal} {\bibinfo  {journal} {Reviews of Modern Physics}\ }\textbf {\bibinfo {volume} {90}},\ \bibinfo {pages} {015002} (\bibinfo {year} {2018})}\BibitemShut {NoStop}%
\bibitem [{\citenamefont {Rajak}\ \emph {et~al.}(2023)\citenamefont {Rajak}, \citenamefont {Suzuki}, \citenamefont {Dutta},\ and\ \citenamefont {Chakrabarti}}]{rajak2023quantum}%
  \BibitemOpen
  \bibfield  {author} {\bibinfo {author} {\bibfnamefont {A.}~\bibnamefont {Rajak}}, \bibinfo {author} {\bibfnamefont {S.}~\bibnamefont {Suzuki}}, \bibinfo {author} {\bibfnamefont {A.}~\bibnamefont {Dutta}},\ and\ \bibinfo {author} {\bibfnamefont {B.~K.}\ \bibnamefont {Chakrabarti}},\ }\bibfield  {title} {\bibinfo {title} {Quantum annealing: An overview},\ }\href@noop {} {\bibfield  {journal} {\bibinfo  {journal} {Philosophical Transactions of the Royal Society A}\ }\textbf {\bibinfo {volume} {381}},\ \bibinfo {pages} {20210417} (\bibinfo {year} {2023})}\BibitemShut {NoStop}%
\bibitem [{\citenamefont {Lloyd}(1996)}]{lloyd1996universal}%
  \BibitemOpen
  \bibfield  {author} {\bibinfo {author} {\bibfnamefont {S.}~\bibnamefont {Lloyd}},\ }\bibfield  {title} {\bibinfo {title} {Universal quantum simulators},\ }\href@noop {} {\bibfield  {journal} {\bibinfo  {journal} {Science}\ }\textbf {\bibinfo {volume} {273}},\ \bibinfo {pages} {1073} (\bibinfo {year} {1996})}\BibitemShut {NoStop}%
\bibitem [{\citenamefont {Heyl}\ \emph {et~al.}(2019)\citenamefont {Heyl}, \citenamefont {Hauke},\ and\ \citenamefont {Zoller}}]{heyl2019quantum}%
  \BibitemOpen
  \bibfield  {author} {\bibinfo {author} {\bibfnamefont {M.}~\bibnamefont {Heyl}}, \bibinfo {author} {\bibfnamefont {P.}~\bibnamefont {Hauke}},\ and\ \bibinfo {author} {\bibfnamefont {P.}~\bibnamefont {Zoller}},\ }\bibfield  {title} {\bibinfo {title} {Quantum localization bounds trotter errors in digital quantum simulation},\ }\href@noop {} {\bibfield  {journal} {\bibinfo  {journal} {Science advances}\ }\textbf {\bibinfo {volume} {5}},\ \bibinfo {pages} {eaau8342} (\bibinfo {year} {2019})}\BibitemShut {NoStop}%
\bibitem [{\citenamefont {Childs}\ \emph {et~al.}(2021)\citenamefont {Childs}, \citenamefont {Su}, \citenamefont {Tran}, \citenamefont {Wiebe},\ and\ \citenamefont {Zhu}}]{childs2021theory}%
  \BibitemOpen
  \bibfield  {author} {\bibinfo {author} {\bibfnamefont {A.~M.}\ \bibnamefont {Childs}}, \bibinfo {author} {\bibfnamefont {Y.}~\bibnamefont {Su}}, \bibinfo {author} {\bibfnamefont {M.~C.}\ \bibnamefont {Tran}}, \bibinfo {author} {\bibfnamefont {N.}~\bibnamefont {Wiebe}},\ and\ \bibinfo {author} {\bibfnamefont {S.}~\bibnamefont {Zhu}},\ }\bibfield  {title} {\bibinfo {title} {Theory of trotter error with commutator scaling},\ }\href@noop {} {\bibfield  {journal} {\bibinfo  {journal} {Physical Review X}\ }\textbf {\bibinfo {volume} {11}},\ \bibinfo {pages} {011020} (\bibinfo {year} {2021})}\BibitemShut {NoStop}%
\bibitem [{\citenamefont {Childs}\ \emph {et~al.}(2019)\citenamefont {Childs}, \citenamefont {Ostrander},\ and\ \citenamefont {Su}}]{childs2019faster}%
  \BibitemOpen
  \bibfield  {author} {\bibinfo {author} {\bibfnamefont {A.~M.}\ \bibnamefont {Childs}}, \bibinfo {author} {\bibfnamefont {A.}~\bibnamefont {Ostrander}},\ and\ \bibinfo {author} {\bibfnamefont {Y.}~\bibnamefont {Su}},\ }\bibfield  {title} {\bibinfo {title} {Faster quantum simulation by randomization},\ }\href@noop {} {\bibfield  {journal} {\bibinfo  {journal} {Quantum}\ }\textbf {\bibinfo {volume} {3}},\ \bibinfo {pages} {182} (\bibinfo {year} {2019})}\BibitemShut {NoStop}%
\bibitem [{\citenamefont {Campbell}(2019)}]{PhysRevLett.123.070503}%
  \BibitemOpen
  \bibfield  {author} {\bibinfo {author} {\bibfnamefont {E.}~\bibnamefont {Campbell}},\ }\bibfield  {title} {\bibinfo {title} {Random compiler for fast hamiltonian simulation},\ }\href {https://doi.org/10.1103/PhysRevLett.123.070503} {\bibfield  {journal} {\bibinfo  {journal} {Phys. Rev. Lett.}\ }\textbf {\bibinfo {volume} {123}},\ \bibinfo {pages} {070503} (\bibinfo {year} {2019})}\BibitemShut {NoStop}%
\end{thebibliography}%

\onecolumngrid
\appendix

\section{Recurrent equation for messages}
\label{appx:self_consistency_derivation}
In this section we derive Eq.~\eqref{eq:local_consistency}. First, we note that $m_{b\rightarrow a}$ is an environment tensor of a sub-tree of a tensor tree which is linked with the root $a$ by the edge $\{a, b\}$. Thus, it can be written as
\begin{eqnarray}
	\label{eq:full_message}
	m_{b\rightarrow a}[j_{ba}, j'_{ba}] = \sum_{\bfj_{\edges_{b\rightarrow a}}}\sum_{\bfj'_{\edges_{b\rightarrow a}}}\sum_{\bfi_{\nodes_{b\rightarrow a}}}\prod_{c\in \nodes_{b\rightarrow a}} T_c[i_c,\bfj_{\partial c}]T^*_c[i_c,\bfj'_{\partial c}],
\end{eqnarray}
where $\nodes_{b \rightarrow a}$ is the set of vertices of the sub-tree linked to $a$ by the edge $\{b, a\}$, $\bfj_{\edges_{b \rightarrow a}}$ is the set of bond indices of the same sub-tree, and $\bfi_{\nodes_{b \rightarrow a}}$ is the set of physical indices of the same sub-tree. Let us rewrite Eq.~\eqref{eq:full_message} using the recursive structure of the sub-tree
\begin{eqnarray}
	m_{b\rightarrow a}[j_{ba}, j'_{ba}] = \sum_{\bfj_{\edges_{b\rightarrow a}}}\sum_{\bfj'_{\edges_{b\rightarrow a}}}\sum_{\bfi_{\nodes_{b\rightarrow a}}}T_b\left[i_b, \bfj_{\partial b}\right]T_b^*\left[i_b, \bfj'_{\partial b}\right]\prod_{c\in \partial b \setminus a}\prod_{d\in \nodes_{c\rightarrow b}} T_d[i_d,\bfj_{\partial d}]T^*_d[i_d,\bfj'_{\partial d}],
\end{eqnarray}
where $\partial b \setminus a$ is the set of all vertices neighboring to $b$ except $a$. Now let us use the relation $xy + xz = x(y + z)$ that allows us to ``propagate'' some of the sums deeper in the equation
\begin{eqnarray}
	\label{eq:recursive_message}
	m_{b\rightarrow a}[j_{ba}, j'_{ba}] = \sum_{j_{\partial b\setminus a}}\sum_{i_b}T_b\left[i_b, \bfj_{\partial b}\right]T_b^*\left[i_b, \bfj'_{\partial b}\right]\prod_{c\in \partial b \setminus a}\left(\sum_{\bfj_{\edges_{c\rightarrow b}}}\sum_{\bfj'_{\edges_{c\rightarrow b}}}\sum_{\bfi_{\nodes_{c\rightarrow b}}}\prod_{d\in \nodes_{c\rightarrow b}} T_d[i_d,\bfj_{\partial d}]T^*_d[i_d,\bfj'_{\partial d}]\right),
\end{eqnarray}
where $\bfj_{\partial b\setminus a} = \{j_{cb} \mid c \in \partial b \setminus a\}$. Finally, let us note that the part of Eq.~\eqref{eq:recursive_message} that is in brackets is another message, this brings us to the final form of the recursive equation for messages
\begin{eqnarray}
	\label{eq:local_consistency_appx}
	m_{b\rightarrow a}[j_{ba},j'_{ba}] = \sum_{\bfj_{\partial b\setminus a}}\sum_{i_b} T_b\left[i_b, \bfj_{\partial b}\right]T_b^*\left[i_b, \bfj'_{\partial b}\right]\prod_{c\in \partial b\setminus a} m_{c\rightarrow b}[j_{cb},j'_{cb}].
\end{eqnarray}

\section{Belief propagation algorithm}
\label{appx:bp_alg}
Here we discuss Algorithm~\ref{alg:bp} known as BP algorithm in details, which is the fixed-point iteration method applied to Eq.~\eqref{eq:local_consistency}. For better stability we solve Eq.~\eqref{eq:local_consistency} up to normalization constants enforcing $\tr(m_{a\rightarrow b}) = 1$, which together with Hermiticity and positivity of $m_{a\rightarrow b}$ allows us to treat messages as density matrices. Since one can always recover the normalization constant of reduced density matrices that are computed from messages by enforcing unit trace, the norm of messages does not play an important role. Note that for a general graph, Eq.~\eqref{eq:local_consistency} could have multiple solutions; we assume that finding any solution is enough for us.

\begin{algorithm}[H]
	\caption{Belief propagation for graph tensor networks}\label{alg:bp}
	\begin{algorithmic}
		\Require Connectivity graph $G$, tensors $\{T_a\}_{a=1}^N$, bond dimensions $\{d_{ab}|\{a, b\}\in \edges\}$, accuracy threshold $\varepsilon$, maximal number of BP iterations $K$
		\Ensure Set of $2|\edges|$ messages $\{m_{a\rightarrow b}, \ m_{b\rightarrow a}\}_{\{a, b\}\in \edges}$ giving approximation of any tensor's environment
		\For{$\{a, b\}\in\edges$} \Comment{Loop initializing messages}
		\State $m_{a\rightarrow b}=\text{sample a random density matrix of size} \ d_{ab}\times d_{ab}$ \Comment{Computed as $\frac{AA^\dagger}{\tr\left(AA^\dagger\right)}$ where $A$ is random}
		\State $m_{b\rightarrow a}=\text{sample a random density matrix of size} \ d_{ab}\times d_{ab}$
		\EndFor
		\For{$i \in \{1,\dots,K\}$} \Comment{Loop running at most $K$ iterations of BP}
		\State ${\rm dist} = 0$ \Comment{Initialization of the aggregated distance between old and new messages}
		\For{$\{a, b\}\in\edges$} \Comment{Loop updating messages}
		\State $m^{({\rm new})}_{a\rightarrow b}[j_{ba},j'_{ba}] = \sum_{\bfj_{\partial a\setminus b}}\sum_{i_a} T_a\left[i_a, \bfj_{\partial a}\right]T_a^*\left[i_a, \bfj'_{\partial a}\right]\prod_{c\in \partial a\setminus b} m_{c\rightarrow a}[j_{ca},j'_{ca}]$ \Comment{Eq.~\eqref{eq:local_consistency}}
		\State $m^{({\rm new})}_{b\rightarrow a}[j_{ba},j'_{ba}] = \sum_{\bfj_{\partial b\setminus a}}\sum_{i_b} T_b\left[i_b, \bfj_{\partial b}\right]T_b^*\left[i_b, \bfj'_{\partial b}\right]\prod_{c\in \partial b\setminus a} m_{c\rightarrow b}[j_{cb},j'_{cb}]$
		\State $m^{({\rm new})}_{a\rightarrow b} \gets \frac{m^{({\rm new})}_{a\rightarrow b}}{\tr\left(m^{({\rm new})}_{a\rightarrow b}\right)}$ \Comment{Enforcing unit trace condition}
		\State $m^{({\rm new})}_{b\rightarrow a} \gets \frac{m^{({\rm new})}_{b\rightarrow a}}{\tr\left(m^{({\rm new})}_{b\rightarrow a}\right)}$
		\State ${\rm dist} \gets {\rm dist} + \left\|m^{({\rm new})}_{a\rightarrow b} - m_{a\rightarrow b}\right\|_1$ \Comment{Aggregating the distance between new and old messages}
		\State ${\rm dist} \gets {\rm dist} + \left\|m^{({\rm new})}_{b\rightarrow a} - m_{b\rightarrow a}\right\|_1$
		\State $m_{a\rightarrow b}\gets m^{({\rm new})}_{a\rightarrow b}$ \Comment{Replacing an old message with the new one}
		\State $m_{b\rightarrow a}\gets m^{({\rm new})}_{b\rightarrow a}$
		\EndFor
		\If{$\frac{\rm dist}{|\edges|} < \varepsilon$} \Comment{Exiting the BP loop if aggregated average distance is less than the accuracy threshold}
		\State \textbf{break}
		\EndIf
		\EndFor
	\end{algorithmic}
\end{algorithm}

Note that this algorithm has two hyperparameters: the maximal allowed discrepancy between messages from subsequent iterations (accuracy threshold) $\varepsilon$, which defines stopping criteria, and the maximal number of BP iterations $K$. Since the BP algorithm does not have convergence guarantees and sometimes it falls into an infinite cycle, one needs $K$ to prevent infinite loops during runtime. However, in practice, problems with convergence appear only when one uses the BP algorithm for measurements sampling. It happens due to the starting point of the BP algorithm in this case. Each measurement heavily breaks the Vidal gauge and to recover the Vidal gauge back, one needs to perform a lot of BP iterations, increasing the probability to fall into an infinite loop. See the inset of Fig.~\ref{fig:maxcut_benchmark} where we show how the Vidal distance, for which BP algorithm falls into an infinite loop, evolves with the number of measured qubits.

\section{Local orthogonality and its residual}
\label{appx:local_orthogonality}
To guarantee that Eq.~\eqref{eq:vidal_gauge_tn} is the Schmidt decomposition, one has to enforce the local orthogonality condition which holds for all directions $a\rightarrow b$:
\begin{eqnarray}
	\label{eq:local_orthogonality}
	\delta[j_{ba}, j'_{ba}] = \sum_{\bfj_{\partial a\setminus b}}\sum_{i_a}\Gamma_a[i_a, \bfj_{\partial a\setminus b}, j_{ba}]\Gamma^*_a[i_a, \bfj_{\partial a\setminus b}, j'_{ba}]\prod_{c\in \partial a\setminus b}\lambda^2_{ca}[j_{ca}],
\end{eqnarray}
where $\delta$ is the Kronecker symbol. For graphical representation see Fig.~\ref{fig:local_orthogonality}(a).
\begin{lemma}
	Eq.~\eqref{eq:vidal_gauge_tn} is the simultaneous Schmidt decomposition with respect to each single edge cut of a tree $G$ iff Eq.~\eqref{eq:local_orthogonality} holds.
\end{lemma}
\begin{proof}
We start by identifying a recurrent relation between Schmidt vectors. Let us consider a set of Schmidt vectors which corresponds to a particular sub-tree
\begin{eqnarray}
	\label{eq:schmidt_vecs}
	u_{a \rightarrow b}[\bfi_{\nodes_{a\rightarrow b}}, j_{ab}] = \sum_{\bfj_{\edges_{a\rightarrow b}}}\left(\prod_{c\in\nodes_{a\rightarrow b}}\Gamma_c[i_c,\bfj_{\partial c}]\right)\left(\prod_{\{d, e\}\in \edges_{a\rightarrow b}}\lambda_{de}[j_{de}]\right),
\end{eqnarray}
where $\nodes_{a \rightarrow b}$ is the set of vertices of the sub-tree linked to $b$ by the edge $\{a, b\}$, $\edges_{a\rightarrow b}$ is the set of edges of the same sub-tree, $\bfj_{\edges_{a \rightarrow b}}$ is the set of bond indices of the same sub-tree, $\bfi_{\nodes_{a \rightarrow b}}$ is the set of physical indices of the same sub-tree, $j_{ab}$ enumerates Schmidt vectors and $\bfi_{\nodes_{a\rightarrow b}}$ enumerates elements in a Schmidt vector. Let us rewrite Eq.~\eqref{eq:schmidt_vecs} using recursive structure of the tree
\begin{eqnarray}
	u_{a\rightarrow b}[\bfi_{\nodes_{a\rightarrow b}}, j_{ab}] = \sum_{\bfj_{\edges_{a\rightarrow b}}}\Gamma_a[i_a,\bfj_{\partial a}]\prod_{f\in \partial a\setminus b}\lambda_{fa}[j_{fa}]\left(\prod_{c\in\nodes_{f\rightarrow a}}\Gamma_c[i_c,\bfj_{\partial c}]\right)\left(\prod_{\{d, e\}\in \edges_{f\rightarrow a}}\lambda_{de}[j_{de}]\right).
\end{eqnarray}
Now let us use the relation $xy + xz = x(y + z)$ that allows us to ``propagate'' some of the sums deeper in the equation
\begin{eqnarray}
	u_{a\rightarrow b}[\bfi_{\nodes_{a\rightarrow b}}, j_{ab}] = \sum_{\bfj_{\partial a\setminus b}}\Gamma_a[i_a,\bfj_{\partial a}]\prod_{f\in \partial a\setminus b}\lambda_{fa}[j_{fa}]\left[\sum_{\bfj_{\edges_{f\rightarrow a}}}\left(\prod_{c\in\nodes_{f\rightarrow a}}\Gamma_c[i_c,\bfj_{\partial c}]\right)\left(\prod_{\{d, e\}\in \edges_{f\rightarrow a}}\lambda_{de}[j_{de}]\right)\right].
\end{eqnarray}
Now one can identify another Schmidt vector in square brackets and substitute it getting the following equation
\begin{eqnarray}
	u_{a\rightarrow b}[\bfi_{\nodes_{a\rightarrow b}}, j_{ab}] = \sum_{\bfj_{\partial a\setminus b}}\Gamma_a[i_a,\bfj_{\partial a}]\prod_{f\in \partial a\setminus b}\lambda_{fa}[j_{fa}]u_{f\rightarrow a}[\bfi_{\nodes_{f\rightarrow a}}, j_{fa}].
\end{eqnarray}
Similar recurrent relation can be obtained for the scalar product of Schmidt vectors
\begin{eqnarray}
	\label{eq:recurent_scalar_prod}
	&&G_{a\rightarrow b}[j_{ab}, j'_{ab}] = \sum_{\bfi_{\nodes_{a\rightarrow b}}}u_{a\rightarrow b}[\bfi_{\nodes_{a\rightarrow b}}, j_{ab}]u^*_{a\rightarrow b}[\bfi_{\nodes_{a\rightarrow b}}, j'_{ab}]\nonumber \\ &&= \sum_{\bfj_{\partial a\setminus b}}\sum_{\bfj'_{\partial a\setminus b}}\sum_{i_a}\Gamma_a[i_a,\bfj_{\partial a}]\Gamma^*_a[i_a,\bfj'_{\partial a}]\prod_{f\in \partial a\setminus b}\lambda_{fa}[j_{fa}]\lambda_{fa}[j'_{fa}]\sum_{\bfi_{\nodes_{f\rightarrow a}}}u_{f\rightarrow a}[\bfi_{\nodes_{f\rightarrow a}}, j_{fa}]u^*_{f\rightarrow a}[\bfi_{\nodes_{f\rightarrow a}}, j'_{fa}]\nonumber \\
	&&= \sum_{\bfj_{\partial a\setminus b}}\sum_{\bfj'_{\partial a\setminus b}}\sum_{i_a}\Gamma_a[i_a,\bfj_{\partial a}]\Gamma^*_a[i_a,\bfj'_{\partial a}]\prod_{f\in \partial a\setminus b}\lambda_{fa}[j_{fa}]\lambda_{fa}[j'_{fa}]G_{f\rightarrow a}[j_{fa}, j'_{fa}].
\end{eqnarray}
Let us now prove the ``only if'' statement that implies orthogonality of Schmidt vectors, i.e. $G_{a\rightarrow b}[j_{ab}, j'_{ab}] = \delta[j_{ab}, j'_{ab}]$ for all directions $a\rightarrow b$. By substituting this to the last line of Eq.~\eqref{eq:recurent_scalar_prod} we get
\begin{eqnarray}
	\delta[j_{ab}, j'_{ab}] &&= \sum_{\bfj_{\partial a\setminus b}}\sum_{\bfj'_{\partial a\setminus b}}\sum_{i_a}\Gamma_a[i_a,\bfj_{\partial a}]\Gamma^*_a[i_a,\bfj'_{\partial a}]\prod_{f\in \partial a\setminus b}\lambda_{fa}[j_{fa}]\lambda_{fa}[j'_{fa}]\delta[j_{fa}, j'_{fa}]\nonumber\\
	&&= \sum_{\bfj_{\partial a\setminus b}}\sum_{i_a}\Gamma_a[i_a,\bfj_{\partial a}, j_{ab}]\Gamma^*_a[i_a,\bfj'_{\partial a},j'_{ab}]\prod_{f\in \partial a\setminus b}\lambda^2_{fa}[j_{fa}],
\end{eqnarray}
which is Eq.~\eqref{eq:local_orthogonality}. To prove the ``if'' statement we use induction. If $G_{f\rightarrow a}[j_{fa}, j'_{fa}] = \delta[j_{fa}, j'_{fa}]$ for all $f\in\partial a\setminus b$ then $G_{a\rightarrow b}[j_{ab}, j'_{ab}] = \delta[j_{ab}, j'_{ab}]$ due to Eq.~\eqref{eq:local_orthogonality} and Eq.~\eqref{eq:recurent_scalar_prod}. This statement serves as the induction step. If $a$ is a leaf of a tree, one has $G_{a\rightarrow b}[j_{ab}, j'_{ab}] = \sum_{i_a}\Gamma_a[i_a, j_{ab}]\Gamma^*_a[i_a, j'_{ab}] =  \delta[j_{ab}, j'_{ab}]$ due to Eq.~\eqref{eq:local_orthogonality}. This is the induction base. Therefore, for any direction $a\rightarrow b$ we have $G_{a\rightarrow b}[j_{ab}, j'_{ab}] = \delta[j_{ab}, j'_{ab}]$, i.e. Schmidt vectors are orthogonal. Since all Schmidt coefficients in Eq.~\eqref{eq:vidal_gauge_tn} are non-negative by definition, Eq.~\eqref{eq:vidal_gauge_tn} is the valid Schmidt decomposition with respect to any single edge cut of $G$.
\end{proof}
For graphical interpretation of ``if'' (``only if'') part proof see Fig.~\ref{fig:local_orthogonality}(b) (Fig.~\ref{fig:local_orthogonality}(c).)

\begin{figure*}
	\centering
	\includegraphics[width=\linewidth]{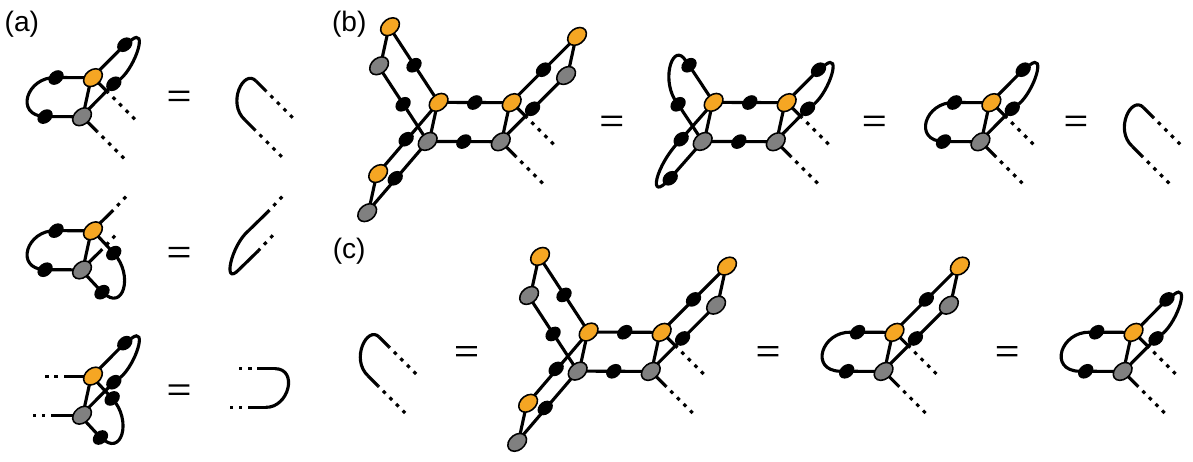}
	\caption{{\bf Diagrammatic interpreteation of the connection between local orthogonality condition and the global orthogonality of Schmidth vectors.} (a) Diagrammatic representation of the local orthogonality condition Eq.~\eqref{eq:local_orthogonality}, it holds for each node of a tensor network tree in the Vidal gauge. (b) Graphical derivation of the Schmidt vector orthogonality from the local orthogonality Eq.~\eqref{eq:local_orthogonality}. From left to right one sequentially applies the local orthogonality to the product of Schmidt vectors leading to the final Kronecker symbol proving orthogonality of Schmidt vectors. (c) Graphical derivation of the local orthogonality Eq.~\eqref{eq:local_orthogonality} from the Schmidt vectors orthogonality. From left to right one notice that the convolution of tree branches is the Kronecker symbol due to the Schmidt vectors orthogonality, then one applies the Schmidt vectors orthogonality condition to some of the branches ending up with the local orthogonality condition.}
	\label{fig:local_orthogonality}
\end{figure*}

Truncations in the Vidal gauge can corrupt its properties. To measure the level of degradation of the Vidal gauge one can introduce the residual of Eq.~\eqref{eq:local_orthogonality} as follows~\cite{tindall2023gauging}:
\begin{eqnarray}
	R = \frac{1}{2|\edges|}\sum_{a\in \nodes}\sum_{b\in\partial a}\Bigg\|\delta[j_{ba}, j'_{ba}] - \sum_{\bfj_{\partial a\setminus b}}\sum_{i_a}\Gamma_a[i_a, \bfj_{\partial a\setminus b}, j_{ba}]\Gamma^*_a[i_a, \bfj_{\partial a\setminus b}, j'_{ba}]\prod_{c\in \partial a\setminus b}\lambda^2_{ca}[j_{ca}]\Bigg\|_{1},
\end{eqnarray}
where $\|\cdot\|_1$ is the trace norm. One can view $R$ as the distance to the Vidal gauge.

\section{Algorithms for finding the Vidal gauge and truncation}
\label{appx:vidal_gauge}
Here we present Algorithm~\ref{alg:canonicalization} for Vidal gauge finding and Algorithm~\ref{alg:truncation} for the Vidal gauge truncation. The graphical interpretation of Algorithm~\ref{alg:canonicalization} is given in Fig.~\ref{fig:vidal_gauge}. Algorithm~\ref{alg:truncation} requires specifying the edge that is being truncated and a new bond dimension that is typically smaller than the initial one.
\begin{algorithm}[H]
	\caption{Vidal gauge finding}\label{alg:canonicalization}
	\begin{algorithmic}
		\Require Connectivity graph $G$, tensors $\{T_a\}_{a=1}^N$, bond dimensions $\{d_{ab}|\{a, b\}\in \edges\}$, messages $\{m_{a\rightarrow b}, m_{b\rightarrow a}\}_{\{a, b\}\in \edges}$
		\Ensure Vidal gauge of a tensor network that includes updated tensors $\{\Gamma_a\}_{a\in\nodes}$ and Schmidt vectors $\{\lambda_{ab}\}_{\{a, b\}\in \edges}$ assigned to each edge of the graph
		\For{$a\in \nodes$} \Comment{Initialization loop}
		\State $\Gamma_a\left[i_a, \bfj_{\partial a}\right] = T_a\left[i_a, \bfj_{\partial a}\right]$
		\EndFor
		\For{$\{a, b\} \in \edges$} \Comment{Loop over all edges finding all Schmidt coefficients and updating all tensors}
		\State $U[k, q], \lambda_{ab}[q], V^\dagger[q, l] = {\rm SVD}\left(\sum_{j=0}^{d_{ab} - 1}m^{\frac{1}{2}}_{a\rightarrow b}[j, k]m^{\frac{1}{2}}_{b\rightarrow a}[j, l]\right)$\Comment{Singular value decomposition}
		\State $\Gamma_a\left[i_a, \bfj_{\partial a}\right] \gets \sum_{j'_{ba}}\sum_{j''_{ba}}\Gamma_a\left[i_a, \bfj_{\partial a\setminus b}, j''_{ba}\right]m^{-\frac{1}{2}}_{a\rightarrow b}[j'_{ba}, j''_{ba}]U[j'_{ba}, j_{ba}]$
		\State $\Gamma_b\left[i_b, \bfj_{\partial b}\right] \gets \sum_{j'_{ab}}\sum_{j''_{ab}}\Gamma_b\left[i_b, \bfj_{\partial b\setminus a}, j''_{ab}\right]m^{-\frac{1}{2}}_{b\rightarrow a}[j'_{ab}, j''_{ab}]V^\dagger[j_{ab}, j'_{ab}]$
		\EndFor
	\end{algorithmic}
\end{algorithm}
\begin{figure}
	\centering
	\includegraphics[width=\linewidth]{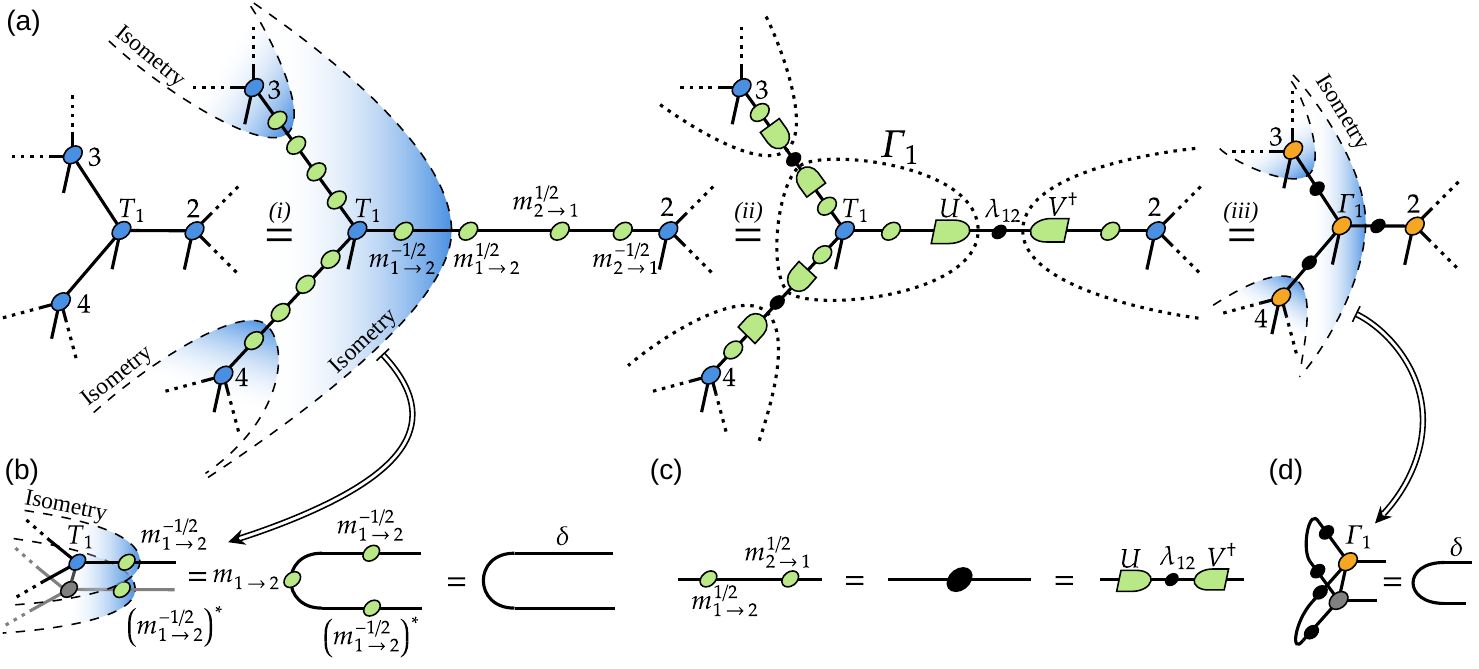}
	\caption{{\bf Diagrammatic interpretation of the algorithm for finding the Vidal gauge of a tree tensor network.} In panel (a) we show a sequence of equalities that leads to the Vidal gauge: (\emph{i}) one inserts four matrices on each edge, e.g. one inserts $m_{1\rightarrow 2}^{-\frac{1}{2}}$, $m_{1\rightarrow 2}^{\frac{1}{2}}$, $m_{2\rightarrow 1}^{\frac{1}{2}}$ and $m_{2\rightarrow 1}^{-\frac{1}{2}}$ on edge $\{1, 2\}$ keeping the tensor tree unchanged. We emphasize, that the purpose of insertion is to make tree branches isometric. This is necessary to make resulting Schmidt vectors orthogonal. In panel (b) we demonstrate that contraction of a branch with its complex conjugation leads to the Kronecker symbol proving the isometry property. (\emph{ii}) One contracts two matrices in the center of each edge into a single matrix and perform an SVD of this matrix as shown in panel (c). (\emph{iii}) One contracts tensors that are in the dotted areas getting the Vidal gauge. In panel (d) we emphasize that the isometry property of the branches, that is also seen as the orthogonality of Schmidt vectors, leads to the local orthogonality Eq.~\eqref{eq:local_orthogonality}, see App.~\ref{appx:local_orthogonality} for the formal prove and Fig.~\ref{fig:local_orthogonality}(c) for more precise visual explanation.}
	\label{fig:vidal_gauge}
\end{figure}
\begin{algorithm}[H]
	\caption{Truncation}\label{alg:truncation}
	\begin{algorithmic}
		\Require Connectivity graph $G$, tensors $\{\Gamma_a\}_{a=1}^N$, bond dimensions $\{d_{ab}|\{a, b\}\in \edges\}$, Schmidt vectors $\{\lambda_{ab}\}_{\{a, b\}\in\edges}$, edge that is being truncated $\{a, b\}$, new bond dimensions $\chi\leq d_{ab}$
		\Ensure Tensors $\{\Gamma_a\}_{a=1}^N$ and Schmidt vectors $\{\lambda_{ab}\}_{\{a, b\}\in\edges}$ with truncated edge $\{a, b\}$
		\State $\lambda_{ab}[j]\gets\sum_{j'}\delta_\chi[j, j']\lambda_{ab}\left[j'\right]$ \Comment{$\delta_{\chi}$ is the truncated Kronecker delta of size $\chi\times d_{ab}$, it removes the smallest Schmidt coefficients}
		\State $\Gamma_a\left[i, \bfj_{\partial a}\right] \gets \sum_{j'_{ba}} \delta_{\chi}[j_{ba}, j'_{ba}] \Gamma_a\left[i, \bfj_{\partial a\setminus b}, j'_{ba}\right]$
		\State $\Gamma_b\left[i, \bfj_{\partial b}\right] \gets \sum_{j'_{ab}} \delta_{\chi}[j_{ab}, j'_{ab}] \Gamma_b\left[i, \bfj_{\partial b\setminus a}, j'_{ab}\right]$
	\end{algorithmic}
\end{algorithm}

\section{Two-qubit gate application in the Vidal gauge}
\label{appx:simple_update}
In this Appendix we consider a two-qubit unitary gate application to neighboring qubits $a$ and $b$ in the Vidal gauge. First, we contract $\Gamma_a$, $\Gamma_b$, the two-qubit gate $W$, and neighboring Schmidt vectors into a single tensor $\Theta$ as follows
\begin{eqnarray}
	{\small \Theta[i_a, \bfj_{\partial a\setminus b}, i_b, \bfj_{\partial b\setminus a}] = \sum_{i'_b, i'_a, j_{ab}}W[i_a, i_b, i'_a, i'_b]\Gamma_a\left[i'_a, \bfj_{\partial a}\right]\Gamma_b\left[i'_b, \bfj_{\partial b}\right]\lambda_{ab}[j_{ab}]\left(\prod_{c\in\partial a\setminus b}\lambda_{ca}[j_{ca}]\right)\left(\prod_{c\in\partial b\setminus a}\lambda_{cb}[j_{cb}]\right)},
\end{eqnarray}
where the unitary matrix $W$ is viewed as a tensor of rank $4$ with two input indices and two output indices. Next, we perform an SVD of tensor $\Theta$ by splitting its indices into two groups and flattening those groups into two matrix indices:
\begin{eqnarray}
	U\left[i_a, \bfj_{\partial a\setminus b}, j\right], \lambda[j], V^\dagger[j, i_b, \bfj_{\partial b\setminus a}] = {\rm SVD}\left(\Theta\left[i_a, \bfj_{\partial a\setminus b}, i_b, \bfj_{\partial b\setminus a}\right]\right).
\end{eqnarray}
Finally, we define the updated versions of $\Gamma_a$, $\Gamma_b$, and $\lambda_{ab}$ as
\begin{eqnarray}
	&&\tilde{\Gamma}_a\left[i_a, \bfj_{\partial a}\right] = U\left[i_a, \bfj_{\partial a\setminus b}, j_{ba}\right]\prod_{c\in\partial a\setminus b}\lambda^{-1}_{ca}[j_{ca}],\nonumber\\
	&&\tilde{\Gamma}_b\left[i_b, \bfj_{\partial b}\right] = V^*\left[j_{ab}, i_b, \bfj_{\partial b\setminus a}\right]\prod_{c\in\partial b\setminus a}\lambda^{-1}_{cb}[j_{cb}],\nonumber\\
	&&\tilde{\lambda}_{ab}[j] = \lambda[j],
\end{eqnarray}
respectively. $\tilde{\lambda}_{ab}$ is the Schmidt vector by construction, while other edges remain intact. Thus the Vidal gauge is preserved.
\begin{figure*}
	\centering
	\includegraphics[width=\linewidth]{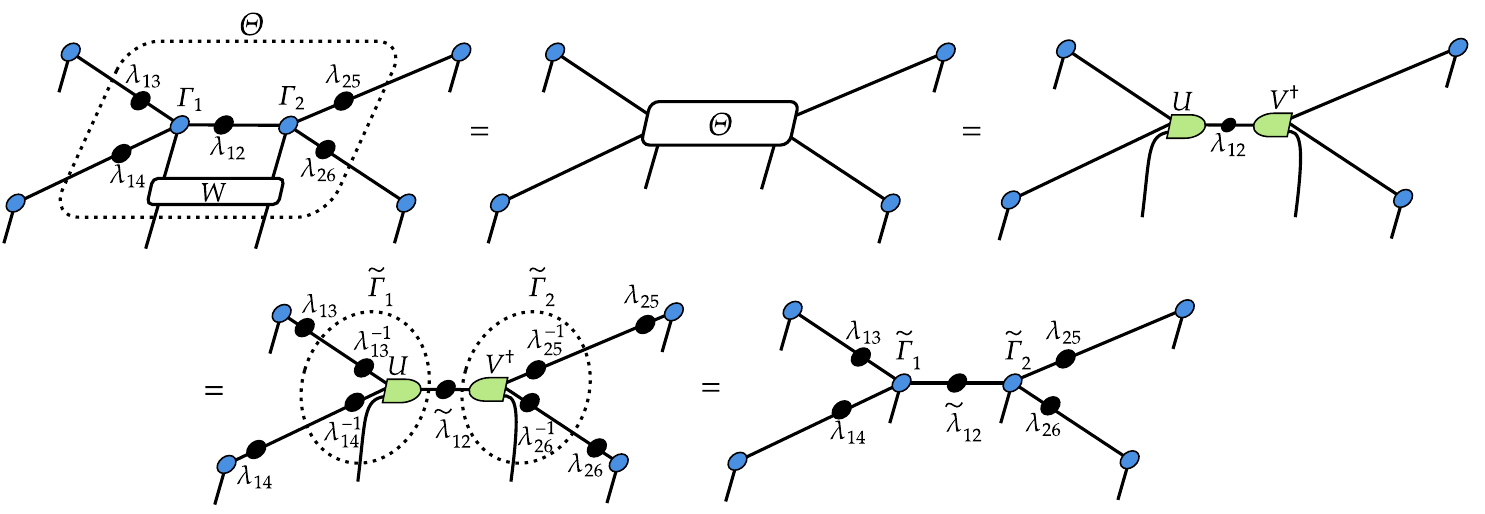}
	\caption{{\bf Diagrammatic representation of the algorithm for a two-qubit gate application to a state in the Vidal gauge.} From the beginning of the equality to its end: one contracts $\lambda_{13}$, $\lambda_{14}$, $\lambda_{12}$, $\lambda_{25}$, $\lambda_{26}$, $\Gamma_1$, $\Gamma_2$ and $W$ into a single tensor $\Theta$; one apply SVD to $\Theta$; one inserts identity matrices decomposed into the product of mutually inverse diagonal matrices $\lambda$ and $\lambda^{-1}$; one computes new vertex tensors $\tilde{\Gamma}_{1}$ and $\tilde{\Gamma}_{2}$ by contracting $\lambda^{-1}_{13}$, $\lambda^{-1}_{14}$ and $U$ into $\tilde{\Gamma}_{1}$ and $\lambda^{-1}_{25}$, $\lambda^{-1}_{26}$ and $V^\dagger$ into $\tilde{\Gamma}_{2}$. Note that the resulting tensor tree is in the Vidal gauge if it was in the Vidal gauge before the gate application.}
	\label{fig:simple_update}
\end{figure*}
Note that the bond dimension $d_{ab}$ is higher after the application of $W$. See Fig.~\ref{fig:simple_update} for schematics.

\section{QA dynamics trotterization}
\label{appx:annealing_trotterization}
For QA, one wants to solve the following Schr\"{o}dinger equation
\begin{eqnarray}
	\label{eq:schrodinger_eq}
	i\frac{\partial\ket{\Psi(t)}}{\partial t} = ((1 - s(t)) H_{\rm Ising} + s(t) H_{\rm mixing})\ket{\Psi(t)},
\end{eqnarray}
where $s(t) = 1 - \frac{t}{T}$ is the schedule function, $H_\text{I}$ is the initial Hamiltonian, $H_\text{F}$ is the target Hamiltonian, and $T$ is the total annealing time. 
We chose $H_{\rm mixing} = \sum_{a=1}^N X_a$ where we start in the maximal energy state $\ket{\Psi(0)} = \ket{+}^{\otimes N}$, and $H_{\rm Ising} = \sum_{\{a,b\} \in \edges}J_{ab} \ Z_aZ_b + \sum_{a=1}^N h_a Z_a$.

To simulate Eq.~\eqref{eq:schrodinger_eq}, we discretize its evolution operator $U(T)$ in time as follows
\begin{eqnarray}
	\label{eq:discretizatoion}
	U(T) \approx \prod_{k = 1}^{T / \delta t} U_X\Big(\delta t \cdot s(k\delta t)\Big) \ U_{Z}\Big(\delta t \cdot \left[1 - s(k \delta t)\right]\Big),
\end{eqnarray}
where $U_Z(t) = \prod_{\{a, b\}\in \edges} \exp\Bigg(-it\cdot\Bigg( J_{ab} \ Z_aZ_b + \frac{h_a}{D_a} \ Z_a + \frac{h_b}{D_b} \ Z_b\Bigg)\Bigg)$ is the interaction layer, $U_X(t) = \prod_{a = 1}^N \exp\left(-it\cdot X_a\right)$ is the mixing layer, $\delta t$ is the discretization time step, $k$ enumerates discrete time steps, and $D_a = |\partial a|$ is the degree of the $a$-th vertex in the graph. In all numerical experiments, we chose $\delta t = 0.2$. Note that the $U_Z$ layer factorizes into a product of two-qubit gates, and $U_X$ layer factorizes into a product of one-qubit gates. It allows one to apply simple update algorithm, BP algorithm and truncations to simulate the QA dynamics.

\section{Comparison of the exact entanglement entropy dynamics with the mean-field computation}
\label{appx:entropy_dynamics_cmp}

In this section we compare the entanglement entropy computed exactly and using the approximation of Eq.~\eqref{eq:apprx_entropy}.
For this purpose we generate a random $3$-regular graph, compute the bipartition Eq.~\eqref{eq:bipartition} and compute the dynamics of the entanglement entropy with respect to this bipartition exactly and using Eq.~\eqref{eq:apprx_entropy}. We define the error of the mean-field based entropy dynamics relative to the exact entropy as follows
\begin{eqnarray}
    \label{eq:entropy_error}
    {\rm err} = \frac{\|S^*_{\rm exact} - S^*_{\rm mean-field}\|_2}{\|S^*_{\rm exact}\|_2},
\end{eqnarray}
where $\|\cdot\|_2$ is the $2$-norm of a vector, and $S^*$ is considered as a vector indexed by discrete time. In Fig.~\ref{fig:entropy_validation}(a) we plot error Eq.~\eqref{eq:entropy_error} for $80$ random $3$-regular graphs in total ($20$ graphs per $N$), $N$ varying from $14$ to $20$ and $T = 20$. We find that the error varies from small to very high, but on median it is about $10\%$. The median error does not increase with $N$ and more likely the standard deviation of the error is narrowing with $N$ due to the self-averaging phenomenon. We also observed, that high error examples mostly correspond to the quench regime where long range correlations break the mean-field approximation Eq.~\eqref{eq:apprx_entropy}.

An example of the typical behavior of the approximate entropy compared to the exact one is plotted in Fig.~\ref{fig:entropy_validation}(b). It corresponds to the red dot in Fig.~\ref{fig:entropy_validation}(a) which has about median error. One can see that qualitatively the approximate entropy dynamics resembles well the exact one.
\begin{figure}
	\centering
	\includegraphics[width=1.\linewidth]{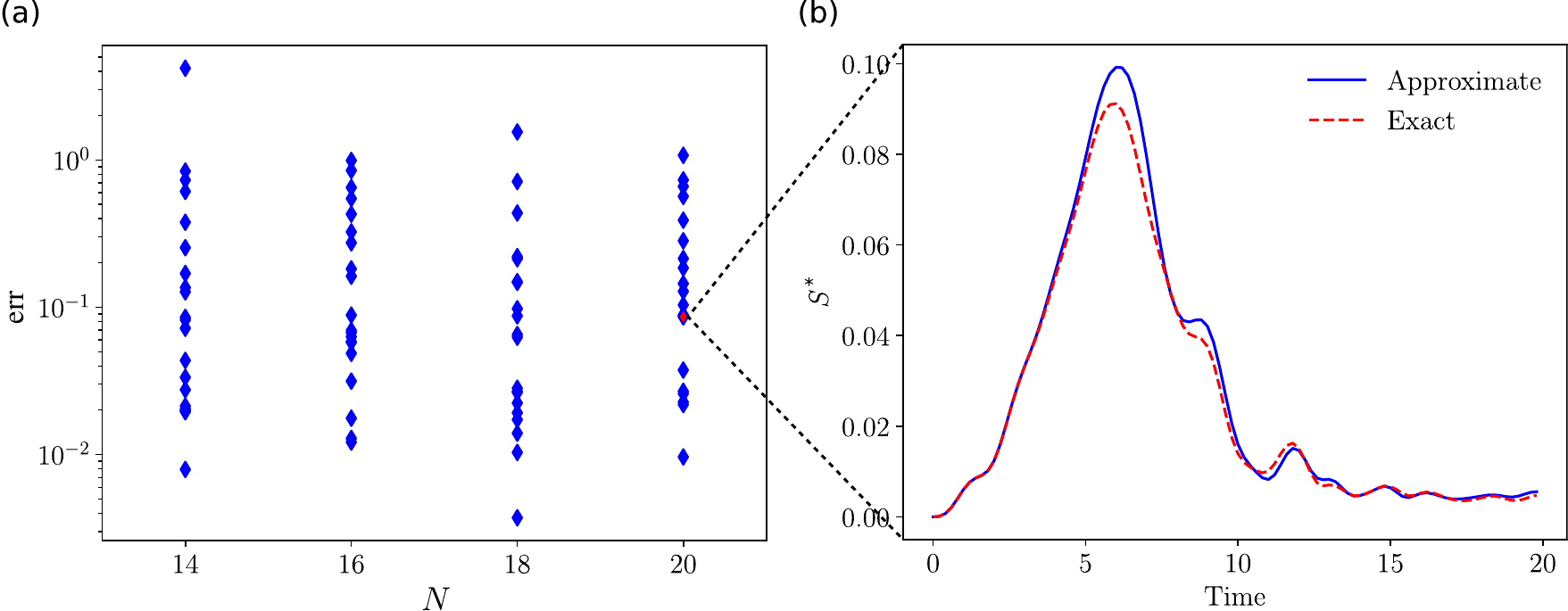}
	\caption{{\bf The comparison of the entanglement entropy dynamics for a small QUBO problem computed exactly and by the mean-field approximation Eq.~\eqref{eq:apprx_entropy}.} (a) The relative error Eq.~\eqref{eq:entropy_error} of the entropy computed approximately using the mean-field approximation Eq.~\eqref{eq:apprx_entropy} for different $N$, $20$ different random $3$-regular graphs per $N$ and $T = 20$. (b) The typical behavior of the entropy dynamics computed using the mean-field approximation Eq.~\eqref{eq:apprx_entropy} compared to the exact entropy dynamics. It corresponds to the red dot in the panel (a) which has approximately the median error of our studied instances. }
	\label{fig:entropy_validation}
\end{figure}

\end{document}